\documentclass[11pt]{article}
\usepackage{amsmath}
\usepackage{graphicx,psfrag,epsf}
\usepackage{enumerate}
\usepackage[english]{babel}
\usepackage{url} % not crucial - just used below for the URL 
\usepackage{natbib}
\usepackage{multirow}
\usepackage{rotating}
\usepackage{placeins}

\newcommand{\blind}{1}

\usepackage{amssymb, amsthm}

\newtheorem{prop}{Proposition}
\newtheorem{lemma}{Lemma}

\newcommand{\mynorm}[1]{{\left\vert\kern-0.25ex\left\vert\kern-0.25ex\left\vert #1      \right\vert\kern-0.25ex\right\vert\kern-0.25ex\right\vert}}  

\newcommand{\vecnorm}[1]{{\left\vert\kern-0.25ex\left\vert #1      \right\vert\kern-0.25ex\right\vert}}
\newcommand{\myscal}[2]{\langle \kern-0.25ex \langle #1, #2\rangle \kern-0.25ex \rangle}
\newcommand{\norm}[1]{{\left\vert\kern-0.25ex\left\vert #1      \right\vert\kern-0.25ex\right\vert_2}}   
\newcommand{\scal}[2]{{\langle #1, #2\rangle_2}}

\begin{document}

%%%%%%%%%%%%%%%%%%%%%%%%%%%%%%%%%%%%%%%%%%%%%%%%%%%%%%%%%%%%%%%%%%%%%%%%%%%%%%
\if1\blind
{
  \title{{\bf Multivariate Functional Principal Component Analysis for Data Observed on Different (Dimensional) Domains} }
  \author{Clara Happ  and
    Sonja Greven\thanks{  The authors acknowledge support from the German Research Foundation through Emmy Noether grant GR 3793/1-1.}\\
    Department of Statistics, LMU Munich\\[12pt]
    for the Alzheimer's Disease Neuroimaging Initiative.\thanks{Data used  in  preparation  of  this  article  were  obtained  from  the  Alzheimer’s  Disease Neuroimaging  Initiative  (ADNI)  database  (\protect\url{http://adni.loni.usc.edu}). As  such,  the  investigators  within the ADNI contributed to the design and implementation of ADNI and/or provided data but  did  not  participate  in  analysis  or  writing  of  this  paper.  A  complete  listing  of  ADNI  investigators can be found at: \protect\url{http://adni.loni.usc.edu/wp-content/uploads/how_to_apply/ADNI_Acknowledgement_List.pdf}. 
Data collection and sharing for the neuroimaging data in Section~5  was funded by the Alzheimer's Disease Neuroimaging Initiative (ADNI, National Institutes of Health Grant U01 AG024904) and DOD ADNI (Department of Defense award number W81XWH-12-2-0012). A detailed list of ADNI funding is available at \protect\url{http://adni.loni.usc.edu/about/funding/}. The grantee organization is the Northern California Institute for Research and Education, and the study is coordinated by the Alzheimer's Disease Cooperative Study at the University of California, San Diego. ADNI data are disseminated by the Laboratory for Neuro Imaging at the University of Southern California.
}
 }
  \maketitle
} \fi

%\bigskip
\begin{abstract}
Existing approaches for multivariate functional principal component analysis are restricted to data on the same one-dimensional interval. The presented approach focuses on multivariate functional data on different domains that may differ in dimension, e.g. functions and images. The theoretical basis for multivariate functional principal component analysis is given in terms of a Karhunen-Lo\`{e}ve Theorem. For the practically relevant case of a finite Karhunen-Lo\`{e}ve representation, a relationship between univariate and multivariate functional principal component analysis is established. This offers an estimation strategy to calculate multivariate functional principal components and scores based on their univariate counterparts. 
For the resulting estimators, asymptotic results are derived.
The approach can be extended to finite univariate expansions in general, not necessarily orthonormal bases. It is also applicable for sparse functional data or data with measurement error. 
A flexible \texttt{R} implementation is available on CRAN.
The new method is shown to be competitive to existing approaches for data observed on a common one-dimensional domain. The motivating application is a neuroimaging study, where the goal is to explore how longitudinal trajectories of a neuropsychological test score covary with FDG-PET brain scans at baseline.
Supplementary material, including detailed proofs, additional simulation results and software is available online.
\end{abstract}

\noindent%

\newpage

%%
%%
%%
%%%\doublespacing
%%
% Introduction
\section{Introduction}

Statistical methods for functional data have become increasingly important in recent years. Functional principal component analysis (FPCA) is one of the key techniques in functional data analysis, as it provides an easily interpretable exploratory analysis of the data. Further, it is an important building block for many statistical models \citep[see e.g.][]{RamsaySilverman:2005}. The technical progress in many fields of application allows the collection of more and more data with functional features, often several kinds per observation unit. This encourages the study of multivariate functional data and new methods are required to reveal e.g. joint variation in the different elements.

As a simple motivating example, consider the gait cycle data \citep{RamsaySilverman:2005} shown in Fig.~\ref{fig:motivationMFPCA}. It contains $39$ observations of hip and knee angle during a gait cycle on a standardized time interval. Both elements of this bivariate data can be described separately by their first three univariate eigenfunctions that explain $94.4 \%$ (hip) and $87.5 \%$ (knee) of the total variability in the data. The associated functional principal component scores, however, reveal that there is a non negligible correlation between almost all score pairs of the two elements. The separate FPCA thus captures joint variation between hip and knee angles only indirectly, which makes the interpretation of the FPCA results difficult. Correlated scores can also lead to multicollinearity issues in a subsequent regression analysis \citep[functional principal component regression, e.g.][]{MuellerStadtmueller:2005}.
Multivariate FPCA, by contrast, directly adresses potential covariation between the hip and knee elements. The first three bivariate principal components shown in Fig.~\ref{fig:motivationMFPCA}, which explain $85.3\%$ of the variability in the data, give insight into the main modes of joint variation in the overall gait movement. The corresponding scores do not only allow a more parsimonious representation of the data (one score value per bivariate principal component and per observation), but they are also uncorrelated by construction.
Finally, the multivariate functional principal components are more natural to represent multivariate functional data in the sense that they have the same structure as each observation.
The extension of FPCA to multivariate functional data is hence of high practical relevance.

\begin{figure}[t]
\centering
\begin{minipage}{0.6\textwidth}
\includegraphics[width = 0.32\textwidth]{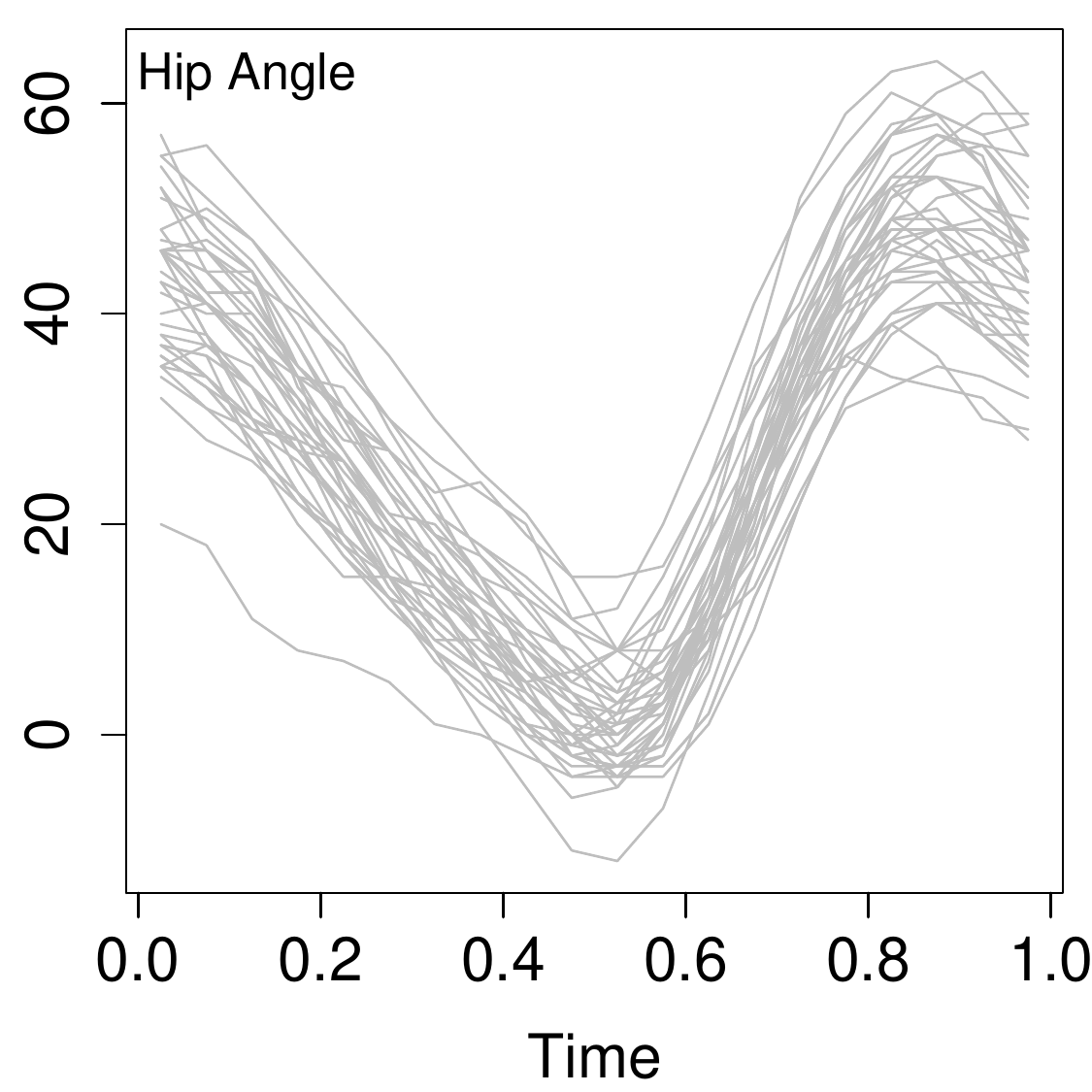} 
\hfill
\includegraphics[width = 0.32\textwidth]{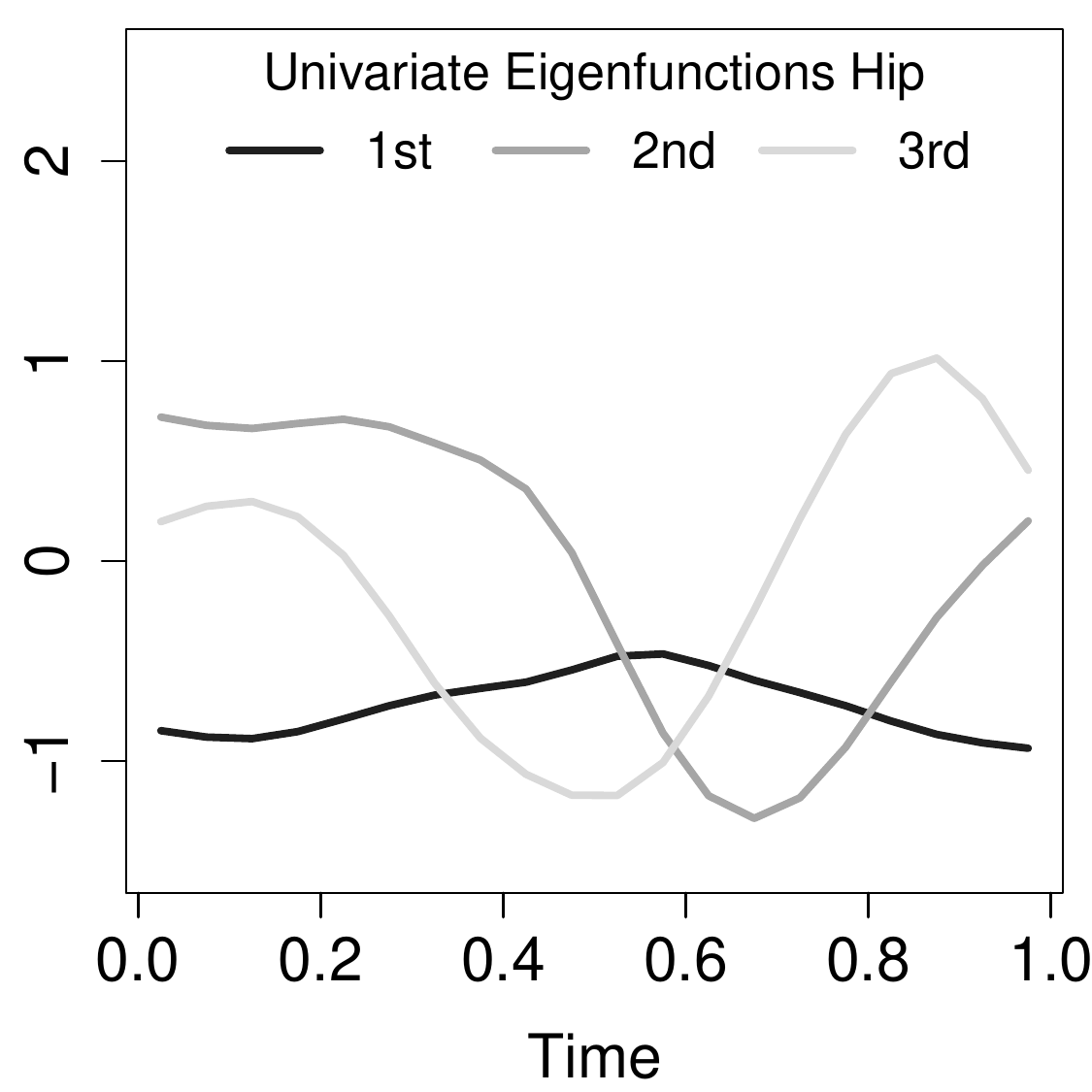} 
\hfill
\includegraphics[width = 0.32\textwidth]{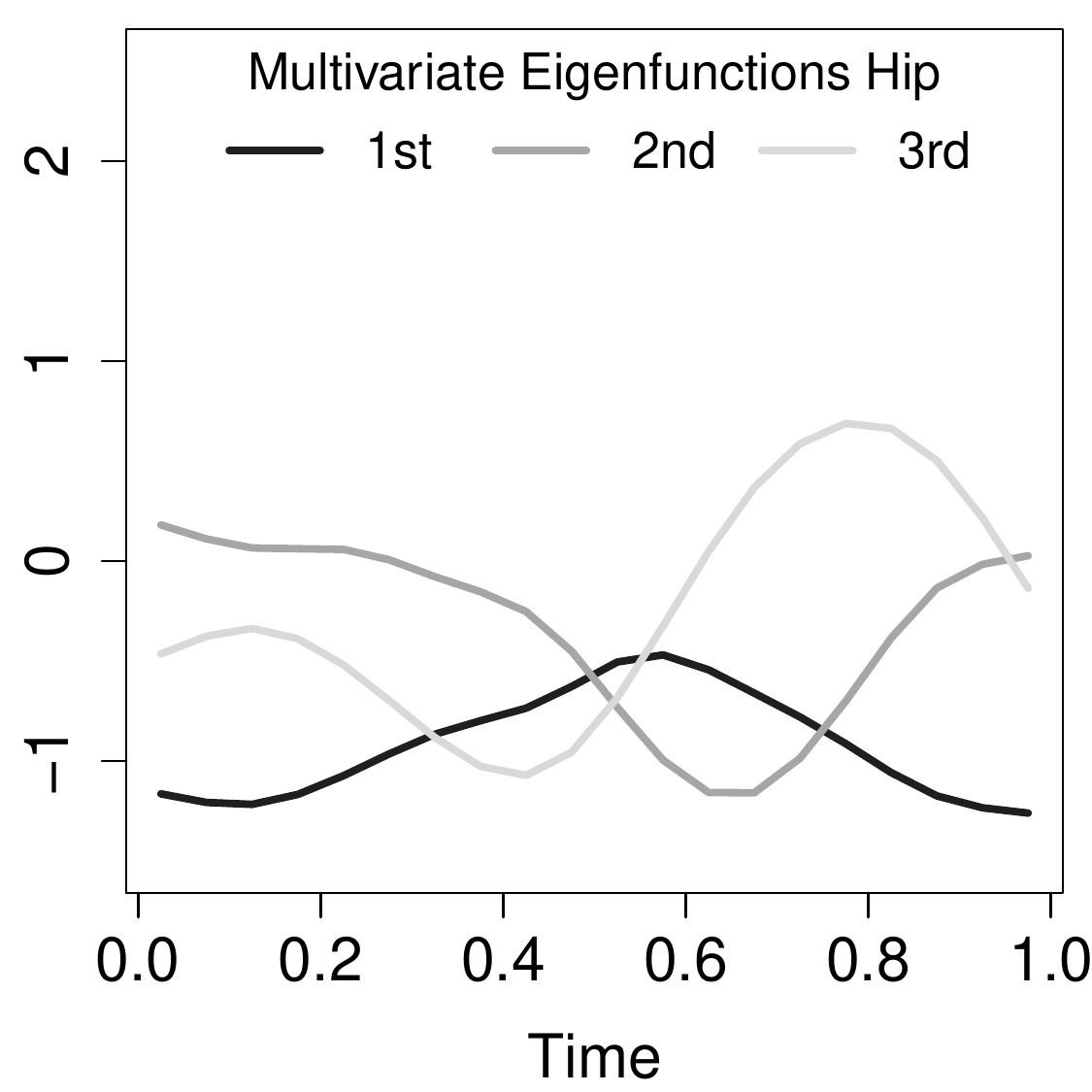}

\includegraphics[width = 0.32\textwidth]{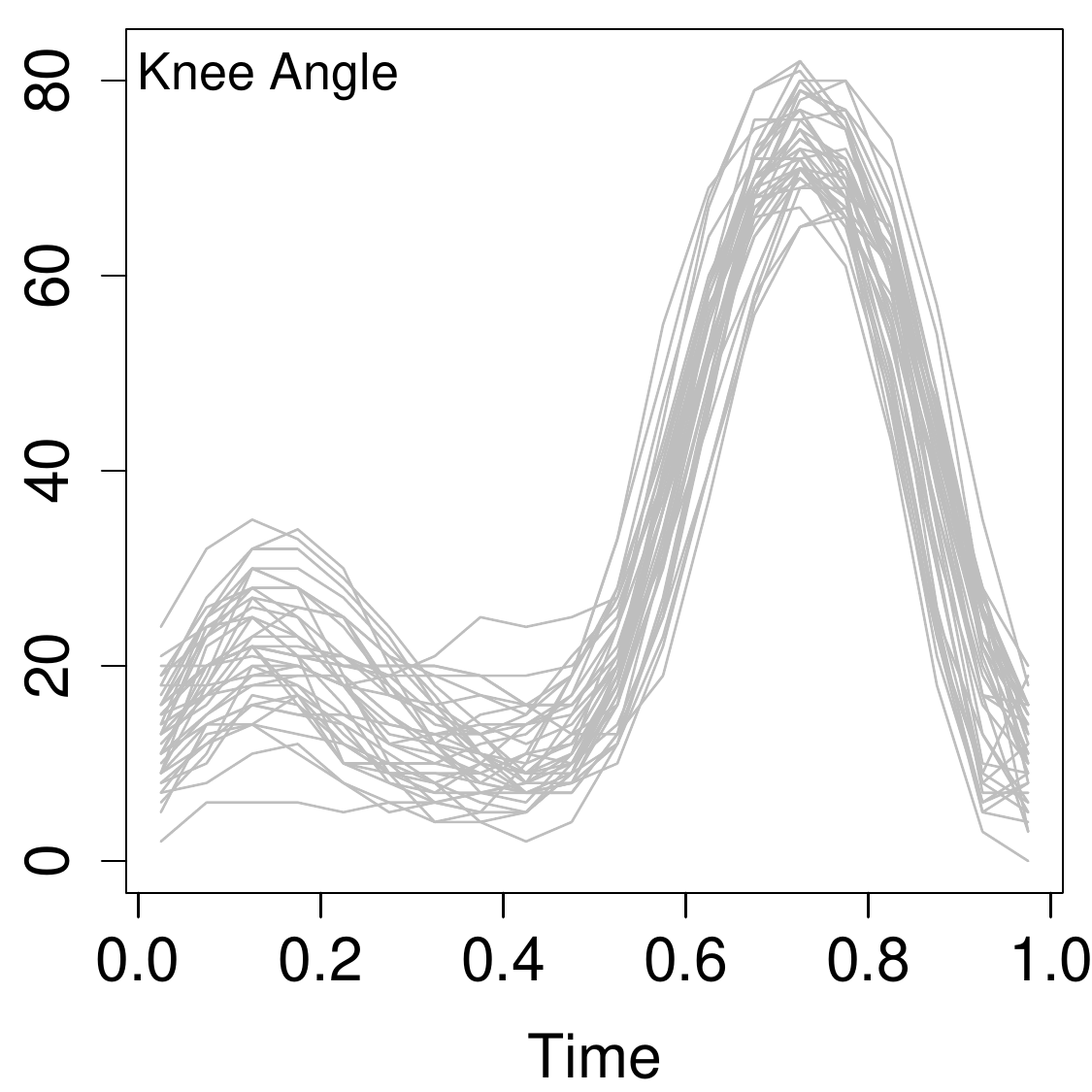} 
\hfill
\includegraphics[width = 0.32\textwidth]{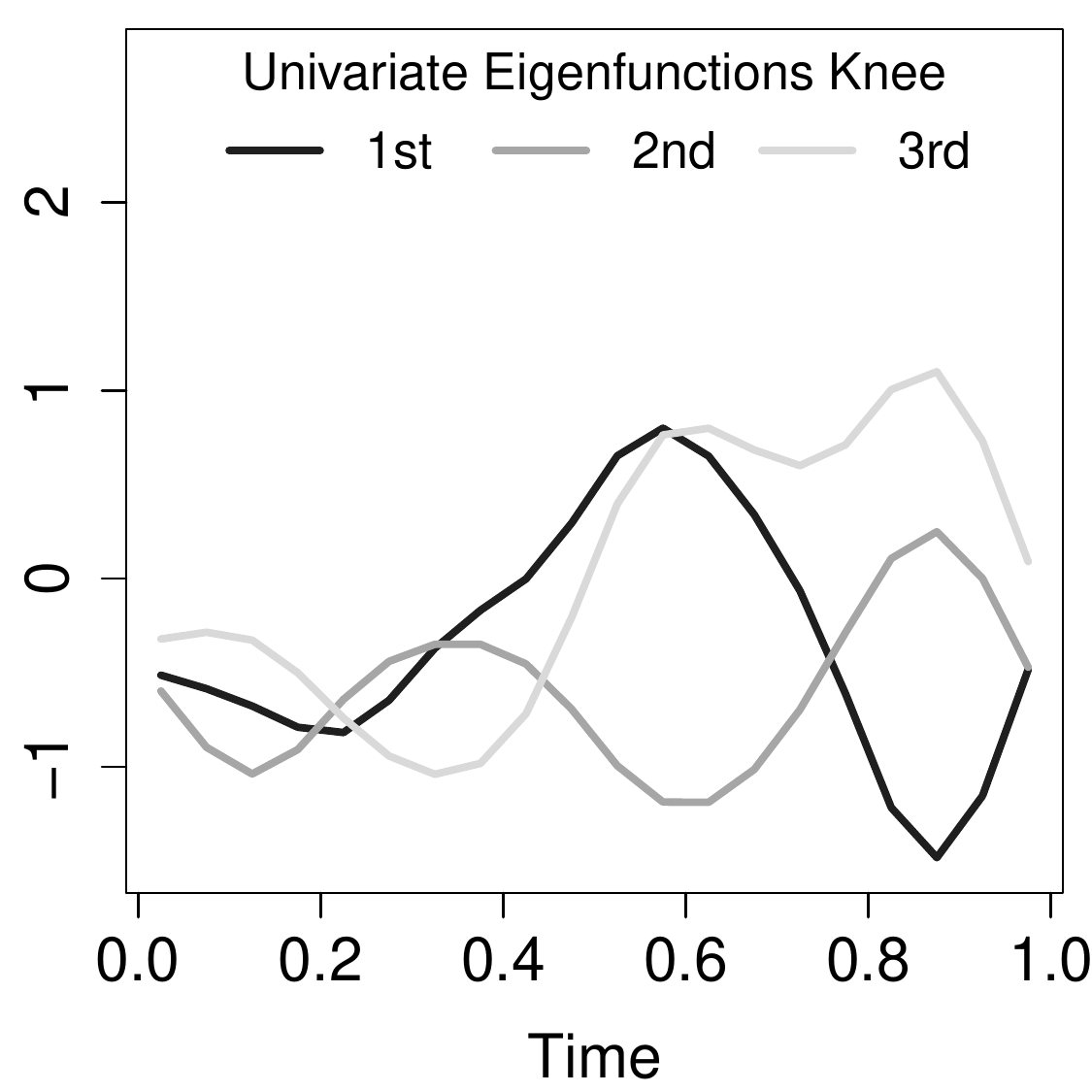}
\hfill
\includegraphics[width = 0.32\textwidth]{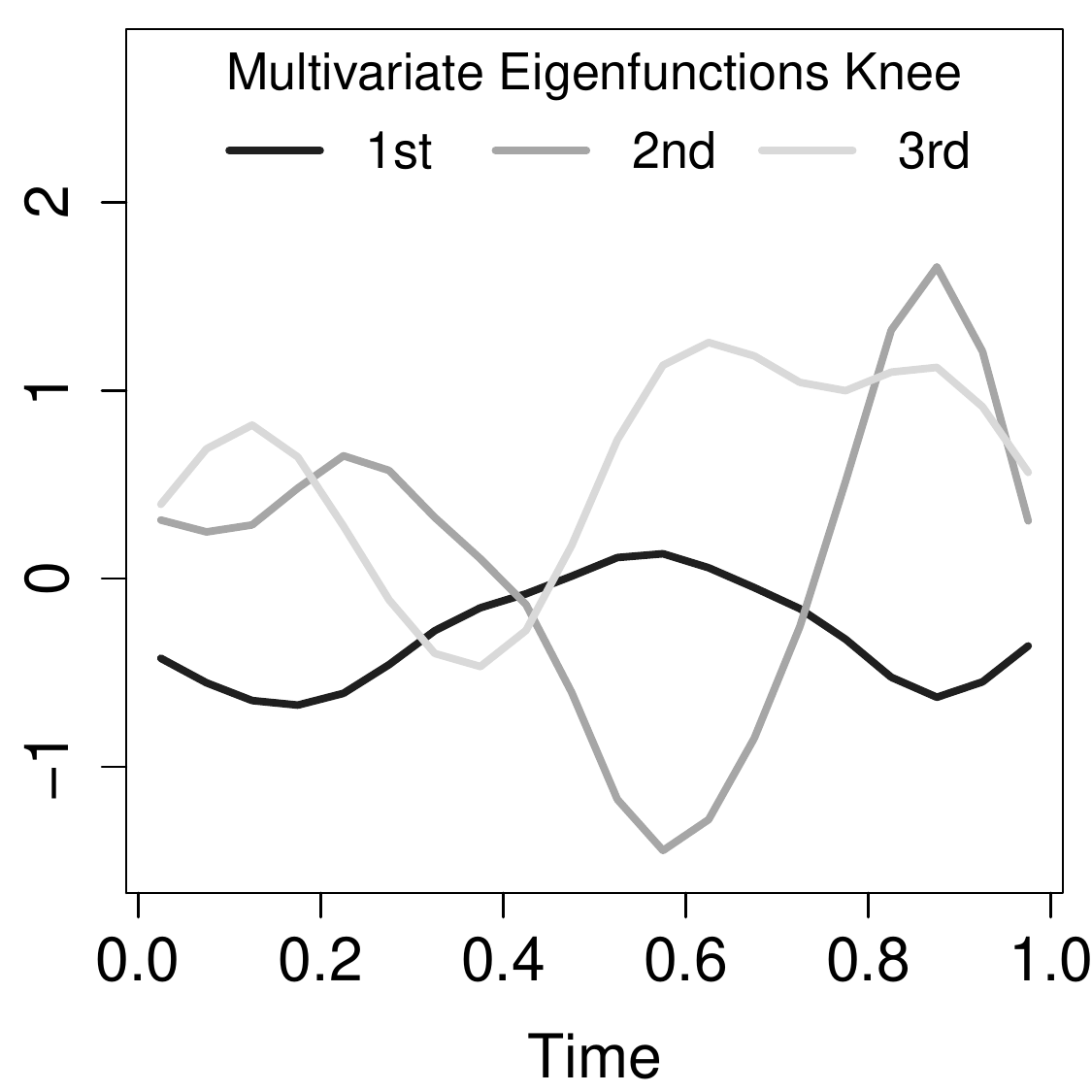}
\end{minipage}
\hfill
\begin{minipage}{0.35\textwidth}
\includegraphics[width = \textwidth]{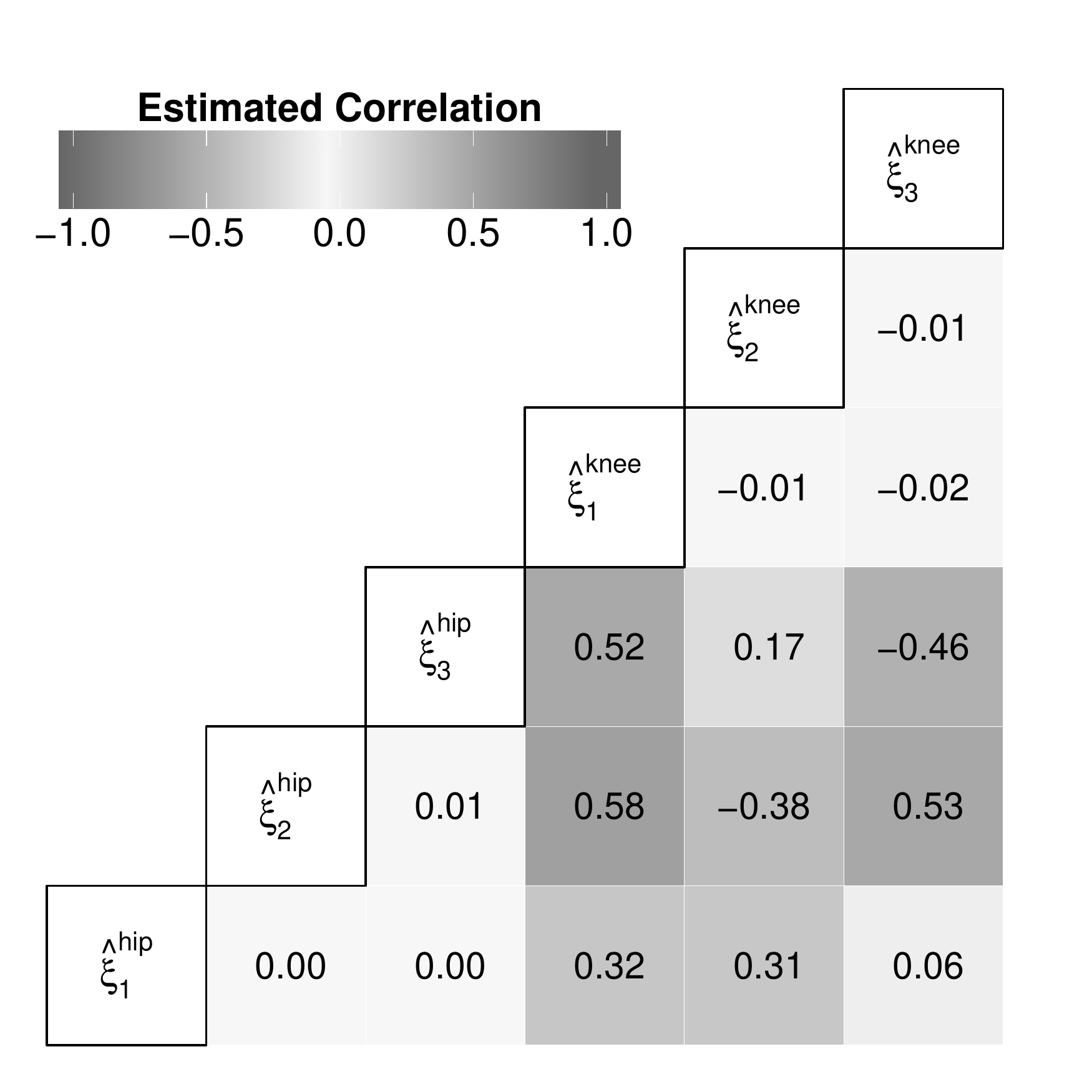} 
\end{minipage}

\caption{
\small Univariate and multivariate FPCA for the gait cycle data. 1st column: Original data. 2nd column: Results for univariate FPCA, calculated separately. The functions have been reflected, if necessary, and rescaled to have the same norm as the multivariate eigenfunctions  for comparison purposes. 3rd column: Results for multivariate FPCA, calculated with the new approach. 4th column: Empirical correlation  of the univariate FPCA scores for hip and knee.}
\label{fig:motivationMFPCA}
\end{figure}

Existing approaches for multivariate functional principal component analysis (MFPCA) are restricted to functions observed on the same finite, one-dimensional interval \citep{RamsaySilverman:2005, JacquesPreda:2014, ChiouEtAl:2014, Berrendero:2011}. Except for \cite{Berrendero:2011}, they all aim at a multivariate functional Karhunen-Lo\`{e}ve representation of the data.  For data measured e.g. in different units, \cite{JacquesPreda:2014} and \cite{ChiouEtAl:2014} also discuss normalized versions of MFPCA based on a normalized covariance operator.

The key motivation for this paper is that in practical applications,  multivariate functional data are neither  restricted to lie on the same interval nor to have one-dimensional domains, e.g. data that consists  of  functions and images, as in our neuroimaging application.
We start by extending the notion of multivariate functional data to the case of different (dimensional)  domains for the different elements. Next, the theoretical foundations of MFPCA are provided in terms of a Karhunen-Lo\`{e}ve  Theorem. 
For the practically relevant case of a finite or truncated Karhunen-Lo\`{e}ve representation, we establish a direct theoretical relationship between univariate and multivariate FPCA. This suggests a simple estimation strategy for multivariate functional principal components and scores based on their univariate counterparts.
For data on higher dimensional domains (tensor data, e.g. images), principal component methods have originally been developed in the context of psychometrics \citep[e.g.][]{Tucker:1966, CarrollChang:1970} and have become particularly important in the machine learning literature \citep{CoppiBolasco:1989, LuEtAl:2013}. Recent approaches for functional or smooth principal component analysis for tensor data have been proposed e.g.~in \cite{Allen:2013}. All these methods can be used as univariate building blocks for MFPCA.
The resulting estimators for MFPCA are shown to be consistent under a given set of assumptions.
In contrast to most of the existing methods for MFPCA, our new approach can be applied to sparse functional data and data with measurement error. It can be generalized to data available in arbitrary basis expansions and hence includes the MFPCA procedure proposed by \citet{JacquesPreda:2014} as a special case. 
The new method further allows to incorporate weights for the elements, if they differ in domain, range or variation.

The paper is organized as follows. Section~\ref{sec:Theory} introduces multivariate functional data and gives the theoretical basis for MFPCA. In Section~\ref{sec:MFPCA} we derive the estimation algorithm for MFPCA based on univariate basis expansions and investigate asymptotic properties of the resulting estimators. The performance of the new method is evaluated in Section~\ref{sec:Simulation} in a simulation with different levels of complexity.
 Section~\ref{sec:applicateADNI} contains the analysis of the motivating neuroimaging dataset.
The paper concludes with a discussion and an outlook in Section~\ref{sec:Discussion}.
 Supplementary material, containing detailed proofs of all  propositions, more simulation results and \texttt{R} code  is available online.

\newpage
% Theory
\section{Theoretical Foundations of Multivariate Functional Data}
\label{sec:Theory}

\subsection{Data Structure and Notation}

This paper is concerned with \textit{multivariate functional data}, i.e. 
each observation consists of $p \geq 2$ functions $X^{(1)} , \ldots ,  X^{(p)}$. They may be defined on different domains $\mathcal{T}_1 , \ldots ,  \mathcal{T}_p$ with possibly different dimensions. Technically, $\mathcal{T}_j$ must be compact sets in $\mathbb R ^{d_j},~  d_j \in \mathbb N $ with finite (Lebesgue-) measure and each element $X^{(j)} \colon \mathcal{T}_j \to \mathbb R $ is assumed to be in $L^2(\mathcal{T}_j)$.

In analogy to other approaches for multivariate functional data, the different functions are combined in a vector $X$ with 
\[X(\boldsymbol{t}) = \left(X^{(1)}(t_1), \ldots , X^{(p)}(t_p) \right) \in \mathbb{R}^p.\]
Note that $\boldsymbol{t} := (t_1, \ldots, t_p) \in \mathcal{T} := \mathcal{T}_1 \times \cdots  \times \mathcal{T}_p$ is a $p$-tuple of $d_1, \ldots, d_p$-dimensional vectors and not a scalar. This is a main difference to earlier approaches, as it allows each element $X^{(j)}$ to have a different argument $t_j$, even in the case of a common one-dimensional domain. In the following, it will be further assumed that 
\[ \mu(\boldsymbol{t}) := \mathbb{E} \left(X(\boldsymbol{t}) \right) =
\left( 
\mathbb{E} \left(X^{(1)}(t_1) \right)
, \ldots , 
\mathbb{E} \left( X^{(p)}(t_p) \right) 
\right) = \boldsymbol{0} \quad \forall~\boldsymbol{t} \in \mathcal{T}.
\]
For $\boldsymbol{s},\boldsymbol{t} \in \mathcal{T}$, define the matrix of covariances $C(\boldsymbol{s},\boldsymbol{t}) := \mathbb{E} \left(X(\boldsymbol{s})\otimes X(\boldsymbol{t})\right)$ with elements
\begin{equation} \label{eq:CovSymm}
C_{ij} (s_i, t_j) :=  \mathbb{E} \left(  X^{(i)}(s_i)X^{(j)}(t_j) \right) = \operatorname{Cov}(X^{(i)}(s_i), X^{(j)}(t_j) ),\quad s_i \in \mathcal{T}_i,~t_j \in \mathcal{T}_j .
\end{equation}

As noted in \citet[Chapter 8.5.]{RamsaySilverman:2005}, a suitable inner product is the basis of all approaches for principal component analysis. For functions $f = (f^{(1)} , \ldots ,  f^{(p)}) $ with elements $f^{(j)} \in L^2(\mathcal{T}_j)$  define the space
$
\mathcal{H} := L^ 2(\mathcal{T}_1) \times \ldots \times L^ 2(\mathcal{T}_p)
$ and 
\begin{equation}
\myscal{f}{g} := \sum \nolimits_{j=1}^p \scal{f^{(j)}}{g^{(j)}} = \sum \nolimits_{j=1}^p  \int_{\mathcal{T}_j} f^{(j)}(t_j) g^{(j)}(t_j) \mathrm{d}  t_j,\qquad f,g \in \mathcal{H}.
\label{eq:scalarProduct}
\end{equation}

\begin{prop} \label{prop:HisHilbert}
$\mathcal{H}$ is a Hilbert space with respect to the scalar product $\myscal{\cdot}{\cdot}$. 
\end{prop}

Proofs for all propositions are  given in the online appendix. The norm induced by $\myscal{\cdot}{\cdot}$ is denoted by $\mynorm{\cdot}$\footnote{The $L^2$-norm induced by $\scal{\cdot}{\cdot}$ on each $L^2(\mathcal{T}_j)$ is denoted by $\norm{\cdot}$. Further, $\vecnorm{\cdot}$ is the Euclidean norm for vectors and $\vecnorm{\cdot}_{\mathcal{T}}$ denotes a norm on $\mathcal{T}$ with $\vecnorm{\boldsymbol{t}}_{\mathcal{T}}^2 = \sum \nolimits _{j=1}^p \vecnorm{t_j}^2$ for $t_j \in \mathcal{T}_j \subset \mathbb{R}^{d_j},~ j = 1 , \ldots, p$.}.  Next, define the covariance operator $\Gamma \colon \mathcal{H}  \to   \mathcal{H}$
with the $j$-th element of $\Gamma f,~f \in \mathcal{H}$ given by
\begin{equation}
(\Gamma f)^{(j)}(t_j)
:= \sum\nolimits_{i=1}^p \int_{\mathcal{T}_i} C_{ij}(s_i,t_j) f^{(i)}(s_i) \mathrm{d}  s_i
=  \myscal{C_{\cdot j}(\cdot, t_j)}{f},\quad t_j \in \mathcal{T}_j.
\label{eq:GammaFormula}
\end{equation}
The setting can be generalized to a weighted scalar product on $\mathcal{H}$, i.e.
\begin{equation}
\myscal{f}{g}_w := \sum\nolimits_{j = 1}^p w_j \scal{f^{(j)}}{g^{(j)}},\qquad f,g \in \mathcal{H}
\label{eq:weightedSP}
\end{equation}
for some positive weights $w_1 , \ldots ,  w_p$, cf. \citet[Chapter 10.3. in the context of  hybrid data]{RamsaySilverman:2005}  or \citet{ChiouEtAl:2014}. The associated weighted covariance operator $\Gamma_w$ is given by its elements $(\Gamma_w f)^{(j)}$ with $f \in \mathcal{H}$  and
\[(\Gamma_w f )^{(j)}(t_j) = \myscal{C_{\cdot, j}( \cdot, t_j)}{f}_w , \quad t_j \in \mathcal{T}_j.
\]
The use of weights may be necessary if the elements have quite different domains or ranges or if they exhibit different amounts of variation,  in order to obtain multivariate functional principal components that have a meaningful interpretation \citep{ChiouEtAl:2014}. 
A weighted scalar product corresponds to a (global) rescaling of the elements by $w_j^{1/2}$. An alternative approach would be pointwise rescaling, e.g. by the inverse of the square root of the pointwise variance $C_{jj}(t_j, t_j)$. This can be seen as normalizing the covariance operator \citep{ChiouEtAl:2014, JacquesPreda:2014}. However, this second approach does not consider the size of the different domains $\mathcal{T}_j$ and would give equal variation per observation point $t_j$ rather than per element $j$. Moreover, rescaling with the pointwise variance would downweight areas in $\mathcal{T}_j$ with stronger variation, hence areas that might contribute relevant information to the functional principal components. Therefore, only global rescaling by means of a weighted scalar product is considered in the following. The weights have to be chosen prior to the analysis. They can be specified based on expert knowledge or estimated from the data, e.g. based on the variation in each element \citep[see references in][]{ChiouEtAl:2014}. A sensible choice will always depend on the specific application and the question of interest. One possible solution that is analogous to standardization in multivariate PCA is proposed in the application in Section \ref{sec:applicateADNI}. For the sake of better readability, all following theoretical results are derived for $w_1 = \ldots = w_p = 1$, but remain valid in the more general case of different weights. For the estimation algorithm discussed in Section~\ref{sec:estMFPCA}, MFPCA based on the weighted scalar product is addressed again.

\subsection{A Karhunen-Lo\`{e}ve Theorem for Multivariate Functional Data}

In the following it is shown that under mild conditions, $\Gamma$ has the same properties as the covariance operator in the univariate case and therefore a Karhunen-Lo\`{e}ve representation for multivariate functional data exists. The main difference to existing approaches for data with elements observed on the same (one-dimensional) domain is that in this special case, $\Gamma$ is an integral operator  with positive definite kernel $C(s,t)$. This directly gives all of the desired properties \citep{Saporta:1981}. In the more general case of elements observed on different domains, this is not obviously the case and the properties are shown explicitly.

\begin{prop} \label{prop:GammaProperties} \label{prop:GammCompact}
The covariance operator $\Gamma$ defined in \eqref{eq:GammaFormula} is a linear,  self-adjoint and positive operator.
If further  for all $i,j = 1 , \ldots ,  p$ there exist  $K_{ij} < \infty$ with
\begin{equation}
 \norm{C_{ij}(\cdot, t_j)}^2 = \int_{\mathcal{T}_i} C_{ij}(s_i,t_j)^ 2 \mathrm{d}  s_i  < K_{ij} \quad \forall~  t_j \in \mathcal{T}_j,
 \label{eq:Kijconstants}
 \end{equation}
%and $C_{ij}(s_i, \cdot),~s_i \in \mathcal{T}_i$, is continuous,
and $C_{ij}$ is uniformly continuous in the sense that
\[\forall~ \varepsilon > 0~ \exists~ \delta_{ij} > 0: \quad \vecnorm{t_j - t_j^\ast} < \delta_{ij} \quad \Rightarrow \quad \left| C_{ij}(s_i, t_j) - C_{ij}(s_i, t_j^\ast) \right| < \varepsilon \quad \forall ~  s_i \in \mathcal{T}_i,\]
then $\Gamma$ is a compact operator.
\end{prop}

In the remainder of this paper, it is assumed that the $C_{ij}$ satisfy all conditions of Prop.~\ref{prop:GammaProperties} and hence $\Gamma$ can always be assumed to be a compact positive operator on $\mathcal{H}$. By the Hilbert-Schmidt Theorem \citep[e.g.][Thm. VI.16]{ReedSimon:2005} it follows that there exists a complete orthonormal basis of eigenfunctions $\psi_m \in \mathcal{H}, ~ m \in \mathbb N$ of $\Gamma$ such that
\[\Gamma \psi_m = \nu_m \psi_m \quad \text{and} \quad \nu_m \to 0 \quad  \text{for} ~ m \to \infty.\]
In particular, since $\Gamma$ is a positive operator, it may be assumed  w.l.o.g.  that $\nu_1 \geq \nu_2 \geq \ldots \geq 0$. Since $\psi_m, ~m \in \mathbb N $ is an orthonormal basis of $\mathcal{H}$ and $\Gamma$ is self-adjoint, by the Spectral Theorem \citep[e.g.][Thm. VI.3.2.]{Werner:2011} it holds that
\[ \Gamma f = \sum \nolimits_{m=1}^\infty \nu_m \myscal{f}{\psi_m} \psi_m \quad \forall~ f \in \mathcal{H}.\]
The following proposition is a multivariate version of Mercer's Theorem \citep{Mercer:1909}. It plays a key role in the proof of the Karhunen-Lo\`{e}ve Theorem (Prop.~\ref{prop:KarhunenMultivariate}).

\begin{prop}[Mercer's Theorem] \label{prop:Mercer}
For $j = 1, \ldots, p$ and $s_j, t_j \in \mathcal{T}_j$ it holds that
\[ \operatorname{Cov}\left(X^{(j)}(s_j), X^{(j)}(t_j)\right) = C_{jj}(s_j, t_j) = \sum \nolimits_{m=1}^\infty \nu_m \psi_m^{(j)}(s_j) \psi_m^{(j)}(t_j),\]
where the convergence is absolute and uniform.
\end{prop}

\begin{prop}[Multivariate Karhunen-Lo\`{e}ve Theorem] \label{prop:KarhunenMultivariate}
Under the assumptions of Prop.~\ref{prop:GammCompact},
\begin{equation}
X(\boldsymbol{t}) = \sum \nolimits_{m=1}^\infty \rho_m \psi_m(\boldsymbol{t}), \quad \boldsymbol{t} \in \mathcal{T},
\label{eq:Karhunen}
\end{equation}
with zero mean random variables $\rho_m = \myscal{X}{\psi_m}$  and $\operatorname{Cov}(\rho_m, \rho_n) = \nu_m \delta_{mn}$.  Moreover
\[ \mathbb{E} \left( \vecnorm{X(\boldsymbol{t}) - \sum \nolimits_{m=1}^M \rho_m \psi_m(\boldsymbol{t})}^2 \right) \to 0 \quad \text{for} ~ M \to \infty\]
uniformly for $ \boldsymbol{t} \in \mathcal{T}$.
\end{prop}

The multivariate Karhunen-Lo\`{e}ve representation has an analogous interpretation as in the univariate case \citep[Chapter 8.2.]{RamsaySilverman:2005}. The eigenvalues $\nu_m$ represent the amount of variability in $X$ explained by the single \textit{multivariate functional principal components} $\psi_m$, while the \textit{multivariate functional principal component scores} $\rho_m$ serve as weights of $\psi_m$ in the Karhunen-Lo\`{e}ve representation of $X$.
As the eigenvalues $\nu_m$ decrease towards 0, leading  eigenfunctions reflect the most important features of $X$. Truncated Karhunen-Lo\`{e}ve expansions, optimal $M$-dimensional approximations to $X$,
\begin{align}
X_{\lceil  M\rceil}(\boldsymbol{t})
:= \sum \nolimits_{m=1}^M \rho_m \psi_m(\boldsymbol{t}), \quad \boldsymbol{t} \in \mathcal{T},
\label{eq:truncKH}
\end{align}
are often used in practice. 
Single observations $x_i$ of $X$ can then be characterized by their score vectors $\left(\rho_{i,1} , \ldots ,  \rho_{i,M} \right)$ with $\rho_{i,m} = \myscal{x_i}{\psi_m}$ for further analysis,  e.g. for regression \citep{MuellerStadtmueller:2005} or clustering \citep{JacquesPreda:2014}.

\newpage
%Estimation algorithm
\section{Multivariate FPCA}
\label{sec:MFPCA}

\subsection{Relationship Between Univariate and Multivariate FPCA for Finite Karhunen-Lo\`{e}ve Decompositions}
\label{sec:UniMultiFPCA}

Given the Karhunen-Lo\`{e}ve representation of multivariate functional data $X$ as in  \eqref{eq:Karhunen}, a natural question is how this representation relates to the univariate Karhunen-Lo\`{e}ve representations of the single elements $X^{(j)}$. The following proposition establishes a direct relationship between these two representations if they are both finite, based on the theory of integral equations \citep{Zemyan:2012}.

\begin{prop} \label{prop:multi_single_EV}
The multivariate functional vector $X = \left(X^{(1)} , \ldots ,  X^{(p)}\right)$ in \eqref{eq:Karhunen} has a finite Karhunen-Lo\`{e}ve representation if and only if all univariate elements $X^{(1)},\ldots,X^{(p)}$, have a finite Karhunen-Lo\`{e}ve representation. In this case, it holds:
\begin{enumerate}
\item Given the multivariate Karhunen-Lo\`{e}ve representation \eqref{eq:Karhunen}, the positive eigenvalues $\lambda_1^{(j)} \geq \ldots \geq \lambda_{M_j}^{(j)} > 0,~ M_j \leq M$ of the univariate covariance operator $\Gamma^{(j)}$ associated with $X^{(j)}$ correspond to the positive eigenvalues of the matrix $\boldsymbol{A^{(j)}} \in \mathbb R ^{M \times M}$ with entries
\[A_{mn}^{(j)} =  (\nu_m \nu_n)^{1/2}  \scal{\psi_m^{(j)}}{\psi_n^{(j)}}, \quad m,n = 1 , \ldots ,  M.\]
The eigenfunctions of $\Gamma^{(j)}$ are given by
\[
\phi_m^{(j)}(t_j) = \left(\lambda_m^{(j)} \right)^{-1/2} \sum \nolimits _{n=1}^M \nu_n^{1/2}  [\boldsymbol{u_m^{(j)}}]_n \psi_n^{(j)}(t_j), \quad t_j \in \mathcal{T}_j,~m = 1, \ldots ,  M_j,\]
where $\boldsymbol{u_m^{(j)}}$ denotes an (orthonormal) eigenvector of $\boldsymbol{A^{(j)}}$ associated with  eigenvalue $\lambda_m^{(j)}$ and $[\boldsymbol{u_m^{(j)}}]_n$ denotes the $n$-th entry of this vector. 
For the univariate scores
\[\xi_m^{(j)} 
= \scal{X^{(j)}}{\phi_m^{(j)}}
= \left(\lambda_m^{(j)}\right)^{-1/2} \sum \nolimits _{n = 1}^M \nu_n^{1/2} \left[\boldsymbol{u_m^{(j)}} \right]_n \sum \nolimits _{k = 1}^M \rho_k \scal{\psi_n^{(j)}}{\psi_k^{(j)}}.\]

\item Assuming the univariate Karhunen-Lo\`{e}ve representation $
X^{(j)} = \sum \nolimits _{m=1}^{M _j}\xi_m^{(j)} \phi_m^{(j)}$
with
$\Gamma^{(j)} \phi_m^{(j)} = \lambda_m^{(j)} \phi_m^{(j)}$
for each element $X^{(j)}$ of $X$, the positive eigenvalues $\nu_1 \geq \ldots \geq \nu_M > 0$ of $\Gamma$ with $M \leq \sum \nolimits _{j=1}^p M_j =: M_+$ correspond to the positive eigenvalues of the matrix
$\boldsymbol{Z} \in \mathbb R ^{M_+ \times M_+}$ consisting of blocks $\boldsymbol{Z^{(jk)}} \in \mathbb{R}^{M_j \times M_k}$ with entries
\[ Z_{mn}^{(jk)} = \operatorname{Cov} \left(\xi_m^{(j)}, \xi_n^{(k)}\right), \quad m = 1, \ldots ,  M_j,~ n = 1, \ldots ,  M_k,~ j,k = 1 , \ldots ,  p.\]
The eigenfunctions of $\Gamma$ are given by their elements
\[
\psi_m^{(j)}(t_j) = \sum \nolimits _{n = 1}^{M_j} [\boldsymbol{c_m}]_n^{(j)} \phi_n^{(j)}(t_j),\quad t_j \in \mathcal{T}_j,~ m = 1 , \ldots ,  M,
\]
where $[\boldsymbol{c_m}]^{(j)} \in \mathbb R ^{M_j}$ denotes the $j$-th block of an (orthonormal) eigenvector $\boldsymbol{c_m}$ of $\boldsymbol{Z}$  associated with eigenvalue $\nu_m$. The scores are given by
\[\rho_m = \sum \nolimits _{j = 1}^p \sum \nolimits _{n=1}^{M_j}  [\boldsymbol{c_m}]_n^{(j)} \xi_n^{(j)}.\]
\end{enumerate}
\end{prop}

\textbf{Extensions:} The second part of Prop.~\ref{prop:multi_single_EV} can be extended in a natural way if univariate elements are expanded in finitely many, not necessarily orthonormal basis functions $b_m^{(j)}$ with coefficients $\theta_m^{(j)}$, i.e.
\begin{equation}
X^{(j)}(t_j) = \sum \nolimits _{m = 1}^{K_j} \theta_m^{(j)} b_m^{(j)}(t_j), \quad t_j \in \mathcal{T}_j.
\label{eq:basisGenExpansion}
\end{equation}
This is a very likely situation in practice, e.g. due to pre-smoothing of noisy observations.
Following analogous steps as in the proof of Prop.~\ref{prop:multi_single_EV} results in an eigenanalysis problem  $\boldsymbol{BQ c} = \nu \boldsymbol{c}$ as starting point for the MFPCA. Here $\boldsymbol{B} \in \mathbb{R}^{K_+ \times K_+}$ with $K_+ = \sum_{j = 1}^p K_j$ is a block diagonal matrix of scalar products $\scal{b_m^{(j)}}{b_n^{(j)}}$ of univariate basis functions associated with each element $X^{(j)}$. In the special case that all univariate bases are orthonormal (e.g. when using the univariate principal component bases as in Prop.~\ref{prop:multi_single_EV}), $\boldsymbol{B}$ equals the identity matrix. The symmetric block matrix $\boldsymbol{Q}$ with entries $Q_{mn}^{(jk)} = \operatorname{Cov} (\theta_m^{(j)}, \theta_n^{(k)})$ corresponds to $\boldsymbol{Z}$ in Prop.~\ref{prop:multi_single_EV}. 
Although $\boldsymbol{BQ}$ is in general not symmetric, its eigenvectors $\boldsymbol{c}_m$ and eigenvalues $\nu_m$, which are at the same time the eigenvalues of $\Gamma$, are real. This can be easily shown using the Cholesky decomposition of the symmetric matrix $\boldsymbol{B} = \boldsymbol{R R^\top}$ and solving $\boldsymbol{ R^\top Q R  \tilde c} = \nu\boldsymbol{\tilde c}$ with $ \boldsymbol{\tilde c = R^{-1} c}$. 
The estimation algorithm for principal components $\psi_m$ and associated scores $\rho_m$ based on this general basis expansion is presented in the next section combined with the case of a weighted scalar product.

\subsection{Estimation of Multivariate FPCA}
\label{sec:estMFPCA}

\textbf{Estimation based on univariate FPCA:}
The second part of Prop.~\ref{prop:multi_single_EV}. suggests a simple and natural approach for estimating the MFPCA.  After calculation of unvariate FPCAs for each element, the estimates can be plugged into  the formulae given in Prop.~\ref{prop:multi_single_EV}. 
Given de-meaned samples $x_1, \ldots ,  x_N$ of $X$, the proposed estimation procedure for MFPCA consists of four steps:
\begin{enumerate}
\item \label{algo:estuFPCA}For each element $X^{(j)}$ estimate a univariate FPCA based on the observations $x_1 ^{(j)} $, $\ldots $, $ x_N ^{(j)}$. This results in estimated eigenfunctions $\hat \phi_m^{(j)}$ and scores $\hat \xi_{i,m}^{(j)},~ i = 1 , \ldots, N,~ m = 1, \ldots, M_j$ for suitably chosen truncation lags $M_j$. 
As there exist numerous estimation procedures, e.g. for irregularly sampled and sparse data with measurement error \citep{YaoEtAl:2005}, the multivariate method is also applicable to this kind of data. 

\item \label{algo:estMFPCA_Z}
 Define the matrix $\boldsymbol{\Xi} \in \mathbb R ^{N \times M_+}$, where each row $( \hat \xi_{i,1}^{(1)} , \ldots ,  \hat \xi_{i,M_1}^{(1)} , \ldots ,  \hat \xi_{i,1}^{(p)} , \ldots ,  \hat \xi_{i, M_p}^{(p)} )$  contains all estimated scores for a single observation.
 An estimate $\boldsymbol{\hat Z} \in \mathbb R ^{M_+ \times M_+}$  of the block matrix $\boldsymbol{Z}$ in Prop.~\ref{prop:multi_single_EV} is given by $\boldsymbol{\hat Z} = (N-1)^{-1} \boldsymbol{\Xi} ^\top  \boldsymbol{\Xi}$.
\item \label{algo:estMFPCA_eigen}
Perform a matrix eigenanalysis for $\boldsymbol{\hat Z}$ resulting in eigenvalues $\hat \nu_m$ and orthonormal  eigenvectors $\boldsymbol{\hat c_m}$.
\item \label{algo:estMFPCA_result}
Estimates for the multivariate eigenfunctions are given by their elements
\begin{align}
 \hat \psi_m^{(j)}(t_j) &= \sum \nolimits _{n = 1}^{M_j} [\boldsymbol{\hat c_m}]_n^{(j)} \hat \phi_n^{(j)}(t_j), \quad t_j \in \mathcal{T}_j,~ m = 1 , \ldots, M^+
 \label{eq:MultEFun} \\
 \intertext{and multivariate scores can be calculated via}
 \hat \rho_{i,m} &= \sum \nolimits _{j = 1}^p \sum \nolimits _{n=1}^{M_j} [\boldsymbol{\hat c_m}]_n^{(j)} \hat \xi_{i,n}^{(j)} =  \boldsymbol{\Xi_{i, \cdot}}  \boldsymbol{\hat c_m}.
\label{eq:MultScore}
\end{align}
\end{enumerate}

Finding an appropriate truncation lag $M_j$ in step \ref{algo:estuFPCA}  is a well-known issue in functional data analysis. Common approaches are based on the decrease of the estimated eigenvalues $\hat \lambda_m^{(j)}$ \citep[scree-plot,][]{Cattell:1966} or the percentage of variance explained \citep[e.g.][Chapter 8.2.]{RamsaySilverman:2005}. An optimal number $M \leq M_+$ of multivariate functional principal components can basically be chosen with the same techniques, while the importance of a ``correct'' choice depends on the specific application: For simply exploratory aims it is less crucial than for subsequent analyses that ignore the information of the eigenvalues (and hence, the proportion of variance explained by the single components) and are based solely on multivariate eigenfunctions or scores, as e.g. clustering or functional principal component regression. For the latter, relevant eigenfunctions can also be selected using model-based approaches such as AIC or cross-validation.
The goodness of the resulting MFPCA estimates of course depends on an appropriate choice of $M_j$, which can also be used as a sensitivity check: If the first $M_j$ eigenfunctions capture all the relevant information in $X^{(j)}$, increasing $M_j$ will add only little information and hence should have only little impact on the results. This relationship is analyzed in a simulation in the online appendix.

\textbf{Extensions}: The estimation algorithm  can easily be extended to elements $X^{(j)}$ available in general basis expansions as in \eqref{eq:basisGenExpansion} and to MFPCA based on a weighted scalar product as in \eqref{eq:weightedSP}.
Given weights $w_1, \ldots, w_p > 0$ and demeaned observations $x_1 , \ldots, x_N$ of $X$ with estimated basis function coefficients $\hat \theta_{i,m}^{(j)}$ for each element, the eigenanalysis problem to solve is
\begin{equation}
(N-1)^{-1} \boldsymbol{B D \Theta}^\top \boldsymbol{\Theta D c} = \nu \boldsymbol{ c}.
\label{eq:eigenExtAlgo}
\end{equation}
The matrix $\boldsymbol{B}$ is the block diagonal matrix of basis scalar products as in Section~\ref{sec:UniMultiFPCA} and  $\boldsymbol{D} = \operatorname{diag}(\boldsymbol{w_1}^{1/2}, \ldots, \boldsymbol{w_p}^{1/2}) \in \mathbb R ^{K_+ \times K_+}$ accounts for the weights, where each $w_j^{1/2}$ is repeated $K_j$ times to give $\boldsymbol{w_j}^{1/2}$. 
$\boldsymbol{\Theta} \in \mathbb{R}^{N \times K_+}$ with rows 
$( \hat \theta_{i,1}^{(1)} , \ldots ,  \hat \theta_{i,K_1}^{(1)} , \ldots ,  \hat \theta_{i,1}^{(p)} , \ldots ,  \hat \theta_{i,K_p}^{(p)} )$ corresponds to the matrix $\boldsymbol{\Xi}$ defined in step~\ref{algo:estMFPCA_Z} of the original algorithm and $(N-1)^{-1}\boldsymbol{\Theta^\top \Theta}$ is an estimate for $\boldsymbol{Q}$ introduced in Section~\ref{sec:UniMultiFPCA}. 
Given eigenvectors $\boldsymbol{\hat c_m}$ and eigenvalues $\hat \nu_m$ for \eqref{eq:eigenExtAlgo}, estimated orthonormal eigenfunctions $\hat \psi_m$ of $\Gamma_w$ and associated scores $\hat \rho_{i,m}$ can be calculated in analogy to \eqref{eq:MultEFun} and \eqref{eq:MultScore} with $\boldsymbol{\hat Q_w} = (N-1)^{-1} \boldsymbol{D \Theta} ^\top \boldsymbol{\Theta D}$:
\begin{align*}
\hat \psi_m^{(j)}(t_j) &= \left(w_j \cdot \hat \nu_m \boldsymbol{\hat c_m}^\top \boldsymbol{ \hat Q_w  \hat c_m} \right)^{-1/2}\sum \nolimits _{k = 1}^p \sum \nolimits _{l = 1}^{K_j} \sum \nolimits _{n = 1}^{K_k} [\boldsymbol{\hat Q_w}]_{ln}^{(jk)} [\boldsymbol{\hat c_m}]_n^{(k)}  b_l^{(j)}(t_j), \\
\hat \rho_{i,m} &= 
\left(\hat \nu_m\right)^{1/2} \left(\boldsymbol{\hat c_m}^\top  \boldsymbol{\hat Q_w \hat c_m} \right)^{-1/2}  \boldsymbol{\Theta_{i,\cdot} D \hat c_m} .
\end{align*}
Clearly, the original algorithm is obtained as a special case with $\boldsymbol{\Theta} = \boldsymbol{\Xi},~ \boldsymbol{B} = \boldsymbol{I}$ (univariate FPCA for each element) and $\boldsymbol{D} = \boldsymbol{I}$ (all weights equal to $1$). Moreover, the extended algorithm allows to flexibly combine univariate FPCA and general basis expansions for different elements of the multivariate functional data.

If all elements $X^{(j)}$ are defined on the same (one-dimensional) interval and $\boldsymbol{D} = \boldsymbol{I}$, expanding each element in a general basis is equivalent to the method of \citet{JacquesPreda:2014}. The approach proposed in this paper, however, is more general, as it allows for different intervals as well as for higher-dimensional $\mathcal{T}_j$ and thus basis functions $b_m^{(j)}$.

\textbf{Implementation:} All presented variations of the MFPCA estimation algorithm are implemented in an \texttt{R} package \texttt{MFPCA} \if1\blind
{\citep{MFPCA}}\fi. Univariate basis expansions include univariate FPCA (1D), smooth tensor PCA (2D), spline bases (1D/2D) and cosine bases (2D/3D). New bases can be added easily and in a modular way.
The \texttt{MFPCA} package is based on the package \texttt{funData}  \if1\blind
{\citep{funData} }\fi for representing (multivariate) functional data on potentially different dimensional domains.
\if0\blind
{Both packages are available in the supplementary material.} \fi

\begin{subsection}{Asymptotic Properties}
\label{sec:Asymptotics}

The results of Prop.~\ref{prop:multi_single_EV} and the estimators proposed in the previous section have been derived under the assumption of a finite sample size $N$ and a finite Karhunen-Lo\`{e}ve representation for each element $X^{(j)}$. 
 This case is relevant in practice, since data is observable only in finite form (finitely many observations, finite resolution) and hence contains only finite information. In this case, the maximal number of principal components which can be estimated is limited to the number of observations $N$. For a growing number of observations, the truncation limits $M_j$ and thus $M_+$ may increase with $N$. 
All asympotic examinations hence have to consider the approximation error caused by truncating the univariate Karhunen-Lo\`{e}ve representations to finite sums as well as the estimation error.
For the eigenfunctions (analogously for the eigenvalues and scores) one hence has the following decomposition:
\[ \mynorm{\psi_m - \hat \psi_m} \leq \mynorm{\psi_m - \psi^{[M]}_m} + \mynorm{\psi^{[M]}_m - \hat \psi_m}.\]
Here $\psi_m$ is the true $m$-th eigenfunction of the covariance operator  $\Gamma$ and $\hat \psi_m$ is the estimator based on the assumption of a finite Karhunen-Lo\`{e}ve representation in each element. This assumption is reflected in $\psi^{[M]}_m$, which denotes the $m$-th eigenfunction of the covariance operator $\Gamma^{[M]}$ associated with $X^{[M]}$ with elements equal to the truncated $X^{(j)}$. These are really the eigenfunctions targeted with the estimation algorithm presented in Section~\ref{sec:estMFPCA}.
The first term on the right hand side of the inequality can be seen as a bias term caused by truncation. It depends on $N$ only implicitly via $M_1, \ldots, M_p$. The second term accounts for the estimation error, thus can be interpreted as a variance term.

\begin{prop}[Approximation Error]  \label{prop:asymptBias}
Let $\nu^{[M]}_m,~ m \in \mathbb{N}$ be the eigenvalues of the covariance operator $\Gamma^{[M]}$ associated with $X^{[M]}$ having truncated univariate elements $X^{[M](j)} = \sum_{m = 1}^{M_j} \xi_m^{(j)} \phi_m^{(j)}$.
Then the approximation error $\mynorm{X^{[M]} - X}$ converges to $0$ in probability for $M_1, \ldots, M_p \to \infty$.
For each $m \in \mathbb{N}$, $\nu^{[M]}_m$ converges to $\nu_m$  including multiplicity  and the total projection $P^{[M]}_m$ of $\mathcal{H}$ onto the eigenspace of $\Gamma^{[M]}$ associated with $\nu^{[M]}_m$ converges in norm to the total projection $P_m$ of $\mathcal{H}$ onto the eigenspace of $\Gamma$ associated with  $\nu_m$. 

 In particular, if $\nu_m$ and $\nu^{[M]}_m$ both have  multiplicity $1$ with associated eigenfunctions $\psi_m$ and $\psi^{[M]}_m$, such that $\myscal{\psi_m}{\psi_m^{[M]}} \geq 0$, then
\[\mynorm{\psi_m^{[M]} - \psi_m} \to 0 \quad \text{for}~M_1, \ldots, M_p \to \infty.\]
The scores $\rho^{[M]}_m := \myscal{X^ {[M]}}{\psi^{[M]}_m}$ converge in probability to $\rho_m$ for all $m \in \mathbb{N}$.
\end{prop}

In the remainder of this section, all nonzero eigenvalues $\nu_m$ are assumed to have multiplicity $1$, as then the eigenfunctions $\psi_m^{[M]}$ converge to $\psi_m$, if their orientation is chosen such that $\myscal{\psi_m}{\psi_m^{[M]}} \geq 0$.

For the estimation error, consider the univariate elements $X^{(j)}$ of $X$  with covariance operator $\Gamma^{(j)}$ and associated eigenvalues $\lambda_m^{(j)}$ and eigenfunctions $\phi_m^{(j)},~ m = 1, \ldots, M_j$. In the following, let $X_1, \ldots, X_N$ be independent copies of $X$ and assume for all $j = 1, \ldots, p$
\begin{gather}
\Delta^{(j)}_{M_j}:= \sup_{m = 1, \ldots, M_j}(\lambda_m^{(j)} - \lambda_{m+1}^{(j)})^{-1} < \infty ~\text{for every finite}~M_j \tag{A1}  \label{ass:eValGap} \\
 \int_{\mathcal{T}_j} \int_{\mathcal{T}_k} \mathbb{E}\left(X^{(j)}(t_j)^2 X^{(k)}(s_k)^2\right) \mathrm{d} s_k ~ \mathrm{d} t_j < \infty \quad \forall~ k = 1, \ldots, p \tag{A2}  \label{ass:finite4Moment} \\
\left  \Vert \Gamma^{(j)}- \hat\Gamma^{(j)} \right  \Vert_\text{op}   = O_p(r_N^{\Gamma}) \tag{A3} \label{ass:asymptCov}\\
 \scal{\phi_m^{(j)}}{\hat \phi_m^{(j)}}   \geq  0~\text{for all}~m = 1,\ldots, M_j  \tag{A4} \label{ass:eFunSign} \\
 \hat \xi_{i,m}^{(j)}   = \scal{X_i^{(j)}}{\hat \phi_m^{(j)}} ~\text{for all}~m = 1,\ldots, M_j,~ i = 1 ,\ldots, N \tag{A5} \label{ass:scoreCalc}
\end{gather}

\eqref{ass:eValGap} -- \eqref{ass:finite4Moment} concern theoretical properties of $X^{(j)}$ and $\Gamma^{(j)}$, while \eqref{ass:asymptCov} -- \eqref{ass:scoreCalc} depend on the univariate decompositions used.
\eqref{ass:eValGap} is a standard assumption in univariate FPCA \citep{Bosq:2000,Hall:2006}. It  guarantees that the first $M_j$ univariate eigenvalues of each element all have multiplicity 1.
With \eqref{ass:finite4Moment}, the integral operator with kernel $\hat C_{jk}(s,t) := N^{-1} \sum_{i = 1}^N X_i^{(j)}(s) X_i^{(k)}(t)$ converges to the one with kernel $C_{jk}(s,t)$ with rate $N^{-1/2}$. \eqref{ass:finite4Moment} is used in combination with \eqref{ass:scoreCalc} to obtain a convergence rate for the maximal eigenvalue of $\boldsymbol{Z} - \boldsymbol{\hat Z}$, which, in turn, affects the convergence of the eigenvectors $\boldsymbol{\hat c_m}$ to $\boldsymbol{c_m}$ \citep{YuEtAl:2015}. 
\eqref{ass:asymptCov} ensures that the operator $\hat \Gamma^{(j)}$, which is the basis of the univariate FPCA, converges to $\Gamma^{(j)}$ in the operator norm $\left  \Vert \cdot \right  \Vert_\text{op}$ induced by $\norm{\cdot}$ with a given rate $r^{\Gamma}_N$. 
For fully observed data, \cite{HallHorowitz:2007} show $r_N^\Gamma = N^{-1/2}$, while the approach of \cite{YaoEtAl:2005} yields $r^{\Gamma}_N = N^{-1/2}h^{-2}$ in the case of measurement error or irregularly sampled data for a certain bandwidth $h$. 
Together with \eqref{ass:eValGap}, $r^{\Gamma}_N$ gives a convergence rate for the univariate eigenfunctions $\hat \phi_m^{(j)}$ \citep[][Lemma 4.3]{Bosq:2000}. 
\eqref{ass:eFunSign} guarantees that $\hat \phi_m^{(j)}$ is an estimator for $\phi_m^{(j)}$ rather than for $-\phi_m^{(j)}$, as eigenfunctions are defined only up to a sign change \citep{Bosq:2000, Hall:2006}. 
Finally, \eqref{ass:scoreCalc} is used to formulate the convergence of the estimated scores in terms of convergence rates for the estimated eigenfunctions. If this assumption does not hold \citep[e.g. in][]{YaoEtAl:2005}, convergence results can still be obtained e.g. by assuming a convergence rate for $\hat \xi_{i,m}^{(j)}$ and replacing \eqref{ass:finite4Moment} by an assumption on the rate of convergence for the maximal eigenvalue of $\boldsymbol{Z} - \boldsymbol{\hat Z}$.

\begin{prop}[Estimation Error]  \label{prop:asymptVar}

Assume \eqref{ass:eValGap} --  \eqref{ass:scoreCalc} hold. Then for  $M_{\max} = \max_{j = 1 , \ldots,  p}M_j$ and $\Delta_M:= \max_{j = 1, \ldots, p}{\Delta^{(j)}_{M_j}}$, the  maximal eigenvalue of $\boldsymbol{Z} - \boldsymbol{\hat Z}$ can be characterized by 
\[ \lambda_{\max}(\boldsymbol{Z} - \boldsymbol{\hat Z}) = O_p(M_{\max} \max(N^{-1/2} , \Delta_M r_N^\Gamma)).\]
Using the same notation as in Prop.~\ref{prop:asymptBias}, it holds for all $m = 1,\ldots, M_+$ that
\begin{align*}
\left|\nu_m^{[M]} - \hat \nu_m \right| &=O_p(M_{\max} \max(N^{-1/2} , \Delta_M r_N^\Gamma)), \\
 \mynorm{\psi^{[M]}_m - \hat \psi_m} &
 = O_p(M_{\max}^{3/2} \max(N^{-1/2} , \Delta_M r_N^\Gamma)),\\
\left|\rho^{[M]}_{i,m} - \hat \rho_{i,m}\right| &=  O_p(M_{\max}^{3/2} \max(N^{-1/2} , \Delta_M r_N^\Gamma)), \\
\mynorm{X_i^{[M]} - \hat X_i^{[M]}} &= O_p(M_{\max} \Delta_M r_N^\Gamma )
\end{align*}
with $\hat X_i^{[M](j)} = \sum_{m = 1}^{M_j} \hat \xi_{i,m}^{(j)} \hat \phi_m^{(j)}$.
\end{prop}

When combining the results of Prop.~\ref{prop:asymptBias} and Prop.~\ref{prop:asymptVar}, the analogy to bias and variance again becomes  apparent:
For fixed $N$, higher values of $M_1, \ldots, M_p$ will reduce the approximation error, but simultaneously increase the estimation error, as both $M_{\max}$ and $\Delta_M$ increase with $M_j$.
If one assumes for example $M_1 = \ldots = M_p = M_{\max} = O(N^\beta)$, $r_N^\Gamma = N ^{-1/2}$, and that the eigengaps fulfill 
$\lambda^{(j)}_m - \lambda_{m+1}^{(j)} \geq C^{-1} m^{-\alpha-1}$ with $\alpha  > 1,~ C > 0$  \citep[cf.][]{HallHorowitz:2007},
the MFPCA estimators given in  Section~\ref{sec:estMFPCA} are consistent for $0 < \beta < (2\alpha +  5)^{-1}$.

\end{subsection}

% Simulation
\section{Simulation}
\label{sec:Simulation}

We illustrate the performance of our new MFPCA estimation procedure in three settings with increasing complexity:
\begin{enumerate}
\item Densely observed bivariate functional data on the same one-dimensional interval.
\item Trivariate functional data on different one-dimensional  intervals with different levels of sparsity.
\item Bivariate functional data on different dimensional domains (images and functions).
\end{enumerate}
The first two settings deal with multivariate functional data on one-dimensional domains and are presented together in Section~\ref{sec:simFun}. Setting 3 is discussed separately in Section~\ref{sec:simImage}. Examples for simulated data and estimation results for all three settings are given in the online appendix, which also includes two additional simulations (cf. Sections~\ref{sec:estMFPCA} and \ref{sec:application}). Unless specified otherwise, the \texttt{MFPCA} package
\if1\blind
{
\citep{MFPCA}
}\fi
is used for all calculations.

Each setting is based on 100 datasets with $N = 250$ observations of the form
\[x_i(\boldsymbol{t}) = \sum \nolimits_{m = 1}^M \rho_{i,m} \psi_m(\boldsymbol{t}) +\boldsymbol{\varepsilon}_i(\boldsymbol{t}),\quad \boldsymbol{\varepsilon}_i(\boldsymbol{t}) \stackrel{\text{iid}}{\sim} \mathrm{N}_p(0, \sigma^2 \boldsymbol{I}),\quad \boldsymbol{t} \in \mathcal{T},~ i = 1 , \ldots,  N.\]
In each case, we consider data without ($\sigma^2 = 0$) and with ($\sigma^2 = 0.25$) measurement error. The scores $\rho_{i,m}$ are independent samples from $\mathrm{N}(0, \nu_m)$ for eigenvalues with exponential ($\nu_m^\text{exp} = \exp(-(m+1)/2)$) or  linear ($\nu_m^\text{lin} = \left(M+1-m\right)/M$) decrease, while the choice of $\mathcal{T}, M$ and $\psi_m$ varies between settings (see Sections~\ref{sec:simFun} and \ref{sec:simImage}). In all cases, we use unit weights ($w_j = 1$).
The accuracy of the resulting estimates $\hat \nu_m$ and $\hat \psi_m$ is measured by the relative errors $\operatorname{Err}(\hat \nu_m) = \left( \nu_m - \hat \nu_m \right)^2/\nu_m^2$ and $
\operatorname{Err}(\hat\psi_m) = \mynorm{\psi_m - \hat \psi_m}^2$. 
As functional principal components are defined only up to a sign change, the estimate $\hat \psi_m$ is reflected, i.e. multiplied by $-1$, if $\myscal{\psi_m}{\hat \psi_m} < 0$.  The goodness of the reconstructed observations $\hat x_i = \sum \nolimits _{m = 1}^M \hat \rho_{i,m} \hat \psi_m^{(j)}$ is evaluated by the mean relative squared error  $\operatorname{MRSE}
 =N^{-1} \sum \nolimits _{i = 1}^N \left( \mynorm{x_i - \hat x_i}^2 / \mynorm{x_i}^2 \right)$.

\subsection{Multivariate Functional Data on One-Dimensional Domains}
\label{sec:simFun}

\textbf{Setting 1:}
For the first setting, the first $M = 8$ Fourier basis functions on $[0,2]$ are split into $p = 2$ parts. The pieces are shifted and multiplied by a random sign to form the elements $\psi_m^{(1)}$ and $\psi_m^{(2)}$ on $\mathcal{T}_1 = \mathcal{T}_2 =[0,1]$ (for technical details, see online appendix). 
The observations $x_i$ are sampled on an equispaced grid of $S_1 = S_2 = 100$ sampling points. The MFPCA is based on $M_1 = M_2 = 8$ univariate functional principal components that are calculated by the PACE algorithm  \citep{YaoEtAl:2005} with penalized splines to smooth the covariance function, as implemented in the \texttt{R} package \texttt{refund} \citep{refund}.  
In this simple setting of a common, one-dimensional domain, the new approach can be compared to the method of \citet{RamsaySilverman:2005}, which is implemented in the \texttt{R} package \texttt{fda} \citep{fda} and in the following denoted by $\text{MFPCA}_\text{RS}$. This method involves pre-smoothing of the elements with $K = 15$ cubic spline basis function. $\text{MFPCA}_\text{RS}$  computes score values $\hat \rho_{i,m}^{(j)} = \scal{x_i^{(j)}}{\hat \psi_m^{(j)}}$ for each observation $i$ and each element $j$. Since they do not have the same interpretation as the scores in the multivariate Karhunen-Lo\`{e}ve representation (Prop.~\ref{prop:KarhunenMultivariate}), $\sum \nolimits _{j = 1}^p \hat \rho_{i,m}^{(j)} = \hat \rho_{i,m}$ is used for comparison purposes.
 
\begin{figure}[hb!]
\centering
\includegraphics[width = 0.85\textwidth]{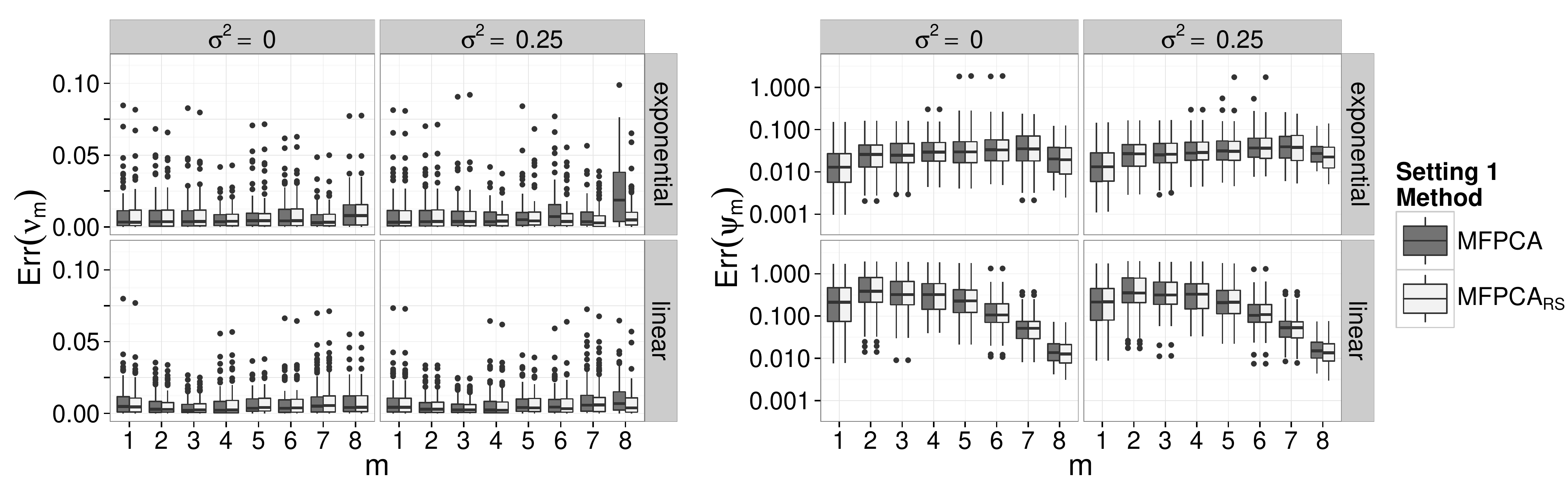}

\includegraphics[width = .85\textwidth]{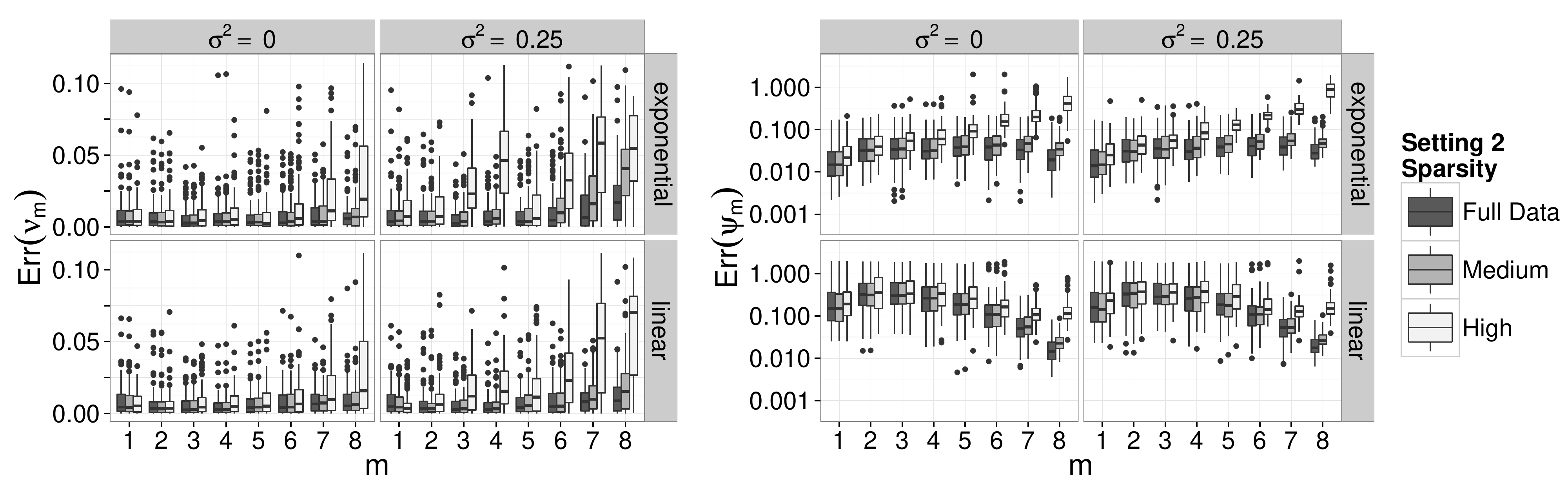}
\caption{\small Relative errors for estimated eigenvalues (left) and eigenfunctions (right, log-scale) for simulation settings 1 and 2, depending on eigenvalue decrease, measurement error and estimation method (setting 1) or sparsity (setting 2).}
\label{fig:sim1eValeFun}
\end{figure}

The results for the first setting are shown in Fig.~\ref{fig:sim1eValeFun}   and Table~\ref{tab:simrMSE}.  In total, the new approach can compete very well with the existing method of \citeauthor{RamsaySilverman:2005} and gives nearly identical results for synthetic and real data (see online appendix for the gait cycle example). Both techniques mostly have higher errors in $\psi_m$ for linearly decreasing eigenvalues, as in these cases, the eigenfunctions are more often confused, i.e. $\hat \psi_m$ is an estimate for e.g. $\psi_{m-1}$ or $\psi_{m+1}$ rather than for $\psi_m$. 
In the ideal case of no measurement error, $\text{MFPCA}_\text{RS}$ yields lower MRSE values than the new approach, which might be an effect of $\text{MFPCA}_\text{RS}$ expecting smooth or presmoothed data. For the practically relevant case of data with measurement error, both methods give almost the same prediction errors (cf. Table~\ref{tab:simrMSE}). Simulations based on Legendre polynomials gave very similar results (not shown here).

\begin{table}[ht]
\centering
\caption{
\small Average $\operatorname{MRSE}$ (in $\%$) for simulation setting 1 and 2, depending on eigenvalue decrease and measurement error.
}
\small
\begin{tabular}{llrrrr}
\hline \hline
\multicolumn{2}{c}{}  & \multicolumn{2}{c}{ $\sigma^2 = 0$} & \multicolumn{2}{c}{$\sigma^2 = 0.25$} \\
\multicolumn{2}{l}{Setting}    & $\nu_m^\text{exp}$ & $\nu_m^\text{lin}$ &  $\nu_m^\text{exp}$ & $\nu_m^\text{lin}$ \\ 
\hline 
1 & MFPCA &         0.006 &  0.009 &          0.740 &  0.355 \\ 
1 & $\text{MFPCA}_\text{RS}$ &    $< 10^{-3}$ & $< 10^{-3}$ &          0.720 &  0.338 \\ 
 \hline
 2 & Full Data &          0.004 &  0.007 &          0.778 &  0.367 \\ 
2 & Medium Sparsity &        0.164 &  0.146 &          2.070 &  1.102 \\ 
2 & High Sparsity &         5.755 &  4.568 &         15.365 & 10.824 \\ 
 \hline
\end{tabular}
\label{tab:simrMSE}
\end{table}

\textbf{Setting 2:}
Here we consider trivariate functional data on  $\mathcal{T}_1 =[-1, 0.5],~ \mathcal{T}_2 =[0,1],~ \mathcal{T}_3 = [1.5,2]$. The eigenfunctions are constructed according to the same scheme as in setting 1 by splitting the first $M = 8$ Fourier basis functions on $[0,2]$ into $p = 3$ parts, followed by a shift and multiplication with a random sign. The observations are sampled on equidistant grids with $S_1 = S_3 = 50$ and $S_2 = 100$ sampling points. We consider the dense observations as well as sparse variants with medium ($50 - 70\%$) and high ($90 - 95 \%$ missings) sparsity. The sparsification mechanism is analogous to \cite{YaoEtAl:2005} and applied to each observation and each element separately.
The MFPCA is calculated in the same way as in setting 1, using the PACE approach to estimate $M_1 = M_2 = M_3 = 8$ functional principal components for each element. For data with high sparsity, we set $M_1 = M_3 = 3$ and $M_2 = 5$ to make computation of the univariate FPCA feasible.

The results are given in Fig.~\ref{fig:sim1eValeFun} and Table~\ref{tab:simrMSE}. Here there is no available competitor. The performance of our MFPCA
for full data is very similar to the simpler case of setting 1. Even for a moderate level of sparsity, the new method yields excellent results for most eigenvalues and eigenfunctions at the expense of somewhat higher reconstruction errors. For very sparse data, the leading eigenvalues and eigenfunctions are still estimated well, but the reconstruction error is considerably higher than for the full data. However, this is still acceptable (average $\operatorname{MRSE}$ is lower than $16\%$ for all levels of sparsity), bearing in mind that data with high sparsity contains at most $10\%$ of the original information. Again, simulations based on Legendre polynomials gave very similar results (not shown here).

\subsection{Multivariate Functional Data Consisting of Functions and Images}
\label{sec:simImage}

\textbf{Setting 3:}
Observations are generated based on $M = 25$ principal components, where the image elements $\psi_m^{(1)}$ are formed by tensor products of Fourier basis functions on $\mathcal{T}_1 = [0,1] \times [0, 0.5]$ and $\psi_m^{(2)}$ are given by  Legendre polynomials on $\mathcal{T}_2 = [-1,1]$. 
The elements are weighted by random factors $\alpha^{1/2}$ and $(1 - \alpha)^{1/2}$, respectively, with $\alpha \in (0.2, 0.8)$ to ensure orthonormality. For the scores, only exponentially decreasing eigenvalues are used.
The observations are discretized using a grid of $S_1 = 100 \times 50$ equidistant points for the image element and $S_2 = 200$ equidistant points for the functions.
 
We consider the new MFPCA approach based on univariate FPCA as well as non-orthogonal basis functions. In the first case, the eigendecomposition for the image data is calculated with the FCP-TPA algorithm for regularized tensor decomposition \citep{Allen:2013}. The smoothing parameters for penalizing second differences in both image directions are chosen via generalized cross-validation in  $[10^{-5}, 10^5]$ \citep{Allen:2013, HuangEtAl:2009}.
Multivariate FPCA is calculated based on $M _1 = 20$ eigenimages and $M_2 = 15$ univariate eigenfunctions. In the case of general basis functions, image elements are expanded in tensor products of $K_1 = 10 \times 12$ B-splines and the one-dimensional element is represented in terms of $K_2 = 15$ B-spline basis functions. In the presence of measurement error the univariate expansions are fit with appropriate smoothness penalties \citep{EilersMarx:1996}.

\begin{figure}[ht]
\centering
\includegraphics[width = 0.85\textwidth]{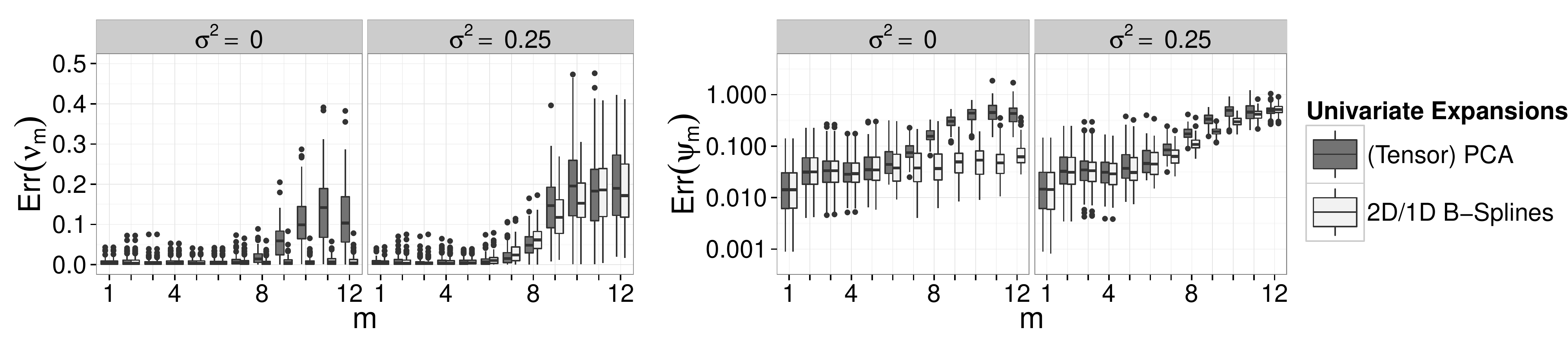}
\caption{\small Relative errors for estimated eigenvalues (left) and eigenfunctions (right, log-scale) for simulation setting 3, depending on measurement error and univariate expansions.}
\label{fig:sim3eValeFun}
\end{figure}

The overall results for the first $M = 12$ eigenvalue/eigenvector pairs are given in Fig.~\ref{fig:sim3eValeFun}. Compared to the settings with one-dimensional domain, the  errors are slightly higher, in particular  for  higher order eigenvalues and eigenfunctions. Exemplary results however, show that even in this case, the new approach is still able to capture the important features of the true eigenfunctions well (see online appendix). The results further show that the general approach with spline basis functions performs mostly better than the pure MFPCA approach. Moreover, the truncated Karhunen-Lo\`{e}ve representation with $M = 12$ (true $M = 25$) estimated eigenfunctions and scores gives an excellent reconstruction of the original data.  The average $\operatorname{MRSE}$ is $1.382\%$/$0.398\%$ (PCA/splines) for data without measurement error and $2.233\%$/$2.048\%$ (PCA/splines) for data with measurement error.

% Application
\section{Application  -- ADNI Study}
\label{sec:application}
\label{sec:applicateADNI}

In this section, the new method is applied to data from the
Alzheimer's Disease Neuroimaging Initiative study (ADNI), which aims at identifying biomarkers for accurate diagnosis of Alzheimer's disease (AD) in an early stage \citep{Mueller:2007}. We use MFPCA to explore how  longitudinal trajectories of a neuropsychological score 
(ADAS-Cog, a current standard for monitoring AD progression) covary with FDG-PET scans at baseline. The latter are used to assess the glucose metabolism in the brain, which is tightly coupled with neuronal function. As the brain images might be predictive of subsequent cognitive decline, common patterns between these two sources of information would be highly relevant.

\textbf{Dataset:} The dataset considered for MFPCA contains data from all $N = 483$ participants enrolled in ADNI1, having an FDG-PET scan at baseline and at least three ADAS-Cog measurements during follow-up. At baseline, $84$ subjects were diagnosed with AD, $302$ were suffering from mild cognitive impairment (MCI, in many cases an early stage of AD) and $97$ were cognitively healthy elderly controls.
The ADAS-Cog trajectories constitute the first element $X^{(1)}$, where high values indicate a high level of cognitive impairment. The measurements contain  missings, mostly in the second half of the study period and thus are sparse.
The second element $X^{(2)}$ is an axial slice of $93 \times 117 $ pixels (139.5 $\times$ 175.5 $\text{mm}^2$) of FDG-PET scans, containing the Precuneus and temporo-parietal regions. Both are believed to show a strong relation between hypometabolism (reduced brain function) and AD \citep{Blennow:2006}. Exemplary data is shown in Fig.~\ref{fig:alzheimerData}.

\begin{figure}[ht!]
\centering
\begin{minipage}{0.45\textwidth}
\includegraphics[width = \textwidth]{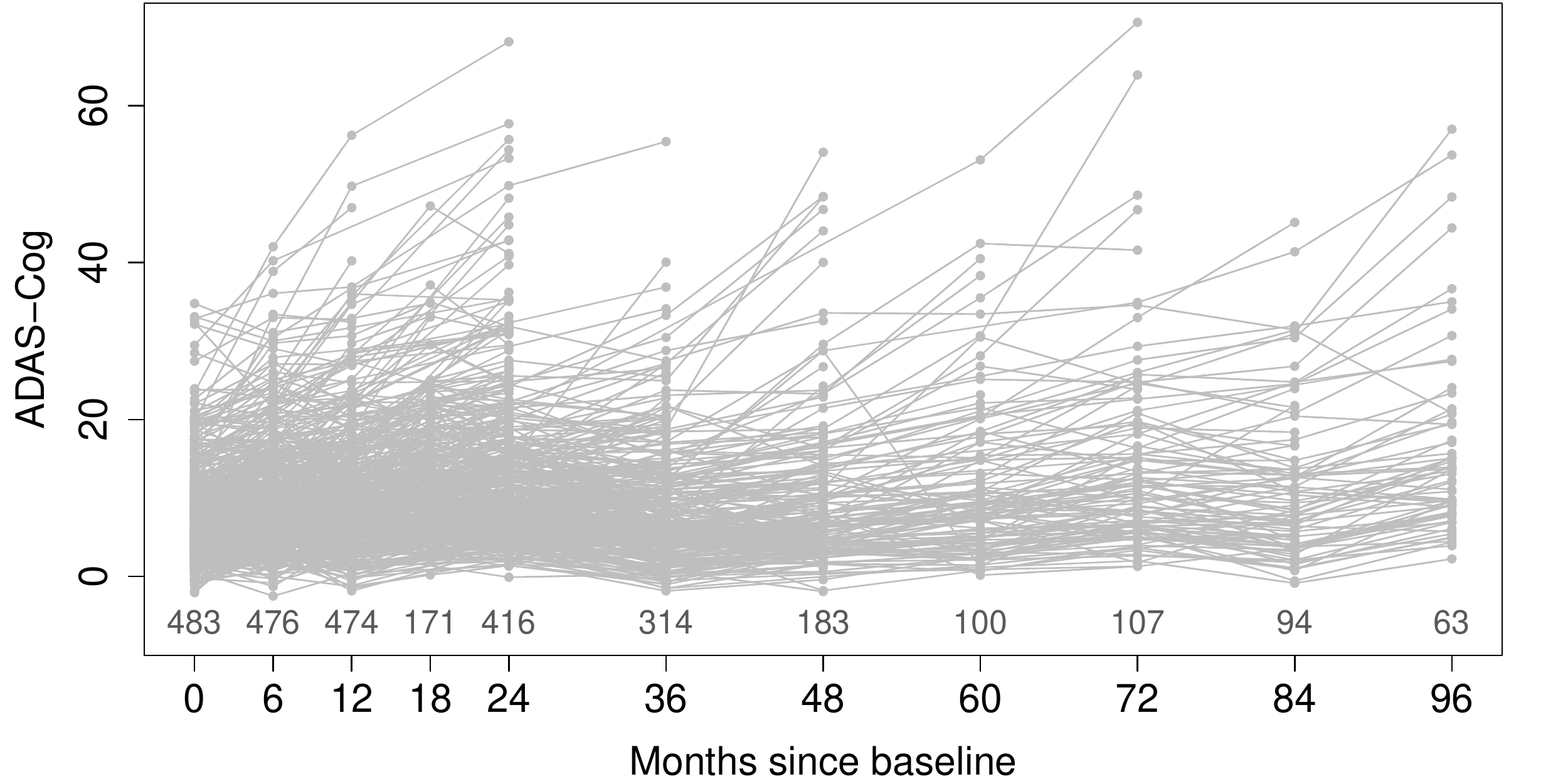}
\end{minipage}
\hfill
\begin{minipage}{0.54\textwidth}
\includegraphics[width =\textwidth]{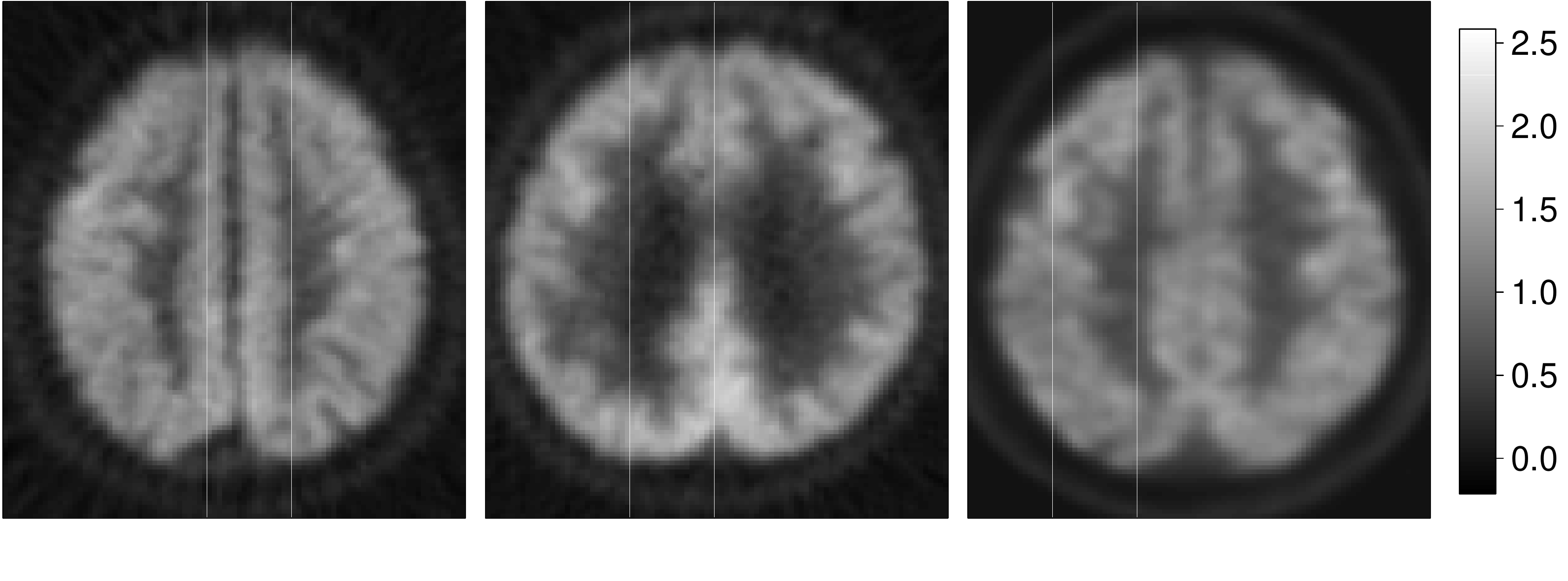}
\end{minipage}
\caption{
\small Left: ADAS-Cog trajectories for all $N = 483$ subjects. Numbers above the x-axis give the total number of measurements for each visit. Right: FDG-PET scans for three randomly chosen male subjects (left to right: AD, MCI, normal; diagnosis at baseline).}
\label{fig:alzheimerData}
\end{figure}

\textbf{Weighted scalar product:}
As the ADAS-Cog trajectories and FDG-PET scans differ considerably in  domain, range and variation (cf. Fig.~\ref{fig:alzheimerData}),  we  use a weighted MFPCA with 
\[w_j = \left(\int_{\mathcal{T}_j} \hat C_{jj} (t_j, t_j) \mathrm{d}  t_j \right)^{-1} = \left( \int_{\mathcal{T}_j} \widehat{\operatorname{Var}} \left(X^{(j)}(t_j) \right) \mathrm{d}  t_j \right)^{-1}, \quad j = 1,2 ,\]
where $\hat C_{jj}$ is estimated from the data. Using these weights, the integrated variance equals $1$ for the rescaled elements $\tilde X^{(j)} = w_j^{1/2}X^{(j)}$. All elements thus contribute equal amounts of variation to the analysis, similarly to multivariate PCA, where the data is usually standardized before the analysis. 
We believe that this a sensible choice for many applications, but there may of course be situations, in which other weighting schemes may be preferable. For example, one could think of data that has two image elements, representing brain regions of different size for the same imaging modality. Here variability is naturally on the same scale and it might be better to keep the information of the site of the individual domains by setting both weights to one. On the other hand, if the images stem from different imaging modalities on the same domain, it might be necessary to correct solely for differences in variation. 
As a general rule, the weights should be chosen in close coordination with practitioners, considering the objective of the analysis and the data at hand.

\textbf{Results: }
The results for the first two multivariate functional principal components, that account for $80.7\%$ of the total weighted variance, are shown in Fig.~\ref{fig:alzheimerMFPCAresults}. 
For the univariate expansions, we use FPCA for $X^{(1)}$ with $M_1 = 3$ principal components (explaining $99.2 \%$ of the univariate variance) and  $20 \times 15$ tensor product B-splines for the images $X^{(2)}$.
Fig.~\ref{fig:alzheimerMFPCAresults} further includes pointwise bootstrap confidence bands for the principal components based on $100$ nonparametric bootstrap iterations on the level of subjects.
The coverage of such confidence bands for data consisting of functions and images has been analyzed in a simulation study, which gave good results, even in the presence of measurement error (see online appendix).
The entire analysis for the ADNI data took around $15$ minutes on a standard laptop (2.7 GHz, 16 GB RAM) including the calculation of the bootstrap confidence bands and without parallelization.

Almost half of the variability in the data ($46.7\%$ of the weighted variance) is explained by the first functional principal component. The ADAS-Cog element -- and hence the degree of cognitive impairment -- is elevated relative to the mean and increases during follow-up. The FDG-PET element exhibits hypometabolism in the Precuneus and the temporo-parietal regions, i.e. this component reflects reduced brain activity in these regions already at baseline.  In total, the first eigenfunction seems to be interpretable as an AD related effect, as the pattern for positive scores perfectly agrees with medical knowledge about AD progression. This interpretation is supported by the estimated scores, which are mainly positive for people diagnosed with AD by their last visit, while scores of subjects who remained cognitively normal during follow-up are nearly all negative. Persons with MCI have intermediate score values, which is in line with the hypothesis that this diagnosis can constitute a transitional phase between normal ageing and AD.

For the second functional principal component (explains $33.9\%$ of  weighted variance),  the ADAS-Cog element is nearly constant and has wide bootstrap confidence bands that include zero during the whole follow-up.  In contrast, the FDG-PET element differs significantly from zero in almost all voxels (cf. Fig.~\ref{fig:alzheimerMFPCAresults}). 
Hence, this principal component reflects variation in the FDG-PET scans at baseline. Plotting the overall mean plus or minus this component   suggests that it can be interpreted as an effect of imperfect registration that manifests in different brain sizes, which are known to correlate with gender \citep{Ruigrok:2014}. This hypothesis is supported by the boxplots of the estimated scores  in Fig.~\ref{fig:alzheimerMFPCAresults}, while scores do not differ notably by diagnosis (not shown here).

\textbf{Discussion: }
The results show that MFPCA is able to capture important sources of variation in the data that have a meaningful interpretation from a medical and neuroimaging point of view. An important issue not addressed here is that for ADAS-Cog, there may well be an informative dropout of patients with high score values (cf. Fig.~\ref{fig:alzheimerData}). While addressing informative missingness goes beyond the scope of this paper, interpretation of results should take this possibility into account. For instance, it is easily conceivable that $\hat \psi_1^{(1)}$ may be underestimating $\psi_1^{(1)}$ towards the end of the study period.

\begin{figure}[h!]
\centering
\includegraphics[width = \textwidth]{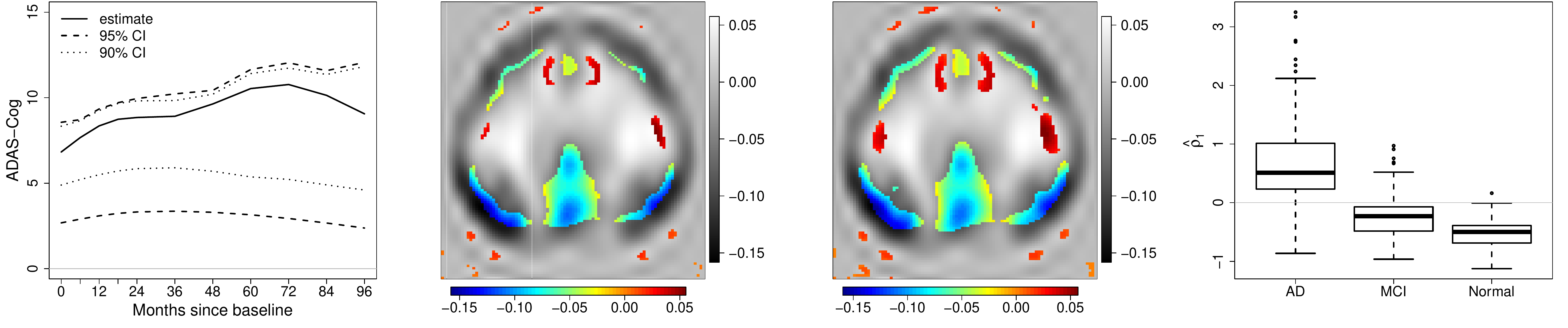}

\includegraphics[width = \textwidth]{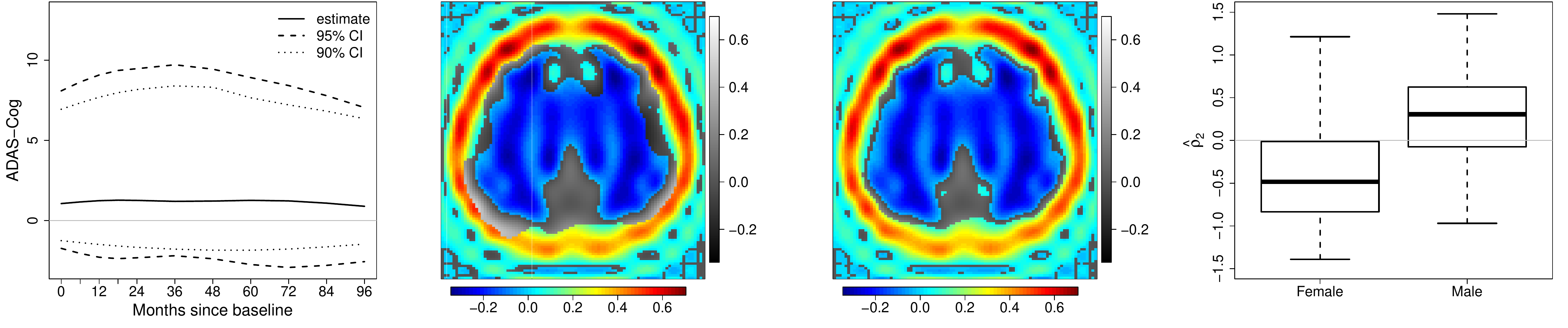}

\caption{
\small The first two estimated multivariate functional principal components for the ADNI data (1st row: $\hat \psi_1$, 2nd row: $\hat \psi _2$). Estimates are given with pointwise $95\%$ and $90\%$ bootstrap confidence bands based on $100$ nonparametric bootstrap iterations (ADAS-Cog, 1st column: Dashed lines; FDG-PET, 2nd and 3rd column: Pixels with pointwise $95\%$ (left) and $90\%$ (right) confidence bands not including zero in color). Boxplots of the scores (4th column) support the interpretation.}
\label{fig:alzheimerMFPCAresults}
\end{figure}

% Discussion
\section{Discussion and Outlook}
\label{sec:Discussion}

This paper introduces methodology and a practical estimation algorithm for multivariate functional principal component analysis. While other methods for MFPCA are restricted to observations on a common, one-dimensional interval, the new approach is suitable for  data on different domains, which may also differ in dimension, such as functions and images. 
The key results are  1. a Karhunen-Lo\`{e}ve Theorem, that establishes the theoretical basis for MFPCA (Prop.~\ref{prop:KarhunenMultivariate}), 2. an explicit relation between multivariate and univariate FPCA, which serves as a starting point for the estimation (Prop.~\ref{prop:multi_single_EV}) and 3. asymptotic results for the estimators (Prop.~\ref{prop:asymptBias} and \ref{prop:asymptVar}).
The  estimation algorithm can be extended to expansions of the univariate elements in  not necessarily orthonormal bases. This allows to flexibly choose an appropriate basis for each element depending on the data structure, in particular also mixtures of univariate FPCA and general bases. The algorithm is applicable to  sparse data or data with measurement error, as well as to images. Notably, the proposed method can be used to calculate smooth univariate functional principal components for data on higher dimensional domains and is hence an alternative to existing methods for tensor PCA \citep{Allen:2013}. The results of MFPCA give insights into simultaneous variation within the data and provide a natural tool for dimension reduction. Moreover, they can be used as a building block for further statistical analyses such as functional clustering methods or functional principal component regression with multiple covariates \citep[cf.][for the univariate case]{MuellerStadtmueller:2005}. If the elements differ in domain, range or variation, the new method can incorporate weights, which should be chosen with respect to the question of interest and the data at hand.

Possible extensions of the approach include normalization methods as an alternative to the weighted scalar product, following the ideas in \citet{JacquesPreda:2014} or \citet{ChiouEtAl:2014} for functions observed on a common interval. However, one should take into account that the domains may have different dimensions and sizes.
The concept of MFPCA could further be extended to \textit{hybrid data}, i.e. data consisting of a functional and a vector part \citep[Chapter 10.3.]{RamsaySilverman:2005}. A natural starting point would be to extend the scalar product suggested by \cite{RamsaySilverman:2005} in this context to multivariate functional data as proposed in Prop.~\ref{prop:HisHilbert}. However, transferring the results for MFPCA shown in this paper requires a  careful revision of the concept of the covariance operator and related proofs.
Finally, one could think of estimating the multivariate covariance operator directly without computing a univariate decomposition for each element. This operator is typically high-dimensional, making smoothing as well as an eigendecomposition hardly feasible, which is avoided in our two-step approach.

\section*{Supplementary Material}
The online appendix contains detailed proofs for all propositions, some  additional simulation results and \texttt{R} code for reproducing the analysis for the ADNI and gait cycle data based on the \texttt{R} packages \texttt{fundata} and \texttt{MFPCA} 
\if1\blind
{
\citep{funData, MFPCA}.
}
\else
{.
}\fi

\newpage
% Bibliography

%\printbibliography

\bibliographystyle{apalike}
\bibliography{bibliography/literature_MFPCA}

% Appendix
 \section*{Proofs of Propositions}

\begin{proof}[Proof of Prop.~\ref{prop:HisHilbert}] \label{proof:HisHilbert}
$\mathcal{H}$ is a direct sum of the  Hilbert spaces  $L^2(\mathcal{T}_j),~j = 1 , \ldots ,  p$, with natural scalar product $\myscal{\cdot}{\cdot}$ \citep[cf.][Chapter II.1.]{ReedSimon:2005}
\end{proof}

\begin{proof}[Proof of Prop.~\ref{prop:GammaProperties}] \label{proof:GammaProperties} \leavevmode
\begin{enumerate}
\item $\Gamma$ is linear: Follows from the linearity of the scalar product in \eqref{eq:GammaFormula}.
\item $\Gamma$ is self-adjoint: Follows from the symmetry $C_{ij}(s_i, t_j) = C_{ji}(t_j, s_i) $.
\item $\Gamma$ is positive: Let $f \in \mathcal{H}$. Then
\begin{align*}
\myscal{f}{\Gamma f}
&=  \sum \nolimits_{j=1}^p \int_{\mathcal{T}_j} f^{(j)}(t_j)  \sum \nolimits  _{i=1}^p \int_{\mathcal{T}_i} \mathbb{E} \left(X^{(i)}(s_i)  X^{(j)}(t_j)\right) f^{(i)}(s_i)  \mathrm{d}  s_i~  \mathrm{d}  t_j \\
&=  \mathbb{E}  \left(  \sum \nolimits_{j=1}^p \int_{\mathcal{T}_j}  f^{(j)}(t_j)  X^{(j)}(t_j)  \mathrm{d}  t_j \right)^ 2 \geq 0.
\end{align*}

\item $\Gamma$ is compact: Let $\mathcal{B} := \{f \in \mathcal{H}: \mynorm{f}^2 \leq B \}$ be a bounded family in $\mathcal{H}$ for some constant $0 < B < \infty$.
Clearly, $\norm{f^{(j)}}^2 \leq B$ for all $j = 1 , \ldots ,  p$.
Define the image of $\mathcal{B}$ under $\Gamma$ by $\mathcal{Z} = \Gamma \mathcal{B} = \{ g \in \mathcal{H} : \exists~ f \in \mathcal{B} \text{ such that } g = \Gamma f\}$, which has the following properties:
\begin{itemize}
\item $\mathcal{Z}$ is uniformly bounded:
Let $g \in \mathcal{Z}$ and $\boldsymbol{t} \in \mathcal{T}$.
Define $K:= \max \nolimits_{i,j = 1 , \ldots ,  p} K_{ij}$ with $K_{ij}$ as in \eqref{eq:Kijconstants}. Then
\begin{align*}
\vecnorm{g(\boldsymbol{t})}^2 
&\leq \sum \nolimits_{j=1}^p  \left( \sum \nolimits_{i =1}^p  \int_{\mathcal{T}_i} \left \vert C_{ij}(s_i, t_j) f^{(i)}(s_i) \right \vert \mathrm{d}  s_i  \right) ^2\\
&\stackrel{\text{H\"older}}{\leq} \sum \nolimits_{j=1}^p \left( \sum \nolimits_{i =1}^p    \left[ \int_{\mathcal{T}_i}  C_{ij}(s_i, t_j)^2  \mathrm{d}  s_i \right]^{1/2}
\left[  \int_{\mathcal{T}_i}   f^{(i)}(s_i)^2  \mathrm{d}  s_i \right ]^{1/2}  \right) ^2\\
& \leq \sum \nolimits_{j=1}^p  \left(  \sum \nolimits_{i =1}^p  K ^{1/2} B ^{1/2} \right) ^2
 = p^3 K B < \infty.
\end{align*}
\item $\mathcal{Z}$ is equicontinuous: 
Denote by $\lambda(\mathcal{T}_j)$ the Lebesgue measure of $\mathcal{T}_j$ and let $T = \max \nolimits_{j = 1 , \ldots ,  p} \lambda(\mathcal{T}_j)$. For $\varepsilon  > 0$, define $\tilde{\varepsilon }:= \frac{\varepsilon }{p (pTB)^{1/2}}$. By the  continuity assumption for $C_{ij}(s_i, \cdot)$, there exist $ \delta_{ij} > 0$ such that
\begin{equation} \label{eq:CovCont}
\vecnorm{ t_j - t^\ast_j } < \delta_{ij}
\quad \Rightarrow \quad
\vert C_{ij}(s_i, t_j)  - C_{ij}(s_i, t_j^\ast) \vert < \tilde \varepsilon ,  \quad \forall~  s_i \in \mathcal{T}_i.
\end{equation}
for all $i,j = 1 , \ldots ,  p$. Set $\delta := \min \nolimits_{i,j = 1, \ldots, p} \delta_{ij}$ and let $\vecnorm{\boldsymbol{t} - \boldsymbol{t^\ast}}_{\mathcal{T}} < \delta$. Clearly, $\vecnorm{ t_j - t_j^ \ast } <\delta$ for all $j = 1, \ldots, p$ and for $g \in \mathcal{Z}$ it holds
\begin{align*}
&\vecnorm{g(\boldsymbol{t}) - g(\boldsymbol{t^\ast})}^2
 \leq  \sum \nolimits_{j=1}^p \left( \sum \nolimits_{i =1}^p \left \vert \int_{\mathcal{T}_i}   \left( C_{ij}(s_i, t_j)  -   C_{ij}(s_i, t_j^\ast)  \right) f^{(i)}(s_i) \mathrm{d}  s_i  \right \vert  \right)    ^2 \\
& \quad ~ ~ \stackrel{\text{H\"older}}{ \leq}  \sum \nolimits_{j=1}^p \left( \sum \nolimits_{i =1}^p 
\left[
\int_{\mathcal{T}_i}   \left( C_{ij}(s_i, t_j)  -  C_{ij}(s_i,t_j^\ast)  \right)^2  \mathrm{d}  s_i 
 \right]^{1/2}
\left[ \int_{\mathcal{T}_i}   f^{(i)}(s_i)^2 \mathrm{d}  s_i \right]^{1/2}
   \right)  ^2 \\
& \qquad \stackrel{\eqref{eq:CovCont}}{<}  \sum \nolimits_{j=1}^p \left( \sum \nolimits_{i =1}^p 
\left(
\int_{\mathcal{T}_i}  \tilde\varepsilon ^2 \mathrm{d}  s_i 
 \right)^{1/2}
 \norm{f^{(i)}}
   \right)  ^2 
    \leq p^3 T B \tilde\varepsilon ^2 = \varepsilon ^2.
\end{align*}
\end{itemize}
By the Theorem of Arzel\`{a}-Ascoli \citep[Thm. I.28. and related  notes for Chapter I]{ReedSimon:2005}, for each sequence $\{ f_n\}_{n \in \mathbb N }$ in $\mathcal{B}$ there exists a convergent subsequence $\{ g_{n(i)}  = \Gamma f_{n(i)}\}_{i \in \mathbb N }$ of the corresponding sequence  $\{g_n\}_{n \in \mathbb N }$ in $\mathcal{Z}$, which implies that $\Gamma$ is a compact operator \citep[cf.][Chapter VI.5.]{ReedSimon:2005}.
\end{enumerate}
\end{proof}

\begin{lemma} \label{lemma:eFunCont}
For fixed $m \in \mathbb N $ and $j  \in \{1 , \ldots ,  p\}$, the $j$-th element $\psi_m^{(j)}$ of the eigenfunction $\psi_m$ is continuous, if $\nu_m > 0$ and  $C_{ij}$ is uniformly continuous as in Prop.~\ref{prop:GammaProperties}.
\end{lemma}

\begin{proof}
Let $\varepsilon  > 0$ and $\tilde \varepsilon := \varepsilon \nu_m  \left(2 \sum_{j = 1}^p \lambda(\mathcal{T}_j)^{1/2} \right)^{-1}$. By the uniform continuity assumption for $C_{ij}$, there exist $\delta_{ij} > 0$ such that for all $i = 1, \ldots, p$ \eqref{eq:CovCont} holds.
Let $\delta_j = \min \nolimits_{i = 1, \ldots, p} \delta_{ij}$ and $\vecnorm{t_j - t_j^\ast} < \delta_j$. Then, as $\nu_m > 0$ and $ \norm{\psi_m^{(i)}} \leq \mynorm{\psi_m} = 1$,
\begin{align*}
& \left| \psi_m^{(j)}(t_j)  - \psi_m^{(j)}(t_j^\ast)\right|
 = \frac{1}{\nu_m}~  \left| \sum \nolimits _{i = 1}^p  \int_{\mathcal{T}_i}\left[C_{ij}(s_i, t_j) - C_{ij}(s_i, t_j^\ast) \right] \psi_m^{(i)}(s_i) \mathrm{d}  s_i\right| \\
& \qquad \leq \frac{1}{\nu_m}  \tilde \varepsilon~  \sum \nolimits _{i = 1}^p  \int_{\mathcal{T}_i}   \left|  \psi_m^{(i)}(s_i) \right| \mathrm{d}  s_i  \stackrel{\text{H\"older}}{\leq}  \frac{\tilde \varepsilon }{\nu_m}~  \sum \nolimits _{i = 1}^p \norm{\psi_m^{(i)}} \lambda(\mathcal{T}_i)^{1/2} \\
& \qquad \leq  \frac{\tilde \varepsilon }{\nu_m}~  \sum \nolimits _{i = 1}^p \lambda(\mathcal{T}_i)^{1/2} 
 = \frac{\varepsilon}{2} < \varepsilon.
\end{align*}
\end{proof}

\begin{proof}[Proof of Prop.~\ref{prop:Mercer}]
\label{proof:Mercer}
The proof follows the idea in \citet[Chapter VI.4.]{Werner:2011} for the proof of Mercer's Theorem in the univariate case.
From the Spectral Theorem for compact self-adjoint operators \citep[Thm. VI.3.2.]{Werner:2011}, it is known that
\[\Gamma f = \sum \nolimits_{m=1}^\infty \nu_m \myscal{f}{\psi_m} \psi_m \quad \forall~ f \in \mathcal{H}.\]
For $M \in \mathbb N $ define
\[\Gamma_M f := \sum \nolimits_{m=1}^M \nu_m \myscal{f}{\psi_m} \psi_m \quad \forall~ f \in \mathcal{H}.\]
Then for all $f \in \mathcal{H}: ~ \myscal{\Gamma f}{f} - \myscal{\Gamma_M f}{f} = \sum \nolimits_{m=M+1}^\infty \nu_m \myscal{f}{\psi_m} ^2 \geq 0$. 
Let $j \in \{1 , \ldots ,  p\}$ and $t^\ast \in \mathcal{T}_j$. Define $f = \left( 0 , \ldots ,  0, f^{(j)}, 0 , \ldots ,  0\right) $ with 
$f^{(j)} = \lambda \left(B_{t^\ast}\left(\frac{1}{n}\right)\right)^{-1}  \boldsymbol{1}_{B_{t^\ast}(\frac{1}{n})} $ 
  for some $n \in \mathbb N $, where $B_{t^\ast}(\frac{1}{n})$ is a closed ball in $\mathcal{T}_j$ with center $t^\ast$ and radius $\frac{1}{n}$, $\lambda(\cdot)$ denotes the Lebesgue measure and $\boldsymbol{1}$ is the indicator function. Clearly, $f ^{(j)} \in L^2(\mathcal{T}_j)$ and $f \in \mathcal{H}$. Therefore
\begin{align*}
0 
& \leq \myscal{\Gamma f}{f} - \myscal{\Gamma_M f}{f} \\
& =  \lambda \left(B_{t^\ast}\left(\tfrac{1}{n}\right)\right) ^{-2} \int_{B_{t^\ast}(\frac{1}{n})} \int_{B_{t^\ast}(\frac{1}{n})} C_{jj}(s_j, t_j)  - \sum \nolimits_{m=1}^M \nu_m  \psi_m^{(j)}(s_j) \psi_m^{(j)}(t_j) \mathrm{d}  s_j~  \mathrm{d}  t_j \\
& \to C_{jj}(t^\ast, t^\ast) - \sum \nolimits_{m=1}^M \nu_m   \psi_m^{(j)}(t^\ast) \psi_m^{(j)}(t^\ast) \quad \text{for} ~ n \to \infty
\end{align*}
by the Lebesgue Differentiation Theorem \citep[Thm. 7.10.]{Rudin:1987}. As $t^\ast$ was arbitrary in $\mathcal{T}_j$, this implies that for all $M \in \mathbb N $
\[\sum \nolimits_{m=1}^M \nu_m   \psi_m^{(j)}(t)^2  \leq  C_{jj}( t, t)  \leq \left\Vert  C_{jj} \right\Vert_\infty  < \infty \quad \forall~  t \in \mathcal{T}_j,\]
since $C_{jj}$ is continuous and $\mathcal{T}_j$ is compact, implying that  $\left\Vert  C_{jj}  \right\Vert_\infty:= \sup_{t \in \mathcal{T}_j} \left|C_{jj}(t,t)\right|$ is finite. Using H\"older's inequality
\[
\sum \nolimits _{m=1}^\infty \left | \nu_m \psi_m^{(j)}(s) \psi_m^{(j)}(t)  \right|
 \leq C_{jj}(s,s)^{1/2} C_{jj}(t,t)^{1/2}
 < \infty,
\]
i.e.  the series $\tilde C_j (s,t) := \sum \nolimits _{m=1}^\infty  \nu_m \psi_m^{(j)}(s) \psi_m^{(j)}(t) $ is absolutely convergent for all $s,t \in \mathcal{T}_j$.
In the following, assume  $t \in \mathcal{T}_j$ to be fixed.
For $\varepsilon  > 0$ choose $M \in \mathbb N $ such that 
$ \sum \nolimits _{m = M+1}^\infty \nu_m   \psi_m^{(j)}(t)^2 < \varepsilon ^2$. Then, again by H\"older's inequality
\begin{align}
\sum \nolimits _{m=M+1}^\infty \left | \nu_m \psi_m^{(j)}(s) \psi_m^{(j)}(t)  \right| 
& \leq  C_{jj}(s,s)^{1/2} \cdot \varepsilon 
\leq  \left\Vert  C_{jj} \right\Vert_\infty^{1/2}\cdot \varepsilon.
\label{eq:MercerUnifConv}
\end{align}
The upper bound in \eqref{eq:MercerUnifConv} does not depend on $s$, hence $\tilde C_j (s,t)$ converges uniformly for fixed $t$.
As the eigenfunctions $\psi_m^{(j)}(s)$ are continuous in $s$ for all $m \in \mathbb N $ (Lemma~\ref{lemma:eFunCont}), $\tilde C_j(s,t)$ is also continuous in $s$ \citep[Uniform Limit Theorem,][Thm. 21.6.]{Munkres:2000}.
Define
\[h_j(s) := C_{jj}(s,t) - \tilde C_j(s,t), \quad s \in \mathcal{T}_j.\]
Let now $g^{(j)} \in L^2(\mathcal{T}_j)$ and define $g := \left( 0 , \ldots ,  0, g^{(j)}, 0 , \ldots ,  0 \right)$, which is clearly in $\mathcal{H}$.
Therefore
\begin{align*}
& \int_{\mathcal{T}_j} h_j(s) g^{(j)}(s) \mathrm{d}  s
 = \sum \nolimits _{i = 1}^p \int_{\mathcal{T}_i} C_{ij}(s_i,t) g^{(i)}(s_i) \mathrm{d}  s_i - \sum \nolimits _{m=1}^\infty \nu_m  \sum \nolimits _{i = 1}^p \int_{\mathcal{T}_i}  \psi_m^{(i)}(s_i)   g^{(i)}(s_i) \mathrm{d}  s_i~\psi_m^{(j)}(t)  \\
& \qquad = (\Gamma g)^{(j)}(t) -  \sum \nolimits _{m=1}^\infty \nu_m   \myscal{\psi_m}{g} \psi_m^{(j)}(t) \\
& \qquad = \sum \nolimits_{m=1}^\infty \nu_m \myscal{g}{\psi_m} \psi_m^{(j)}(t) -  \sum \nolimits _{m=1}^\infty \nu_m   \myscal{\psi_m}{g} \psi_m^{(j)}(t)  = 0
\end{align*}
according to the Spectral Theorem.
Choosing $g^{(j)} = h_j$ implies $h_j(s) = 0$ for all $s \in \mathcal{T}_j$, as $h_j$ is continuous in $s$.
Therefore,
\[\tilde C_j(s,t) = \sum \nolimits _{m = 1}^\infty \nu_m \psi_m^{(j)}(s) \psi_m^{(j)}(t) = C_{jj}(s,t) \quad \forall~ s \in \mathcal{T}_j.\]
By Dini's Theorem \citep[Thm. VI.4.6.]{Werner:2011} the series $C_{jj}(t,t) = \sum \nolimits _{m = 1}^\infty \nu_m \psi_m^{(j)}(t)^2$ converges uniformly. Hence, $M$ can be chosen independent of $t$ in \eqref{eq:MercerUnifConv}. This implies that $\tilde C_j(s,t)$ converges absolutely and uniformly to $C_{jj}(s,t)$ for all $s,t \in \mathcal{T}_j$.
\end{proof}

\begin{proof}[Proof of Prop.~\ref{prop:KarhunenMultivariate}] \label{proof:KarhunenMultivariate}
By the Hilbert-Schmidt Theorem \citep[Thm. VI.16.]{ReedSimon:2005}, the (deterministic) eigenfunctions of $\Gamma$ form an orthonormal basis of $\mathcal{H}$, i.e. $X$ can be written in the form
$X(\boldsymbol{t}) = \sum   \nolimits_{m=1}^\infty \rho_m \psi_m(\boldsymbol{t}),~ \boldsymbol{t} \in \mathcal{T}$ with random variables $\rho_m = \myscal{X}{\psi_m}$. Hence for $m,n \in \mathbb N $
\begin{enumerate}
\item 
$\mathbb{E} \left(\rho_m \right)
= \displaystyle \sum \nolimits_{j = 1}^p \displaystyle \int_{\mathcal{T}_j} \mathbb{E}  \left(X^{(j)}(t_j) \right) \psi_m^{(j)}(t_j) \mathrm{d}  t_j
= 0,$
since $\mathbb{E}\left( X^{(j)}(t_j) \right) = 0$ for all $t_j \in~ \mathcal{T}_j$, $j = 1 , \ldots ,  p$ by assumption.
 \item $
\begin{aligned}[t]
& \operatorname{Cov}(\rho_m, \rho_n)
 = \mathbb{E} \left(
	\sum \nolimits_{i = 1}^p \int_{\mathcal{T}_i} X^{(i)}(s_i) \psi_m^{(i)}(s_i) \mathrm{d}  s_i
	\cdot \sum \nolimits_{j = 1}^p \int_{\mathcal{T}_j} X^{(j)}(t_j) \psi_n^{(j)}(t_j) \mathrm{d}  t_j
 \right) \\
 & \qquad = 
\sum   \nolimits_{j = 1}^p   \int_{\mathcal{T}_j}
\sum  \nolimits_{i = 1}^p    \int_{\mathcal{T}_i}
 C_{ij}(s_i, t_j)  \psi_m^{(i)}(s_i)  \mathrm{d}  s_i~ \psi_n^{(j)}(t_j) \mathrm{d}  t_j \\
     & 
 \qquad = \sum  \nolimits_{j = 1}^p    \int_{\mathcal{T}_j} \nu_m \psi_m^{(j)}(t_j)  \psi_n^{(j)}(t_j) \mathrm{d}  t_j = \nu_m \delta_{mn}.
\end{aligned}$

\item Let $X_{\lceil M \rceil}(\boldsymbol{t}):=  \displaystyle \sum  \nolimits_{m=1}^M  \rho_m \psi_m(\boldsymbol{t}) =  \displaystyle \sum  \nolimits_{m=1}^M
 \left[ \displaystyle \sum \nolimits_{i=1}^p \displaystyle \int_{\mathcal{T}_i} X^{(i)}(s_i) \psi_m^{(i)}(s_i)\mathrm{d}  s_i \right]  \psi_m(\boldsymbol{t})$ for $\boldsymbol{t} \in~\mathcal{T}$ be the truncated Karhunen-Lo\`{e}ve representation of $X$. Then
\begin{align*}
&\mathbb{E} \left( \vecnorm{X(\boldsymbol{t}) - X_{\lceil M \rceil}(\boldsymbol{t})}^2 \right) 
=  \sum  \nolimits_{j=1}^p \left[ \mathbb{E}  \left( X^{(j)}(t_j)^2 \right) - 2 \mathbb{E} \left( X^{(j)}(t_j)S_M^{(j)}(t_j) \right) + \mathbb{E} \left( S_M^{(j)}(t_j)^2 \right) \right]\\
  & \qquad =  \sum  \nolimits_{j=1}^p \Big[ C_{jj}(t_j, t_j)
 - 2 \sum  \nolimits_{m=1}^M \sum \nolimits_{i=1}^p \int_{\mathcal{T}_i} C_{ij}(s_i, t_j) \psi_m^{(i)}(s_i)\mathrm{d}  s_i~ \psi_m^{(j)}(t_j)  \\
 & \qquad ~ ~  + \sum   \nolimits_{m=1}^M    \sum  \nolimits_{n=1}^M   \sum \nolimits_{i=1}^p 
  \int_{\mathcal{T}_i}
  \sum  \nolimits_{k=1}^p 
  \int_{\mathcal{T}_k} 
 C_{ki}(u_k, s_i) 
 \psi_n^{(k)}(u_k)\mathrm{d}  u_k~
 \psi_m^{(i)}(s_i)\mathrm{d}  s_i~ \psi_m^{(j)}(t_j) \psi_n^{(j)}(t_j)  \Big]\\
   & \qquad =  \sum \nolimits_{j=1}^p \Big[ C_{jj}(t_j, t_j)
 - 2 \sum \nolimits_{m=1}^M \nu_m \psi_m^{(j)}(t_j)  \psi_m^{(j)}(t_j)  \\
 & \qquad ~ ~  + \sum   \nolimits_{m=1}^M    \sum  \nolimits_{n=1}^M   \sum \nolimits_{i=1}^p
  \int_{\mathcal{T}_i} 
 \nu_n \psi_n^{(i)}(s_i)
 \psi_m^{(i)}(s_i)\mathrm{d}  s_i  \psi_m^{(j)}(t_j)\psi_n^{(j)}(t_j) \Big] \\
     & \qquad=  \sum   \nolimits_{j=1}^p \left[ C_{jj}(t_j, t_j)
 -  \sum  \nolimits_{m=1}^M \nu_m \psi_m^{(j)}(t_j) ^2 \right] \to 0 \quad \text{for} ~ M \to \infty
\end{align*}
uniformly for $\boldsymbol{t} \in \mathcal{T}$ by Prop.~\ref{prop:Mercer}.
\end{enumerate}
\end{proof}

\begin{proof}[Proof of Prop.~\ref{prop:multi_single_EV}] \label{proof:multi_single_EV}
\leavevmode
\begin{enumerate}
\item

Let $X$ have a finite Karhunen-Lo\`{e}ve representation \eqref{eq:truncKH}. Then, each element is given by $X^{(j)} = \sum \nolimits _{m =1}^M \rho_m \psi_m^{(j)}$. For $t \in \mathcal{T}_j$
\begin{align}
\left(\Gamma^{(j)}  \tilde \phi^{(j)} \right)(t)
&= \int_{\mathcal{T}_j} \operatorname{Cov}(X^{(j)}(s), X^{(j)}(t)) \tilde \phi^{(j)}(s) \mathrm{d}  s  \nonumber \\
& = \int_{\mathcal{T}_j} \sum \nolimits  _{m=1}^M \nu_m \psi_m^{(j)}(s) \psi_m^{(j)}(t) \tilde \phi^{(j)}(s) \mathrm{d}  s
  \stackrel{!}{=} \lambda^{(j)} \tilde \phi^{(j)}(t),
 \label{eq:eigeneqSingle}
\end{align}
which is a homogenous Fredholm integral equation of the second kind with separable kernel function $K(s,t) = \sum \nolimits  _{m=1}^M \nu_m \psi_m^{(j)}(s) \psi_m^{(j)}(t) = \sum \nolimits  _{m=1}^M a_m(s) b_m(t)$ with continuous functions $a_m(s) = \nu_m^{1/2} \psi_m^{(j)}(s),~ b_m(t) = \nu_m^{1/2} \psi_m^{(j)}(t)$ (cf. Lemma \ref{lemma:eFunCont}). Following the argumentation in \citet[Chapter 1.3.]{Zemyan:2012}, \eqref{eq:eigeneqSingle} can be transformed into the matrix eigenequation
\[\boldsymbol{A^{(j)}u} = \lambda^{(j)} \boldsymbol{u}\]
with a symmetric matrix $\boldsymbol{A^{(j)}} \in \mathbb R ^{M \times M}$ given by $A_{mn}^{(j)} = \scal{a_m}{b_n}$. Positivity of $\Gamma^{(j)}$ implies $\lambda^{(j)} \geq 0$ and therefore $\boldsymbol{A^{(j)}}$ is positive semidefinite. Hence it can have at most $M$ strictly positive eigenvalues $\lambda_m ^{(j)}$ associated with eigenvectors $\boldsymbol{u_m^{(j)}} \in \mathbb R ^M$
\[\boldsymbol{A^{(j)} u_m^{(j)}} = \lambda_m^{(j)} \boldsymbol{u_m^{(j)}}.\]
Let $\lambda_1^{(j)} \geq \ldots \geq \lambda_{M_j}^{(j)} > 0,~  M_j \leq M$ be the non-zero eigenvalues of $\boldsymbol{A^{(j)}}$. They are at the same time the only non-zero eigenvalues of $\Gamma^{(j)}$. The associated eigenfunctions $\tilde \phi_m^{(j)}$ of $\Gamma^{(j)}$ are parametrized by the eigenvectors $\boldsymbol{u_m^{(j)}}$ via \citep[Chapter 1.3.]{Zemyan:2012}
\[\Gamma^{(j)} \tilde \phi_m^{(j)} = \sum \nolimits _{n=1}^M \nu_n^{1/2} \psi_n^{(j)}  [\boldsymbol{u_m^{(j)}}]_n = \lambda_m^{(j)} \tilde \phi^{(j)}_m
\quad
\Leftrightarrow
\quad
 \tilde \phi^{(j)}_m
= \left({\lambda_m^{(j)}}\right)^{-1} \sum \nolimits _{n=1}^M \nu_n^{1/2} \psi_n^{(j)} [\boldsymbol{u_m^{(j)}}]_n
.\]
Since
\[
\scal{\tilde \phi_m^{(j)}}{\tilde \phi_n^{(j)}}
 =\left(\lambda_m^{(j)}\lambda_n^{(j)}\right)^{-1} \boldsymbol{ {u_m^{(j)}} ^\top  A^{(j)} u_n ^{(j)}}
 =\left(\lambda_m^{(j)} \right)^{-1} \delta_{mn},
\]
orthonormal eigenfunctions $\phi_m^{(j)}$ are given by
\[ \phi_m^{(j)} = {\lambda_m^{(j)}}^{1/2} \tilde \phi_m ^{(j)} = \left(\lambda_m^{(j)} \right)^{-1/2} \sum \nolimits _{n=1}^M \nu_n^{1/2} \psi_n^{(j)} [\boldsymbol{u_m^{(j)}}]_n.\]
Therefore, $X^{(j)}$ has a finite Karhunen-Lo\`{e}ve representation
$X^{(j)} = \sum \nolimits _{m=1}^{M_j} \xi_m^{(j)} \phi_m^{(j)}$
with scores 
\[\xi_m^{(j)} 
= \scal{X^{(j)}}{\phi_m^{(j)}}
= \left(\lambda_m^{(j)}\right)^{-1/2} \sum \nolimits _{n = 1}^M \nu_n^{1/2} \left[\boldsymbol{u_m^{(j)}} \right]_n \sum \nolimits _{k = 1}^M \rho_k \scal{\psi_n^{(j)}}{\psi_k^{(j)}}
\]
and $\mathbb{E}(\xi_m^{(j)}) = 0,~ \operatorname{Cov}(\xi_m^{(j)}, \xi_n^{(j)}) = \lambda_m^{(j)} \delta_{mn}$.

\item
Assume the functional covariates $X^{(1)} , \ldots ,  X^{(p)}$ do each have a finite Karhunen-Lo\`{e}ve representation, i.e. for each $j = 1, \ldots, p: ~X^{(j)} = \sum \nolimits _{m = 1}^{M_j} \xi_m^{(j)} \phi_m^{(j)}$. Let $\Gamma \psi = \nu \psi$. Then for all $j = 1 , \ldots ,  p$ and $t_j \in \mathcal{T}_j$
\begin{align*}
\left( \Gamma \psi \right)^{(j)}(t_j)  
&= \sum \nolimits _{k=1}^p \int_{\mathcal{T}_k} \operatorname{Cov} \left(X^{(k)}(s_k), X^{(j)}(t_j) \right) \psi ^{(k)}(s_k) \mathrm{d}  s_k \\
&= \sum \nolimits _{k=1}^p  \sum \nolimits _{l = 1}^{M_j} \sum \nolimits _{n = 1}^{M_k}   \underbrace{\operatorname{Cov} \left( \xi_l^{(j)}, \xi_n^{(k)} \right)}_{=: Z_{ln}^{(jk)}}  \phi_l^{(j)}(t_j)   \underbrace{\int_{\mathcal{T}_k}  \phi_n^{(k)}(s_k)  \psi ^{(k)}(s_k) \mathrm{d}  s_k}_{=: c_n^{(k)}} 
\stackrel{!}{=} \nu \psi ^{(j)}(t_j).
\end{align*}
With a similar argumentation as in \citet[Chapter 1.3.]{Zemyan:2012} it holds
\begin{align}
& \sum \nolimits _{k=1}^p  \sum \nolimits _{l = 1}^{M_j}  \sum \nolimits _{n = 1}^{M_k}  Z_{ln}^{(jk)}  \phi_l^{(j)}(t_j)  c_n^{(k)} = \nu \psi ^{(j)}(t_j) \label{eq:eVecsMulti}\\
\Rightarrow &
\int_{\mathcal{T}_j} \phi_m^{(j)}(t_j) \cdot \sum \nolimits _{k=1}^p \sum \nolimits _{l = 1}^{M_j}  \sum \nolimits _{n = 1}^{M_k}   Z_{ln}^{(jk)}  \phi_l^{(j)}(t_j)  c_n^{(k)} \mathrm{d}  t_j = \int_{\mathcal{T}_j} \phi_m^{(j)}(t_j) \cdot \nu \psi ^{(j)}(t_j) \mathrm{d}  t_j \nonumber \\
\Leftrightarrow &
\sum \nolimits _{k=1}^p  \sum \nolimits _{n = 1}^{M_k}   Z_{mn}^{(jk)}  c_n^{(k)} = \nu c_m^{(j)}  \nonumber
\end{align}
for $m  = 1 , \ldots ,  M_j$ due to orthonormality of $\phi_m^{(j)}$.  Since $m$ and $j$  were arbitrarily chosen, this is equivalent to 
\[
\underbrace{
\begin{pmatrix}
\boldsymbol{Z^{(11)} }& \ldots  & \boldsymbol{Z^{(1p)}} \\
\vdots& \ddots & \vdots \\
\boldsymbol{Z^{(p1)}}& \ldots & \boldsymbol{Z^{(pp)}}
\end{pmatrix}
}_{=:\boldsymbol{Z}}
\underbrace{
\begin{pmatrix}
\boldsymbol{c^{(1)}} \\
\vdots \\
\boldsymbol{c^{(p)}}
\end{pmatrix}}_{=:\boldsymbol{c}}
=
\nu
\begin{pmatrix}
\boldsymbol{c^{(1)}} \\
\vdots \\
\boldsymbol{c^{(p)}}
\end{pmatrix}
\] 
with matrices $\boldsymbol{Z^{(jk)}} \in \mathbb R ^{M_j \times M_k}$ and $\boldsymbol{c^{(j)}} \in \mathbb R ^{M_j}$. The last equation is again an eigenequation for the symmetric (and since $\nu \geq 0$) positive semidefinite block matrix $\boldsymbol{Z} \in \mathbb R ^{M_+ \times M_+}$.
Let $\nu_1 \geq \ldots \geq \nu_M > 0,~ M \leq M_+$ be the non-zero eigenvalues of $\boldsymbol{Z}$. These are also the only non-zero eigenvalues of $\Gamma$ and the elements $\psi_m^{(j)}$ of the associated eigenfunctions $\psi_m$ are parametrized by the (orthonormal) eigenvectors $\boldsymbol{c_1} , \ldots ,  \boldsymbol{c_M}$ associated with $\nu_1 , \ldots ,  \nu_M$:
\[
 \psi_m^{(j)}(t_j)
\stackrel{\eqref{eq:eVecsMulti}}{=} \frac{1}{\nu_m}
\sum \nolimits _{k=1}^p  \sum \nolimits _{l = 1}^{M_j} \sum \nolimits _{n = 1}^{M_k}   Z_{ln}^{(jk)}   [\boldsymbol{c_m}]^{(k)}_n  \phi_l^{(j)}(t_j) 
= \sum \nolimits _{l = 1}^{M_j} [\boldsymbol{c_m}]_l^{(j)} \phi_l^{(j)}(t_j), \quad t_j \in \mathcal{T}_j
\]
for $m = 1 , \ldots ,  M,~ j = 1 , \ldots ,  p.$ The eigenfunctions form an orthonormal system with respect to $\myscal{\cdot}{\cdot}$:
\[
\myscal{\psi_m}{\psi_n}
= \sum \nolimits _{j = 1}^p \scal{\psi_m^{(j)}}{\psi_n^{(j)}} 
= \sum \nolimits _{j = 1}^p \sum \nolimits _{l = 1}^{M_j} [\boldsymbol{c_m}]^{(j)}_{l} [\boldsymbol{c_n}]^{(j)}_{l}  
= \boldsymbol{c_n} ^\top  \boldsymbol{c_m} = \delta_{mn}.
\]
The Karhunen-Lo\`{e}ve decomposition of $X$ is therefore given by $X = \sum \nolimits  _{m=1}^M \rho_m \psi_m$ with scores
\[
\rho_m
 = \myscal{X}{\psi_m} 
= \sum \nolimits _{j = 1}^p \sum \nolimits _{n = 1}^{M_j}\left[\boldsymbol{c_m}\right]_n^{(j)}  \xi_n^{(j)}
\]
and $\mathbb{E}(\rho_m)=0,~\operatorname{Cov}(\rho_m, \rho_n) = \nu_m \delta_{mn},~ m = 1 , \ldots ,  M \leq M_+$.
\end{enumerate}
\end{proof}

\begin{proof}[Proof of Prop.~\ref{prop:asymptBias}]  \label{proof:asymptBias}

For $f \in \mathcal{H},~\boldsymbol{t} \in \mathcal{T}$ and $j = 1 , \ldots,  p$, the covariance operator $\Gamma^{[M]}$ associated with $X^{[M]}$ is given by
\[(\Gamma^{[M]} f)^{(j)}(t_j) = \sum_{i = 1}^p \int_{\mathcal{T}_i} \operatorname{Cov}(X^{[M](i)}(s_i),  X^{[M](j)}(t_j)) f^{(i)}(s_i) \mathrm{d} s_i. \]
In the following, use $C^{[M]}_{ij}(s_i, t_j) := \operatorname{Cov}(X^{[M](i)}(s_i),  X^{[M](j)}(t_j))$ as short notation for the covariance functions (cf. the definition of $C_{ij}$ in \eqref{eq:CovSymm} in the paper).
Next, recall some well-known results for univariate functional data: By Mercer's Theorem \citep{Mercer:1909}
\begin{equation}
C^{[M]}_{jj}(t_j, t_j) =  \sum_{m = 1}^{M_j} \lambda_m^{(j)} \phi_m^{(j)}(t_j)^2 \nearrow \sum_{m = 1}^\infty \lambda_m^{(j)}   \phi_m^{(j)}(t_j)^2 = C_{jj}(t_j, t_j) ~ \text{for}~M_j \to \infty,~t_j \in \mathcal{T}_j.
\label{eq:varMtoVar}
\end{equation} 
The univariate Karhunen-Lo{\`e}ve Theorem \citep[e.g.][Thm 1.5.]{Bosq:2000} states that
\[\mathbb{E}\left[ \left|X^{(j)}(t_j) - \sum\nolimits_{m = 1}^{M_j} \xi_m^{(j)} \phi_m^{(j)}(t_j) \right|^2 \right]\] converges uniformly to 0 for $t_j \in \mathcal{T}_j$ and $M_j \to \infty$. As both $X^{(j)}$ and $X^{[M](j)}$ have zero mean ($X^{(j)}$ by assumption and  $X^{[M](j)}$ since the scores $\xi_m^{(j)}$ have zero mean), this implies 
\begin{equation}
\operatorname{Var}\left(X^{(j)}(t_j) - X^{[M](j)}(t_j)\right) \to 0 \quad \text{for}~M_j \to \infty.
\label{eq:varDiffto0}
\end{equation}
With the assumptions of Prop.~\ref{prop:GammaProperties} it further holds (cf. proof of Prop.~\ref{prop:Mercer}) that
\begin{equation}
 \operatorname{Var}(X^{(j)}(t_j))= C_{jj}(t_j,t_j) \leq \left \Vert C_{jj} \right \Vert_\infty< \infty.
\label{eq:varFinite}
\end{equation}
For fixed $s_i \in \mathcal{T}_i,~t_j \in \mathcal{T}_j$ with $i,j = 1,\ldots,p$, these three properties give
 \begin{align*}
&  \left|C_{ij}(s_i, t_j) - C^{[M]}_{ij}(s_i, t_j) \right| \\
& \quad \leq  \left|\operatorname{Cov}(X^{(i)}(s_i) - X^{[M](i)}(s_i),  X^{(j)}(t_j)) \right| + \left|\operatorname{Cov}(X^{[M](i)}(s_i), X^{(j)}(t_j) -  X^{[M](j)}(t_j)) \right| \\
   & \quad \stackrel{\eqref{eq:varMtoVar}}{\leq}  \underbrace{\operatorname{Var}(X^{(i)}(s_i) - X^{[M](i)}(s_i))^{1/2}}_{\to 0~\text{for}~M_i \to \infty ~ \eqref{eq:varDiffto0}} \underbrace{C_{jj}(t_j, t_j)^{1/2}}_{ < \infty  ~ \eqref{eq:varFinite}} + \underbrace{C_{ii}(s_i, s_i))^{1/2}}_{< \infty ~ \eqref{eq:varFinite} }  \underbrace{\operatorname{Var}( X^{(j)}(t_j) -  X^{[M](j)}(t_j))^{1/2}}_{\to 0 ~\text{for}~M_j \to \infty ~ \eqref{eq:varDiffto0}}
 \end{align*}
Hence it holds that
\begin{equation}
C^{[M]}_{ij}(s_i, t_j) \to  C_{ij}(s_i, t_j) \quad \text{for}~M_i, M_j \to \infty.
 \label{eq:convCov}
\end{equation}
The main proof is now in three steps:
\begin{enumerate}
\item $\Gamma^{[M]}$ converges in norm to $\Gamma$ for $M_1, \ldots, M_p \to \infty$: Let $\mynorm{\cdot}_\text{op}$ be the operator norm induced by $\mynorm{\cdot}$. Then
\begin{align}
\mynorm{\Gamma - \Gamma^{[M]}}_\text{op}^2
& = \sup_{\mynorm{f} = 1} \sum_{j = 1}^p \int_{\mathcal{T}_j} \left[ \sum_{i = 1}^p \int_{\mathcal{T}_i} \left(C_{ij}(s_i, t_j) -C_{ij}^{[M]}(s_i, t_j)\right) f^{(i)}(s_i) \mathrm{d} s_i \right]^2 \mathrm{d} t_j \nonumber \\
& \leq \sup_{\mynorm{f} = 1} \sum_{j = 1}^p \int_{\mathcal{T}_j} \left[ \sum_{i = 1}^p \int_{\mathcal{T}_i} \left| (C_{ij}(s_i, t_j) -C_{ij}^{[M]}(s_i, t_j)) f^{(i)}(s_i) \right| \mathrm{d} s_i \right]^2 \mathrm{d} t_j \nonumber \\
& \stackrel{\text{H{\"o}lder}}{\leq}
\sup_{\mynorm{f} = 1} \sum_{j = 1}^p \int_{\mathcal{T}_j} \left[ \sum_{i = 1}^p \norm{ C_{ij}(\cdot,t_j) -C_{ij}^{[M]}(\cdot, t_j)} \norm{f^{(i)}}\right]^2 \mathrm{d} t_j \nonumber \\
& \leq \sum_{j = 1}^p \int_{\mathcal{T}_j} \left[ \sum_{i = 1}^p \norm{ C_{ij}(\cdot, t_j) -C_{ij}^{[M]}(\cdot, t_j)} \right]^2 \mathrm{d} t_j,
\label{eq:normConv}
\end{align}
where the last equality holds since $\norm{f^{(i)}} \leq 1$ for all $f \in \mathcal{H}$ with $\mynorm{f} =1$. The final bound for $\mynorm{\Gamma - \Gamma^{[M]}}_\text{op}^2$ converges to zero for $M_1 , \ldots,  M_p \to \infty$ by \eqref{eq:convCov}, applying the dominated convergence theorem twice: For the norm term in \eqref{eq:normConv} consider
\[
\left|C_{ij}(s_i, t_j) - C^{[M]}_{ij}(s_i, t_j) \right|
 \leq \left|C_{ij}(s_i, t_j)\right| + \left| C^{[M]}_{ij}(s_i, t_j) \right|   \stackrel{\eqref{eq:varMtoVar}}{\leq} 2 C_{ii}(s_i, s_i)^{1/2} C_{jj}(t_j, t_j)^{1/2} ,
\]
thus $\left|C_{ij}(s_i, t_j) - C^{[M]}_{ij}(s_i, t_j) \right|^2
  \stackrel{\eqref{eq:varFinite}}{\leq} 4 \left \Vert C_{ii} \right \Vert_\infty\left \Vert C_{jj} \right \Vert_\infty  < \infty$. This upper bound is constant and therefore integrable over $\mathcal{T}_i$, which implies
\[\lim_{M_i \to \infty}
\norm{ C_{ij}(\cdot, t_j) -C_{ij}^{[M]}(\cdot, t_j)}
= \norm{  C_{ij}(\cdot, t_j) - \lim_{M_i \to \infty}C_{ij}^{[M]}(\cdot, t_j)}.\]

For the outer integral in \eqref{eq:normConv} the results of the norm term give
\begin{align*}
\left(\sum_{i = 1}^p  \norm{  C_{ij}(\cdot, t_j) -C_{ij}^{[M]}(\cdot, t_j)} \right)^2 
& \leq  4 \left \Vert C_{jj} \right \Vert_\infty \left(\sum_{i = 1}^p  \left( \left \Vert C_{ii} \right \Vert_\infty \lambda({\mathcal{T}_i})\right)^{1/2} \right)^2 ,
\end{align*}
where $\lambda({\mathcal{T}_i})$ is the Lebesuge measure of $\mathcal{T}_i$ as in the Proof of Prop.~\ref{prop:GammaProperties}. The term on the right hand side is constant and hence integrable over $\mathcal{T}_j$, which gives that for $M_1, \ldots, M_p \to \infty$, the limit $M_j \to \infty$ and the integral over $\mathcal{T}_j$ in \eqref{eq:normConv} can be interchanged.
In summary, these results give that $\Gamma^{[M]}$ converges to $\Gamma$ in norm for $M_1, \ldots, M_p \to \infty$.

\item $\Gamma^{[M]}$ is bounded: Let $f \in \mathcal{H}$. Clearly, $\norm{f^{(i)}} \leq \mynorm{f}$ for all $i = 1, \ldots, p$ and therefore
\begin{align*}
\mynorm{\Gamma^{[M]}f }^2 
& = \sum_{j = 1}^p \int_{\mathcal{T}_j} \left(\sum_{i = 1}^p \int_{\mathcal{T}_i} C_{ij}^{[M]}(s_i, t_j) f^{(i)}(s_i) \mathrm{d} s_i \right)^2 \mathrm{d} t_j \\
&  \stackrel{\text{H{\"o}lder}}{\leq} \sum_{j = 1}^p \int_{\mathcal{T}_j} \left(\sum_{i = 1}^p \left(\int_{\mathcal{T}_i} \left|C_{ij}^{[M]}(s_i, t_j)\right|^2 \mathrm{d} s_i \right)^{1/2} \left(\int_{\mathcal{T}_i} \left|f^{(i)}(s_i)\right|^2 \mathrm{d} s_i \right)^{1/2} \right)^2 \mathrm{d} t_j \\
&  \stackrel{\eqref{eq:varMtoVar} \eqref{eq:varFinite}}{\leq} \sum_{j = 1}^p \int_{\mathcal{T}_j} \left(\sum_{i = 1}^p \left \Vert C_{ii} \right \Vert_\infty^{1/2} \left \Vert C_{jj} \right \Vert_\infty^{1/2} \lambda(\mathcal{T}_i)^{1/2}  \norm{f^{(i)}} \right)^2 \mathrm{d} t_j \\
&  \leq \sum_{j = 1}^p \left \Vert C_{jj} \right \Vert_\infty\lambda(\mathcal{T}_j) \left(\sum_{i = 1}^p  \left \Vert C_{ii} \right \Vert_\infty^{1/2} \lambda(\mathcal{T}_i)^{1/2}   \right)^2  \mynorm{f}^2
 \leq p^3 C^2 T^2 \mynorm{f}^2,
\end{align*}
for $T=\max_{j = 1, \ldots,  p} \lambda(\mathcal{T}_j)$ and $C = \max_{j = 1, \ldots,  p} \left \Vert C_{jj} \right \Vert_\infty$. The value $p^{3/2} C T$ is constant and finite, hence $\Gamma^{[M]}$ is bounded.

\item Convergence results for $\nu_m^{[M]}, \psi_m^{[M]}$ and $\rho_m^{[M]}$: In Prop.~\ref{prop:GammaProperties}, it was shown that $\Gamma$ is compact, which implies that this operator is also bounded \citep[][Chapter VI.5.]{ReedSimon:2005}. As $\Gamma$ and $\Gamma^{[M]}$ are both bounded, norm convergence is equivalent to convergence in the generalized sense \citep[Chapter IV, \S 2.6., Thm. 2.23]{Kato:1980}. This implies that the eigenvalues $\nu_m^{[M]}$ of $\Gamma^{[M]}$ converge to the eigenvalues $\nu_m$ of $\Gamma$ including multiplicity (if the multiplicity is finite, which holds for all nonzero eigenvalues, as $\Gamma$  is compact \citep[cf.][Thm. VI.15.]{ReedSimon:2005}) and the associated total projections converge in norm \citep[Chapter IV, \S 3.5.]{Kato:1980}. If the $m$-th eigenvalue has multiplicity $1$, then the projections on the eigenspaces spanned by $\psi_m$ and $\psi^{[M]}_m$, respectively, are given by
\[P_m f = \myscal{\psi_m}{f}\psi_m,\quad P_m^{[M]} f = \myscal{\psi_m^{[M]} }{f}\psi_m^{[M]},\quad f \in \mathcal{H}.\]
Without loss of generality one may choose the orientation of $\psi_m$ and $\psi_m^{[M]}$ such that $\myscal{\psi_m}{\psi_m^{[M]}} \geq 0$. In this case, as $\psi_m,\psi_m^{[M]}$ both have norm $1$,
\begin{align*}
\mynorm{P_m - P^{[M]}_m}_\text{op} ^2
& \geq \mynorm{\myscal{\psi_m}{\psi_m}\psi_m - \myscal{\psi_m^{[M]} }{\psi_m}\psi_m^{[M]}}^2 
 = \mynorm{\psi_m - \myscal{\psi_m^{[M]} }{\psi_m}\psi_m^{[M]}}^2 \\
& = \mynorm{\psi_m}^2 -2\myscal{\psi_m}{\myscal{\psi_m^{[M]} }{\psi_m}\psi_m^{[M]}}
+ \myscal{\psi_m^{[M]} }{\psi_m}^2 \mynorm{\psi_m^{[M]}}^2\\
& = (1  -\myscal{\psi_m^{[M]} }{\psi_m})(1  +\myscal{\psi_m^{[M]} }{\psi_m}) 
 \geq (1  -\myscal{\psi_m^{[M]} }{\psi_m}) \\
& =  \frac{1}{2} \cdot \mynorm{\psi_m - \psi_m^{[M]}}^2.
\end{align*}
Norm convergence of the total projections hence implies $\mynorm{\psi_m -\psi_m^{[M]}} \to 0$  for $M_1, \ldots, M_p \to \infty$.

To derive convergence of the scores $\rho_m^{[M]}$, note that for $\varepsilon > 0$  and $c:= \left(\frac{2}{\varepsilon} \sum_{j=1}^p \left \Vert C_{jj} \right \Vert_\infty\lambda(\mathcal{T}_j) \right)^{1/2}$
\begin{align}
P\left( \mynorm{X} > c \right)
& \stackrel{\text{Markov}}{\leq} \frac{1}{c^2}\mathbb{E}\left[ \mynorm{X}^2 \right]
\stackrel{\text{Fubini}}{=} \frac{1}{c^2} \sum_{j = 1}^p \int_{\mathcal{T}_j} \mathbb{E} \left[X^{(j)}(t_j)^2 \right] \mathrm{d} t_j \nonumber \\
& =  \frac{1}{c^2} \sum_{j = 1}^p \int_{\mathcal{T}_j} C_{jj}(t_j, t_j)\mathrm{d} t_j 
 \leq   \frac{1}{c^2} \sum_{j = 1}^p \left \Vert C_{jj} \right \Vert_\infty\lambda(\mathcal{T}_j) = \frac{\varepsilon}{2} < \varepsilon,
 \label{eq:XnormBounded}
\end{align}
i.e. the norm of $X$ is bounded in probability. Moreover, 
\begin{align*}
\mathbb{E} \left[\mynorm{X-X^{[M]}}^2 \right]
& \stackrel{\text{Fubini}}{=} \sum_{j = 1}^p \int_{\mathcal{T}_j} \mathbb{E} \left[ \left|X^{(j)}(t_j)-X^{[M](j)}(t_j) \right|^2\right] \mathrm{d} t_j \to 0
\end{align*}
for $M_1 , \ldots, M_p \to \infty$, as the expectation in the integral converges uniformly to $0$ and is thus bounded (by univariate Karhunen-Lo{\`e}ve). As $\mathcal{T}_j$ has finite measure, the overall integral converges to $0$. Hence $\mynorm{X-X^{[M]}}$ converges in the second mean to $0$, thus $\mynorm{X-X^{[M]}} = o_p(1)$.
Finally, this leads to
\begin{align*}
\left|\rho_m - \rho^{[M]}_m \right|
&= 
 \left| \myscal{X}{\psi_m} - \myscal{X^{[M]}}{\psi_m^{[M]}}\right|
  \leq   \left| \myscal{X}{\psi_m-\psi_m^{[M]}} \right|
 +  \left| \myscal{X-X^{[M]}}{\psi_m^{[M]}}\right| \\
& \leq    \mynorm{X}\mynorm{\psi_m-\psi_m^{[M]}} + 
\mynorm{X-X^{[M]}}\mynorm{\psi_m^{[M]}}
 = O_p(1) o(1) + o_p(1) = o_p(1).
\end{align*}
i.e. $\rho^{[M]}_m$ converges in probability to $ \rho_m$. 
\end{enumerate}
\end{proof}

\begin{lemma}{} \label{lemma:covMatrixRate}

Under the assumptions of Prop.~\ref{prop:asymptVar} it holds that
\[\lambda_{\max}(\boldsymbol{Z} - \boldsymbol{\hat Z}) \leq O_p(M_{\max}\max(N^{-1/2} , \Delta_M r_N^\Gamma)).\]
with $\boldsymbol{Z}$ as defined in Prop.~\ref{prop:multi_single_EV}, $ \boldsymbol{\hat Z} = (N-1)^{-1} \boldsymbol{\Xi}^\top \boldsymbol{\Xi}$ as in Section~\ref{sec:estMFPCA}, $M_{\max} = \max_{j = 1 , \ldots,  p}M_j$ and  $\Delta_M:= \max_{j = 1, \ldots, p}{\Delta^{(j)}_{M_j}}$.
\end{lemma}

\begin{proof}[Proof of Lemma 2]

For $j,k = 1, \ldots, p$ and $f \in L^2(\mathcal{T}_j)$ define the bounded operator $\Gamma^{(jk)}\colon L^2(\mathcal{T}_j) \to L^2(\mathcal{T}_k)$ via 
\[
\left(\Gamma^{(jk)}f \right)(t)
:= \int_{\mathcal{T}_j} \operatorname{Cov}\left(X^{(j)}(s), X^{(k)}(t)\right) f(s) \mathrm{d} s
= \int_{\mathcal{T}_j} C_{jk}(s,t) f(s) \mathrm{d} s
\]
Analogously, define  $\hat \Gamma^{(jk)}\colon L^2(\mathcal{T}_j) \to L^2(\mathcal{T}_k)$ by
\[\left(\hat \Gamma^{(jk)}f \right)(t)
:= \int_{\mathcal{T}_j} \widehat{\operatorname{Cov}}\left(X^{(j)}(s), X^{(k)}(t)\right) f(s) \mathrm{d} s
= \int_{\mathcal{T}_j} \hat C_{jk}(s,t) f(s) \mathrm{d} s
\]
with $\hat C_{jk}(s,t):=\widehat{\operatorname{Cov}}\left(X^{(j)}(s), X^{(k)}(t)\right) = \frac{1}{N} \sum_{i = 1}^N X_i^{(j)}(s) X_i^{(k)}(t)$.
If the $X_i$ are independent copies of the process $X$, it holds
\begin{align*}
& \mathbb{E} \left[ \int_{\mathcal{T}_j} \int_{\mathcal{T}_k} 
\left(C_ {jk}(s,t) - \hat C_{jk}(s,t)  \right)^2 \mathrm{d} s ~ \mathrm{d} t \right]\\
& \quad = \mathbb{E} \left[ \frac{1}{N^2} \sum_{i = 1}^N \sum_{l = 1}^N \int_{\mathcal{T}_j} \int_{\mathcal{T}_k} 
\left(C_{jk}(s,t) - X_i^{(j)}(s)X_i^{(k)}(t)  \right) \left(C_{jk}(s,t) - X_l^{(j)}(s)  X_l^{(k)}(t)  \right)\mathrm{d} s ~ \mathrm{d} t \right] \\
& \quad  =  \frac{1}{N^2} \sum_{i = 1}^N \int_{\mathcal{T}_j} \int_{\mathcal{T}_k} 
\mathbb{E} \left[ \left(C_{jk}(s,t) - X_i^{(j)}(s) X_i^{(k)}(t) \right)^2  \right] \mathrm{d} s ~ \mathrm{d} t \\
& \quad  =  \frac{1}{N} \int_{\mathcal{T}_j} \int_{\mathcal{T}_k} 
\mathbb{E} \left[  X^{(j)}(s)^2 X^{(k)}(t) ^2 \right] - C_{jk}(s,t)^2 \mathrm{d} s ~ \mathrm{d} t = O(N^{-1}).
\end{align*}
The last step follows from the fact that the integral term does not depend on $N$ and is finite by assumption~\eqref{ass:finite4Moment} and the conditions in Prop.~\ref{prop:GammaProperties} for $C_{jk}$. This implies \[\left \Vert \Gamma^{(jk)} - \hat \Gamma^{(jk)} \right \Vert_\text{op} 
\leq \left(  \int_{\mathcal{T}_j} \int_{\mathcal{T}_k} 
\left(C_ {jk}(s,t) - \hat C_{jk}(s,t)  \right)^2 \mathrm{d} s ~ \mathrm{d} t \right)^{1/2} \stackrel{\text{Markov}}{=} O_p(N^{-1/2}).\]

Recall $Z_{ln}^{(jk)} = \operatorname{Cov}(\xi_l^{(j)}, \xi_n^{(k)})$ and $\hat Z_{ln}^{(jk)}=\frac{1}{N-1} \sum_{i  =1}^N \hat \xi_{i,l}^{(j)} \hat \xi_{i,n}^{(k)}$ for $j,k = 1, \ldots, p,~ l = 1,\ldots, M_j,~ n = 1, \ldots, M_k.$
As $\boldsymbol{Z}$ and $\boldsymbol{\hat Z}$ are both symmetric  matrices in $\mathbb{R}^{M_+ \times M_+}$ it holds \citep[cf.][Chapter 3.7]{HornJohnson:1991}
\begin{equation}
\lambda_{\max}(\boldsymbol{Z} - \boldsymbol{\hat Z})
\leq \max_{j = 1, \ldots, p} \max_{l = 1, \ldots, M_j} \sum_{k= 1}^p \sum_{n = 1}^{M_k} \left|Z^{(jk)}_{ln} - \hat Z^{(jk)}_{ln} \right|.
\label{eq:boundLambdaMax}
\end{equation}

\allowdisplaybreaks

Let now  $j,k = 1, \ldots, p,~ l = 1,\ldots, M_j,~ n = 1, \ldots, M_k$ be fixed. Assumption~\eqref{ass:scoreCalc} gives
\begin{align*}
& \left|Z^{(jk)}_{ln} - \hat Z^{(jk)}_{ln} \right|
= \left| \operatorname{Cov}\left(\xi_l^{(j)}, \xi_n^{(k)}\right) 
- \frac{1}{N} \sum_{i  =1}^N  \hat \xi_{i,l}^{(j)} \hat \xi_{i,n}^{(k)}
-\frac{1}{N(N-1)} \sum_{i  =1}^N  \hat \xi_{i,l}^{(j)} \hat \xi_{i,n}^{(k)} \right| \\
& \quad \stackrel{\eqref{ass:scoreCalc}}{\leq}  \left| \operatorname{Cov}\left( \scal{X^{(j)}}{\phi_l^{(j)}}, \scal{X^{(k)}}{\phi_n^{(k)}}  \right) 
-\frac{1}{N} \sum_{i  =1}^N  \scal{X_i^{(j)}}{\hat \phi_l^{(j)}} \cdot \scal{X_i^{(k)}}{\hat \phi_n^{(k)}} \right| \\
&  \qquad + \frac{1}{N(N-1)} \sum_{i = 1}^N \left| \scal{X_i^{(j)}}{\hat \phi_l^{(j)}}\right| \left| \scal{X_i^{(k)}}{\hat \phi_n^{(k)}} \right|
 \\
& \quad \leq  \left| \int_{\mathcal{T}_j} \int_{\mathcal{T}_k} \mathbb{E} \left( X^{(j)}(s)X^{(k)}(t) \right) \phi_l^{(j)}(s) \phi_n^{(k)}(t) \mathrm{d} s ~ \mathrm{d} t -   \int_{\mathcal{T}_j} \int_{\mathcal{T}_k}  \frac{1}{N} \sum_{i  =1}^N \left(X_i^{(j)}(s) X_i^{(k)}(t) \right) \hat \phi_l^{(j)}(s) \hat \phi_n^{(k)}(t) \mathrm{d} s ~ \mathrm{d} t
\right| \\
& \qquad + \frac{1}{N(N-1)} \sum_{i  =1}^N  \norm{X_i^{(j)}} \norm{\hat \phi_l^{(j)}} \norm{X_i^{(k)}} \norm{\hat \phi_n^{(k)}} \\
& \quad \stackrel{\eqref{eq:XnormBounded}}{=} 
\left| \int_{\mathcal{T}_j} \int_{\mathcal{T}_k}
\operatorname{Cov} \left( X^{(j)}(s), X^{(k)}(t) \right) \phi_l^{(j)}(s) \phi_n^{(k)}(t)  -  \widehat {\operatorname{Cov} }\left( X^{(j)}(s), X^{(k)}(t) \right)
\hat \phi_l^{(j)}(s) \hat \phi_n^{(k)}(t)  \mathrm{d} s ~ \mathrm{d} t \right| \\
& \qquad + \frac{1}{N(N-1)} \sum_{i  =1}^N O_p(1)\cdot 1 \cdot O_p(1) \cdot 1\\
& \quad = \left| \int_{\mathcal{T}_j} \int_{\mathcal{T}_k} C_{jk} (s,t) \left[ \phi_l^{(j)}(s) \phi_n^{(k)}(t)  - \hat \phi_l^{(j)}(s) \hat \phi_n^{(k)}(t)\right]  \mathrm{d} s ~ \mathrm{d} t \right| \\
& \qquad +  \left| \int_{\mathcal{T}_j} \int_{\mathcal{T}_k} \left[C_{jk} (s,t) - \hat C_{jk} (s,t) \right] \hat \phi_l^{(j)}(s) \hat \phi_n^{(k)}(t) \mathrm{d} s ~ \mathrm{d} t
\right| + O_p (N^{-1})\\
& \quad \leq  \int_{\mathcal{T}_j} \int_{\mathcal{T}_k} C_{jj} (s,s)^{1/2} C_{kk} (t,t)^{1/2} \left| \phi_l^{(j)}(s) \phi_n^{(k)}(t)  - \hat \phi_l^{(j)}(s) \hat \phi_n^{(k)}(t)\right|\mathrm{d} s ~ \mathrm{d} t  \\
& \qquad +  \int_{\mathcal{T}_k}  \left| \left((\Gamma^{(jk)} - \hat \Gamma^{(jk)}) \hat \phi_l^{(j)} \right) (t) \hat \phi_n^{(k)}(t) \right|  \mathrm{d} t + O_p (N^{-1})
\\
& \quad \leq 
 \left \Vert C_{jj} \right \Vert_\infty^{1/2} \left \Vert C_{kk}\right \Vert_\infty^{1/2} \left(\int_{\mathcal{T}_j}  \int_{\mathcal{T}_k}  \left| \phi_l^{(j)}(s) \right| \left| \phi_n^{(k)}(t)  - \hat \phi_n^{(k)}(t)\right| \mathrm{d} s ~ \mathrm{d} t
+
\int_{\mathcal{T}_j} \int_{\mathcal{T}_k}\left| \phi_l^{(j)}(s) - \hat \phi_l^{(j)}(s) \right|  \left|\hat \phi_n^{(k)}(t)\right| \mathrm{d} s ~ \mathrm{d} t \right) \\
& \qquad +  \norm{(\Gamma^{(jk)} - \hat \Gamma^{(jk)}) \hat \phi_l^{(j)}} \norm{\hat \phi_n^{(k)}} + O_p (N^{-1})
\\
&  \quad \leq 
 \left( \left \Vert C_{jj} \right \Vert_\infty\left \Vert C_{kk} \right \Vert_\infty \lambda(\mathcal{T}_j) \lambda(\mathcal{T}_k) \right)^{1/2} \left( \norm{\phi_l^{(j)}} \norm{ \phi_n^{(k)}  - \hat \phi_n^{(k)}} 
+
\norm{\phi_l^{(j)} - \hat \phi_l^{(j)}} \norm{\hat  \phi_n^{(k)}} \right) \\
& \qquad +  \left \Vert \Gamma^{(jk)} - \hat \Gamma^{(jk)} \right \Vert _\text{op} \norm{\hat \phi_l^{(j)}} \norm{\hat \phi_n^{(k)}} + O_p (N^{-1})
\\ 
& \quad =
 \left( \left \Vert C_{jj} \right \Vert_\infty\left \Vert C_{kk} \right \Vert_\infty \lambda(\mathcal{T}_j) \lambda(\mathcal{T}_k) \right)^{1/2} \left(  \norm{ \phi_n^{(k)}  - \hat \phi_n^{(k)}} 
+\norm{\phi_l^{(j)} - \hat \phi_l^{(j)}} \right)+  \left  \Vert \Gamma^{(jk)} - \hat \Gamma^{(jk)} \right  \Vert_\text{op} + O_p (N^{-1}) \\
& \quad = O_p( \Delta^{(j)}_{M_j} r_N^\Gamma) + O_p(\Delta^{(k)}_{M_k} r_N^\Gamma) + O_p(N^{-1/2}) + O_p (N^{-1}) \\
& \quad = O_p(\max( \Delta^{(j)}_{M_j}  r_N^\Gamma,  \Delta^{(k)}_{M_k} r_N^\Gamma,N^{-1/2} )).
\end{align*}
The rate for $\phi_n^{(k)}$ and $\phi_{l}^{(j)}$ in the last steps is shown at the beginning of the proof of Prop.~\ref{prop:asymptVar}.
In total, equation \eqref{eq:boundLambdaMax}, $M_{\max} = \max_{j = 1 , \ldots,  p}M_j$ and $\Delta_M:= \max_{j = 1, \ldots, p}{\Delta^{(j)}_{M_j}}$ give
\[\lambda_{\max}(\boldsymbol{Z} - \boldsymbol{\hat Z}) \leq O_p(M_{\max}\max(N^{-1/2} , \Delta_M r_N^\Gamma)).\]

\end{proof}

\begin{proof}[Proof of Prop.~\ref{prop:asymptVar}] \label{proof:asymptVar}

Under assumption~\eqref{ass:eFunSign} and using the convention $\lambda_0^{(j)} := \infty$, Lemma 4.3. in \citet{Bosq:2000}  gives for $m = 1, \ldots, M_j$
\begin{align*}
\norm{\phi_m^{(j)} - \hat \phi_m^{(j)}} 
& \leq 8^{1/2} \left[ \min\left(\lambda_m^{(j)} - \lambda_{m+1}^{(j)},\lambda_{m-1}^{(j)} - \lambda_{m}^{(j)}\right) \right] ^{-1}  \left \Vert \Gamma^{(j)}- \hat\Gamma^{(j)} \right \Vert_\text{op} \\
& \leq  8^{1/2} \Delta^{(j)}_{M_j} \left \Vert \Gamma^{(j)}- \hat\Gamma^{(j)} \right  \Vert_\text{op}
= O_p(\Delta^{(j)}_{M_j} r^\Gamma_N).
\end{align*}
Based on this result, Lemma~\ref{lemma:covMatrixRate} states that 
$\lambda_{\max}(\boldsymbol{Z} - \boldsymbol{\hat Z}) \leq O_p(M_{\max} \max(N^{-1/2} , \Delta_M r_N^\Gamma))$
with $\Delta_M:= \max_{j = 1, \ldots, p}{\Delta^{(j)}_{M_j}}$ and  $M_{\max} = \max_{j = 1 , \ldots,  p}M_j$.

\begin{enumerate}
\item Eigenvalues: Let $\boldsymbol{\xi} \in \mathbb{R}^{M_+}$ with entries  $\xi_{m}^{(j)} = \scal{X^{(j)}}{\phi_m^{(j)}} = \scal{X^{[M](j)}}{\phi_m^{(j)}},~ m = 1, \ldots, M_j,~ j = 1 , \ldots, p$. For fixed $m = 1, \ldots, M_+$ it holds that
\begin{align*}
\left| \nu^{[M]}_m - \hat \nu_m \right|
& = \left| \operatorname{Var}(\myscal{X_i^{[M]}}{\psi_m^{[M]}})  - \boldsymbol{ \hat c_m} ^\top \boldsymbol{\hat Z} \boldsymbol{\hat c_m}\right|
= \left|  \boldsymbol{c_m} ^\top  \operatorname{Var}(\boldsymbol{\xi})  \boldsymbol{c_m} -  \boldsymbol{\hat c_m} ^\top   \boldsymbol{\hat Z}  \boldsymbol{\hat c_m} \right| \\
& = \left|  (\boldsymbol{c_m} - \boldsymbol{\hat c_m}) ^\top \boldsymbol{Z} \boldsymbol{c_m} + \boldsymbol{\hat c_m} ^\top \boldsymbol{Z} (\boldsymbol{c_m} - \boldsymbol{\hat c_m}) + \boldsymbol{\hat c_m} ^\top ( \boldsymbol{Z}  -  \boldsymbol{\hat Z} ) \boldsymbol{\hat c_m} \right| \\
& \leq  \vecnorm{\boldsymbol{c_m} - \boldsymbol{\hat c_m}}\cdot  \nu^{[M]}_m \vecnorm{\boldsymbol{c_m}} 
	+ \lambda_{\max}(\boldsymbol{Z}) \vecnorm{ \boldsymbol{c_m} - \boldsymbol{\hat c_m}} + \lambda_{\max}(\boldsymbol{Z}  -  \boldsymbol{\hat Z}) \vecnorm{\boldsymbol{\hat c_m}} \\	
& =   \vecnorm{\boldsymbol{c_m} - \boldsymbol{\hat c_m}} ( \nu^{[M]}_m  + \nu^{[M]}_1)
	+\lambda_{\max}(\boldsymbol{Z}  -  \boldsymbol{\hat Z}) \\
& \leq \frac{8^{1/2}\lambda_{\max}(\boldsymbol{Z} - \boldsymbol{\hat Z)}}{\min(\nu^{[M]}_{m-1}  - \nu^{[M]}_m , \nu^{[M]}_m  - \nu^{[M]}_{m+1})} 2 \nu^{[M]}_1
	+\lambda_{\max}(\boldsymbol{Z}  -  \boldsymbol{\hat Z}) \\
& = \left[ \frac{2^{5/2} \nu^{[M]}_1}{\min(\nu^{[M]}_{m-1}  - \nu^{[M]}_m , \nu^{[M]}_m  - \nu^{[M]}_{m+1})} 
	+ 1 \right] \lambda_{\max}(\boldsymbol{Z}  -  \boldsymbol{\hat Z}) \\
&	= O_p(M_{\max} \max(N^{-1/2} , \Delta_M r_N^\Gamma)) ,
\end{align*} 
as the expression in square brackets converges to a constant $C < \infty$ (cf. Prop.~\ref{prop:asymptBias} and the fact that $\nu_m$ is assumed to have multiplicity $1$, see p.~18 in the main paper). Here  $\lambda_{\max}(\boldsymbol{A})$ denotes the maximal eigenvalue of a symmetric matrix $\boldsymbol{A}$. The second inequality follows from Corollary~1 in \cite{YuEtAl:2015} and the fact that $\nu^{[M]}_m \leq \nu^{[M]}_1$.

\item Eigenfunctions: Consider the $j$-th element of the $m$-th eigenfunctions:
\begin{align*}
\norm{ \psi^{[M](j)}_m - \hat \psi_m^{(j)}} 
& \leq  \sum_{n = 1}^{M_j}  \left| [\boldsymbol{c_m}]_n^{(j)} - [\boldsymbol{\hat c_m}]_n^{(j)} \right| \norm{ \phi_n ^{(j)}} + \left|[\boldsymbol{\hat c_m}]_n^{(j)} \right| \norm{\phi_n ^{(j)} - \hat \phi_n ^{(j)}} \\
& \leq M_j^{1/2}   \vecnorm{ \boldsymbol{c_m} - \boldsymbol{\hat c_m}}  + M_j^{1/2} \vecnorm{ \boldsymbol{\hat c_m} } O_p(\Delta^{(j)}_{M_j} r_N^\Gamma)   \\
& \leq  M_j^{1/2} \left( \frac{8^{1/2} \lambda_{\max}(\boldsymbol{Z} - \boldsymbol{\hat Z)}}{\min(\nu^{[M]}_{m-1}  - \nu^{[M]}_m , \nu^{[M]}_m  - \nu^{[M]}_{m+1})} +  O_p(\Delta^{(j)}_{M_j} r_N^\Gamma) \right) \\
& =  M_j^{1/2}  O_p \left( \max (M_{\max}N^{-1/2} ,M_{\max} \Delta_M r_N^\Gamma, \Delta^{(j)}_{M_j} r_N^\Gamma) \right),
\end{align*}
where the last inequality uses again Corollary~1 in \cite{YuEtAl:2015}. 
By definition of the norm, the result for the single elements implies
\[
\mynorm{\psi^{[M]}_m - \hat \psi_m}
  =   O_p \left(M_{\max}^{3/2} \max (N^{-1/2} ,\Delta_M r_N^\Gamma) \right).
\]
\item Scores and reconstructed $\hat X$: 
For $\hat \rho_{i,m}  = \boldsymbol{\Xi_{i,\cdot}\hat c_m} = \myscal{\hat X_i^{[M]}}{\hat \psi_m}$ as in  Section~\ref{sec:estMFPCA}   with $\hat X_i^{[M](j)} = \sum_{m = 1}^{M_j} \hat \xi_{i,m}^{(j)} \hat \phi_m^{(j)}$,
\begin{align*}
\left| \rho_{i,m}^{[M]}-  \hat \rho_{i,m} \right|
& = \left| \myscal{ X_i^{[M]}}{\psi_m^{[M]} - \hat \psi_m} +
\myscal{X_i^{[M]} -  \hat X_i^{[M]}}{\hat \psi_m}  \right| \\
& \leq  \mynorm{X_i^{[M]}} \cdot \mynorm{\psi_m^{[M]} - \hat \psi_m}
+\mynorm{X_i^{[M]}- \hat X_i^{[M]}} \cdot \mynorm{\hat \psi_m}.
\end{align*}

$ \mynorm{X_i^{[M]}}$ is bounded in probability using equation \eqref{eq:varMtoVar} and analogous arguments as for $ \mynorm{X}$ in the proof of Prop.~\ref{prop:asymptBias} (convergence of $ \rho_m^{[M]}$).
For the second term, note that
\begin{align*}
\norm{X_i^{[M](j)} - \hat X_i^{[M](j)}}
&= \norm{\sum_{m = 1}^{M_j} \xi_{i,m}^{(j)} \phi_m^{(j)} - \hat \xi_{i,m}^{(j)} \hat \phi_m^{(j)}} 
 \leq \sum_{m = 1}^{M_j} \left|  \xi_{i,m}^{(j)} \right| \norm{\phi_m^{(j)} - \hat \phi_m^{(j)}} + \left| \xi_{i,m}^{(j)} - \hat \xi_{i,m}^{(j)} \right| \norm{\hat \phi_m^{(j)}} \\
&  \leq \sum_{m = 1}^{M_j} \left|  \xi_{i,m}^{(j)} \right| \norm{\phi_m^{(j)} - \hat \phi_m^{(j)}} + \left| \scal{X_i^{(j)}}{\phi_m^{(j)}}- \scal{X_i^{(j)}}{\hat \phi_m^{(j)}} \right| \\
 & \leq \sum_{m = 1}^{M_j} \left(\left|  \xi_{i,m}^{(j)} \right|+ \norm{X_i^{(j)}}\right) \norm{\phi_m^{(j)} - \hat \phi_m^{(j)}}.
\end{align*}
The univariate scores are uniformly bounded in probability: For $m = 1, \ldots, M_j$, let $\varepsilon > 0$ and $c := (\frac{2 \lambda_1^{(j)}}{\varepsilon})^{1/2} < \infty$. Then
\[P(|  \xi_{i,m}^{(j)} | > c)
\stackrel{\text{Markov}}{\leq}  \frac{1}{c^2} \mathbb{E} [ |  \xi_{i,m}^{(j)} |^2 ] 
= \frac{1}{c^2} \operatorname{Var}(\xi_{i,m}^{(j)}) = \frac{\lambda_m^{(j)}}{c^2} < \varepsilon.
\]
Hence $\norm{X_i^{[M](j)} - \hat X_i^{[M](j)}} = M_j O_p(1) O_p(\Delta^{(j)}_{M_j} r_N^\Gamma)$ and
\begin{align*}
\mynorm{X_i^{[M]} - \hat X_i^{[M]}}
= O_p(M_{\max} \Delta_M r_N^\Gamma ).
\end{align*}
In total,
\begin{align*}
\left| \rho_{i,m}^{[M]}-  \hat \rho_{i,m} \right|
& \leq  \mynorm{X_i^{[M]}} \mynorm{\psi_m^{[M]} - \hat \psi_m}
+\mynorm{X_i^{[M]}- \hat X_i^{[M]}}\\
& = O_p(1) O_p(M_{\max}^{3/2} \max(N^{-1/2}, \Delta_M r_N^\Gamma)) + O_p(M_{\max} \Delta_M r_N^\Gamma)\\
& = O_p (M_{\max}^{3/2} \max (N^{-1/2} ,\Delta_M r_N^\Gamma)).
\end{align*}
\end{enumerate}

\end{proof}

\newpage

%\FloatBarrier
%\newpage

\section*{Simulation -- Additional Results}

\subsection*{Construction of Eigenfunctions (Technical Details)}

\textbf{Setting 1 and 2:}
The first two settings of the simulation study consider multivariate functional data where each element has a one-dimensional domain (cf. Section~\ref{sec:simFun}). As a starting point for the construction of the multivariate eigenfunctions $\psi_m$ with $p$ elements, we use Fourier basis functions $f_1 , \ldots, f_M$ on the interval $[0,2]$. Next, choose split points $0 = T_1 < T_2 < \ldots < T_p < T_{p+1} = 2$ and shift values $\eta_1 , \ldots ,  \eta_p \in \mathbb R $ such that $\mathcal{T}_j = [T_{j} + \eta_j, T_{j+1}+\eta_j]$. In the first setting with $p = 2$, one has $T_1 = 0, T_2 = 1, T_3 = 2$ and $\eta_1 = 0, \eta_2 = 1$, i.e. the functions are cut at $T_2 = 1$, and the second part is shifted to the left by $1$ such that $\mathcal{T}_1 = \mathcal{T}_2 = [0,1]$. Given random signs $\sigma_1 , \ldots ,  \sigma_p \in \{-1,1\}$, the multivariate eigenfunctions are given by their elements
\begin{align*}
\psi_m^{(j)}(t_j) & = \sigma_j \cdot f_m \vert_{[T_{j}, T_{j+1}]}(t_j- \eta_j), \quad m = 1 , \ldots ,  M.
\end{align*}
The constuction process is illustrated in Fig.~\ref{fig:bigSplitONB}.
Clearly, $\{\psi_m,~ m  = 1 , \ldots ,  M\}$ is an orthonormal system in $\mathcal{H} = L^2(\mathcal{T}_1) \times \ldots \times L^2(\mathcal{T}_p)$. The observations $x_i$ for the simulation are constructed as a truncated Karhunen-Lo\`{e}ve expansion, cf. the introduction of Section~\ref{sec:Simulation}. Exemplary data for the second simulation setting including sparse data and data with measurement error is given in Fig.~\ref{fig:simData}

\begin{figure}[ht]
\centering
\includegraphics[height = 3.5cm]{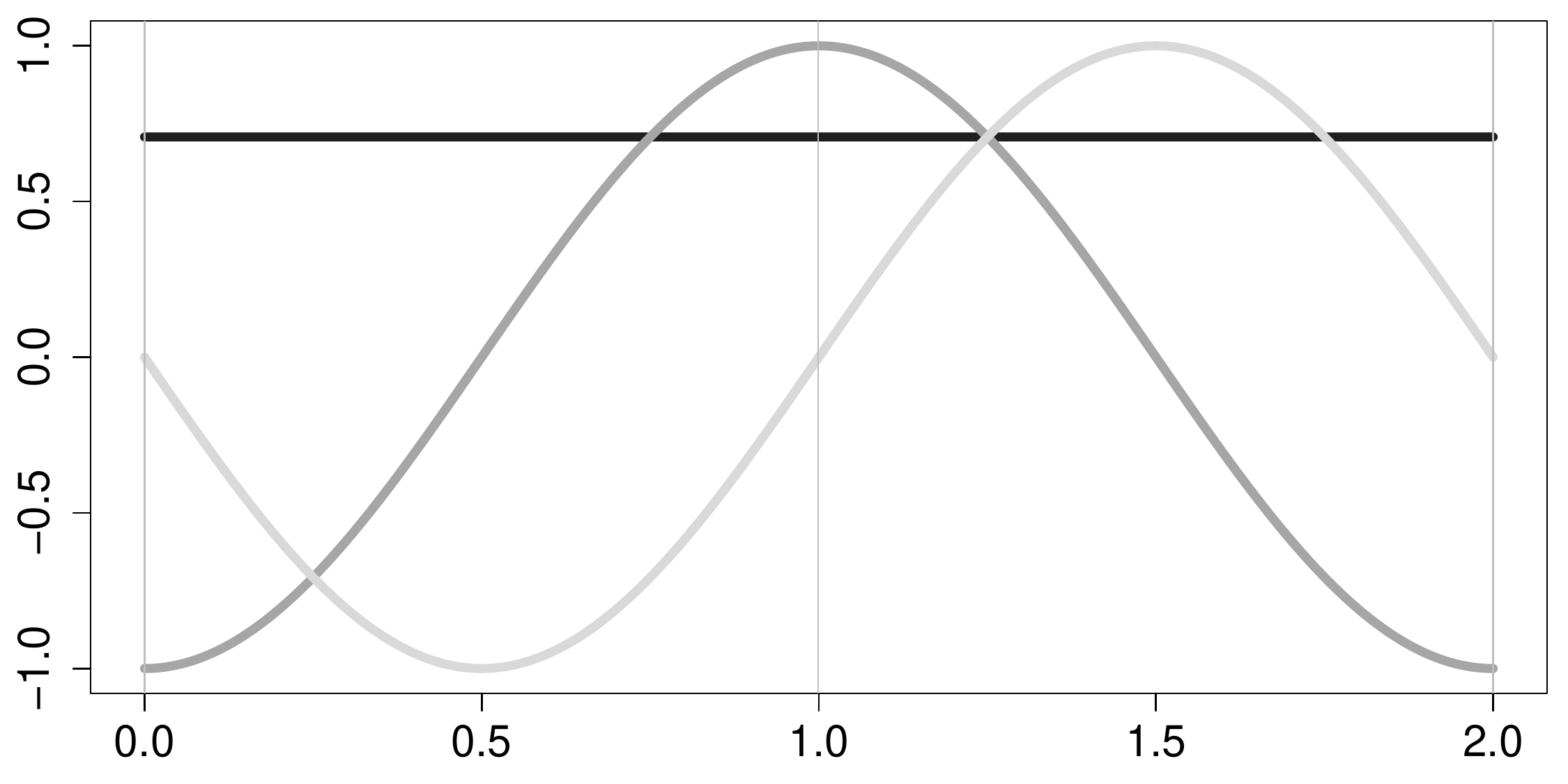}
\hspace{0.75cm}
\includegraphics[height = 3.5cm]{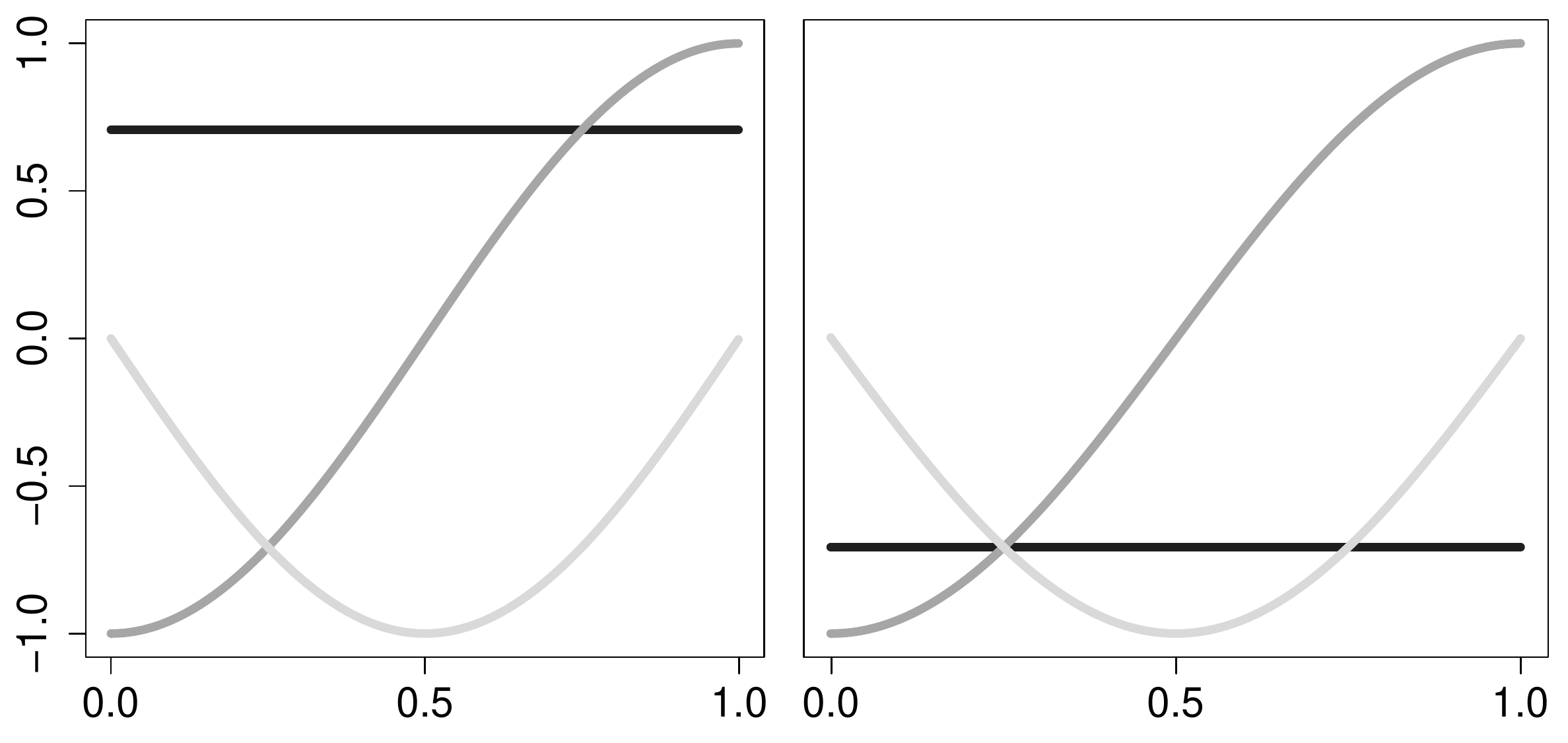}
\caption{Illustration of the construction of the multivariate eigenfunctions $\psi_m$ for the first setting. Left: The first $M = 3$ functions of the Fourier basis on $[0,2]$ with one split point. Right: The shifted pieces multiplied with random signs form the first three bivariate eigenfunctions.}
\label{fig:bigSplitONB}
\end{figure}

\begin{figure}[ht]
\centering
\includegraphics[height = 3.5cm]{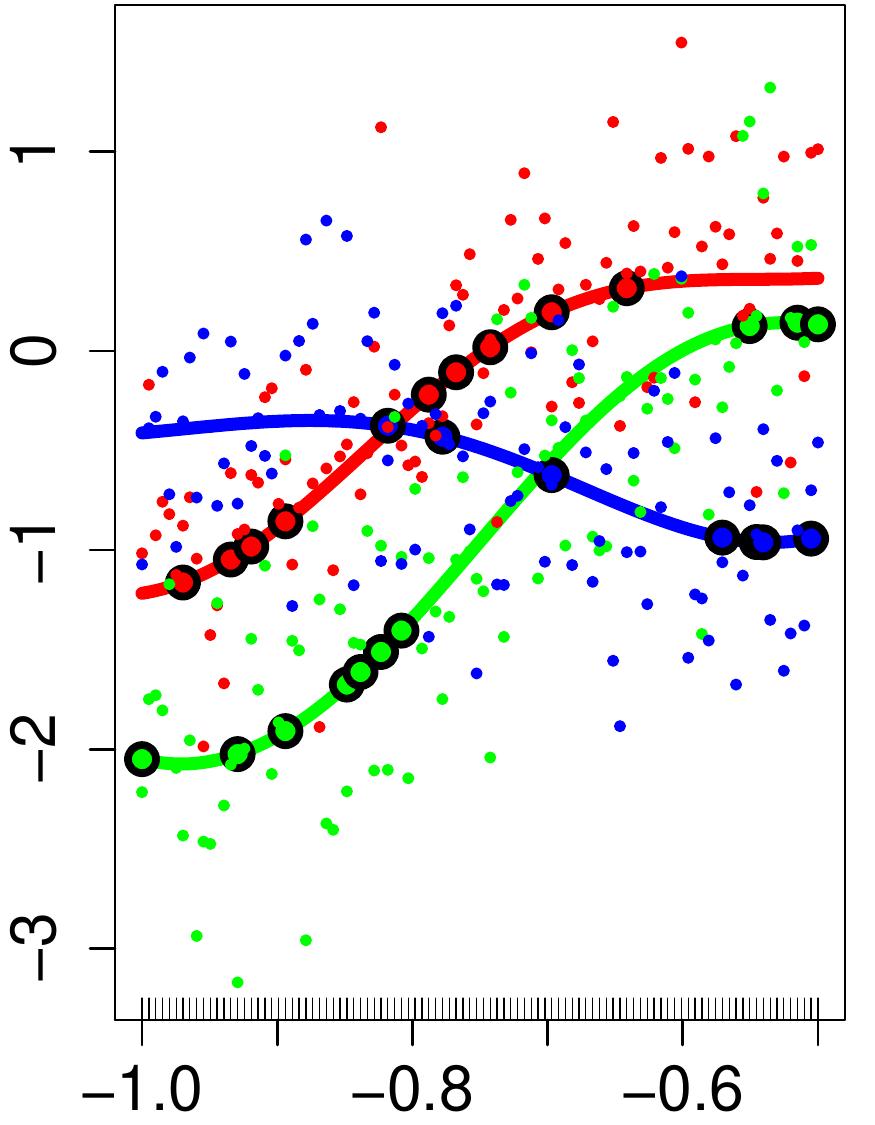}
%\hfill
\includegraphics[height = 3.5cm]{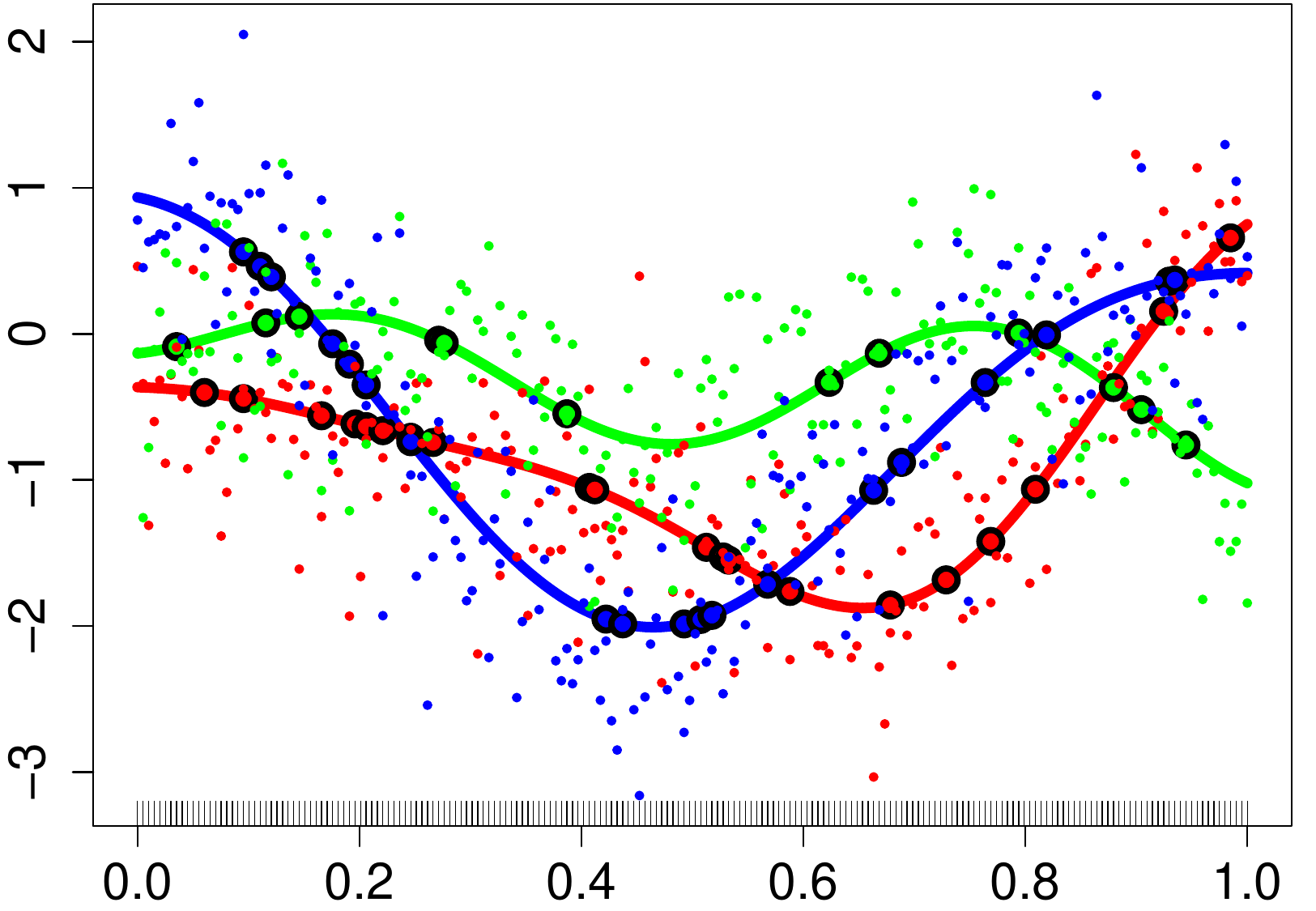}
%\hfill
\includegraphics[height = 3.5cm]{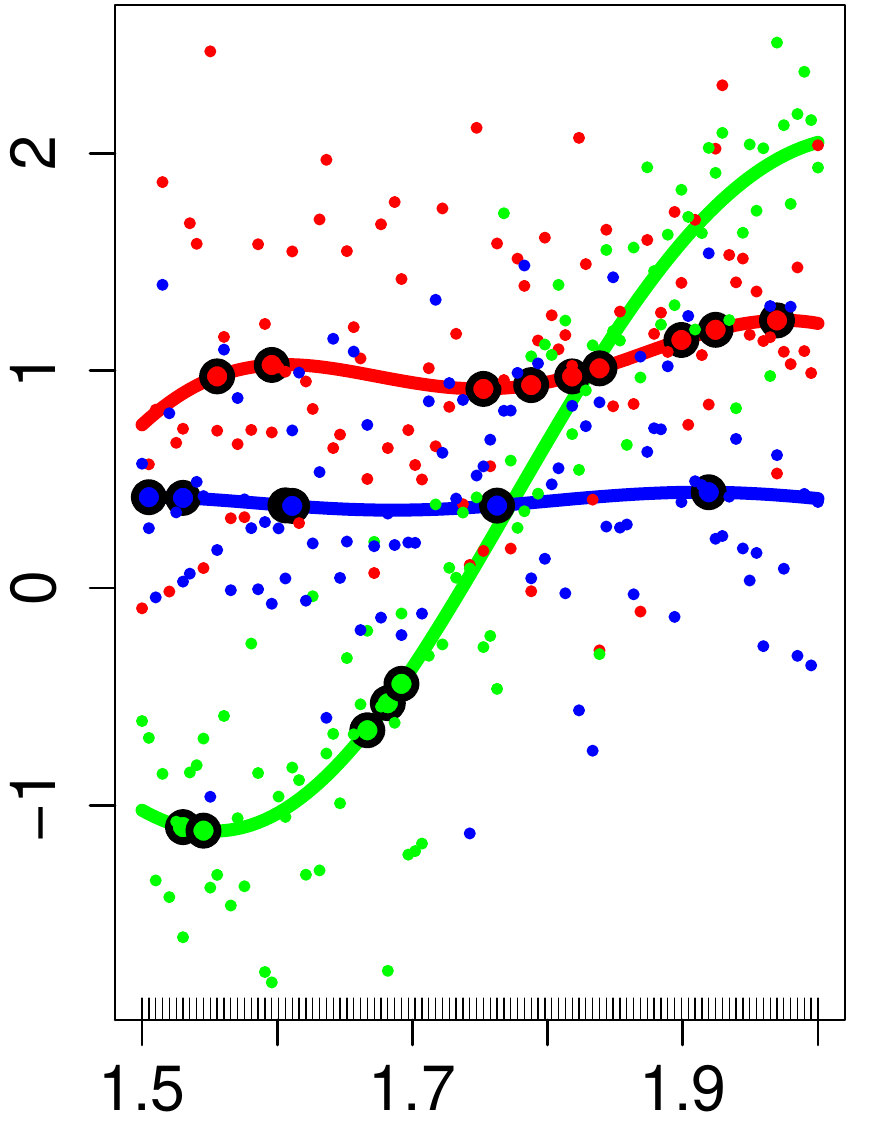}
\caption{Three examples for simulated data in simulation setting 2 based on the leading $M = 8$ Fourier basis functions and exponential eigenvalue decay. Left: $x_i^{(1)}$, Middle: $x_i^{(2)}$, Right: $x_i^{(3)}$. Solid lines show the realizations $x_i$, small points are the corresponding data with measurement error, big points mark measurements of the artificially sparsified data (high sparsity level).}
\label{fig:simData}
\end{figure}

\textbf{Setting 3:}
The data in the third setting consists of images and functions, hence multivariate functional data with elements having different dimensional domains (cf. Sectrion~\ref{sec:simImage}).
The basic idea here is to find orthonormal bases for each of the domains and to construct the eigenfunctions as weighted combinations of those bases. Specifically, we use five Fourier basis functions $f_{m_1}^{(1,1)},~ f_{m_2}^{(1,2)}$ on $[0,1]$ or $[0, 0.5]$, respectively,  to form $M = 25$ tensor product functions $f_m^{(1)}$ on $[0,1] \times [0,0.5]$ and $M = 25$ Legendre Polynomials $f_m^{(2)}$ on $[-1,1]$.
The eigenfunctions are defined via
\begin{align*}
\psi_m^{(1)}(s,t) &= \sqrt{\alpha} f_{m_1}^{(1,1)}(s) \cdot  f_{m_2}^{(1,2)}(t),\quad (s,t) \in \mathcal{T}_1:= [0, 1]\times[0, 0.5], \\
\psi_m^{(2)}(t) &= \sqrt{1 - \alpha} f_m^{(2)}(t),\quad  t \in \mathcal{T}_2:= [-1,1]
\end{align*}
with a random weight $\alpha \in (0,1)$.  This choice implies that $\psi_m$ forms an orthonormal system in $\mathcal{H} = L^2(\mathcal{T}_1) \times L^2(\mathcal{T}_2)$. In order to avoid extreme weights, $\alpha$ is set to $u_1 /(u_1 + u_2)$ with $u_1, u_2 \sim U(0.2, 0.8)$. This construction restricts $\alpha \in (0.2, 0.8)$ and can easily be generalized to the simulation of multivariate functional data with $p$ elements. Example data based on this type of eigenfunctions is shown in Fig.~\ref{fig:exImageData}.

\begin{figure}[ht]
\centering
\includegraphics[height = 3.75cm]{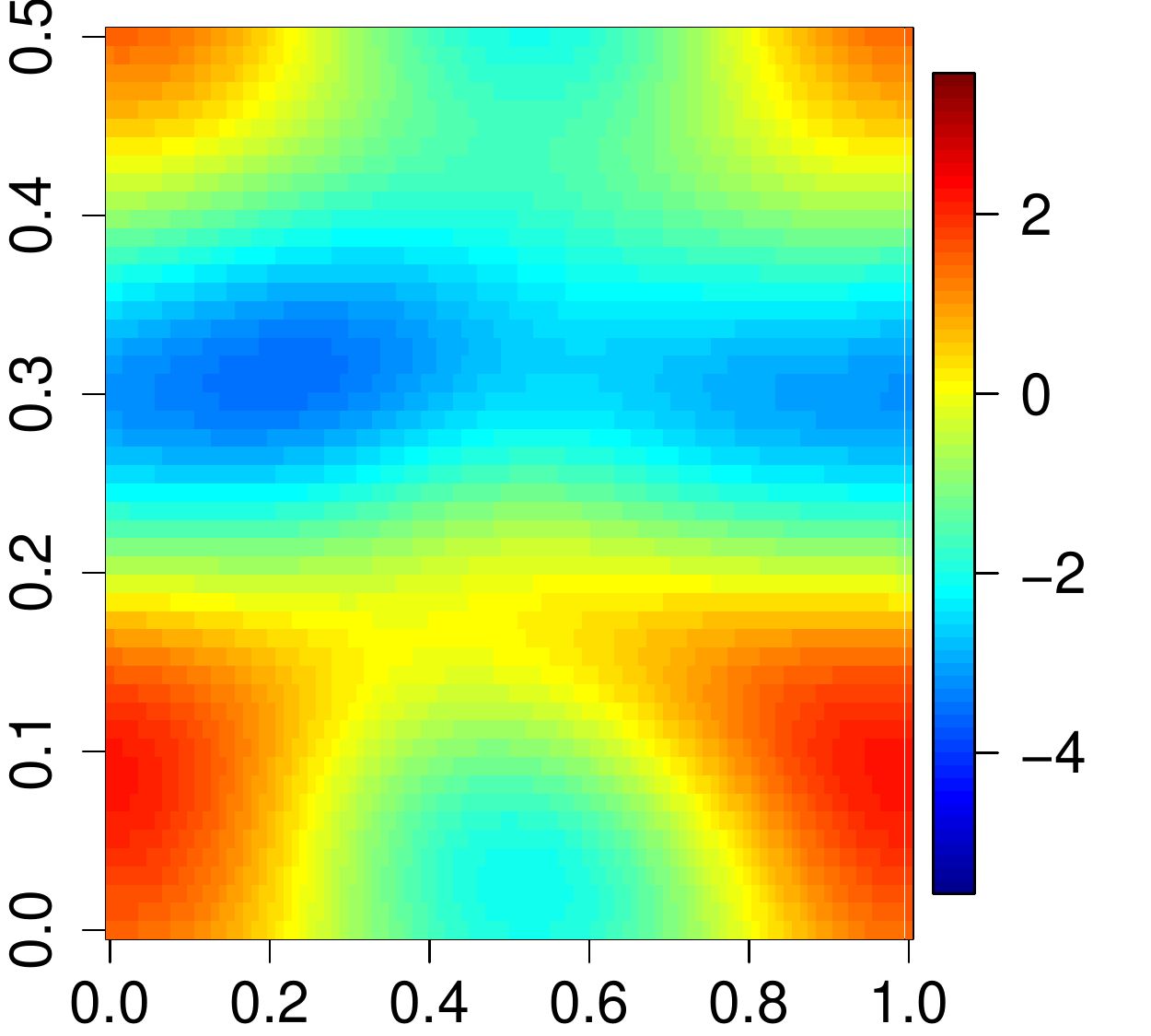}
\hspace{0.5cm}
\includegraphics[height = 3.75cm]{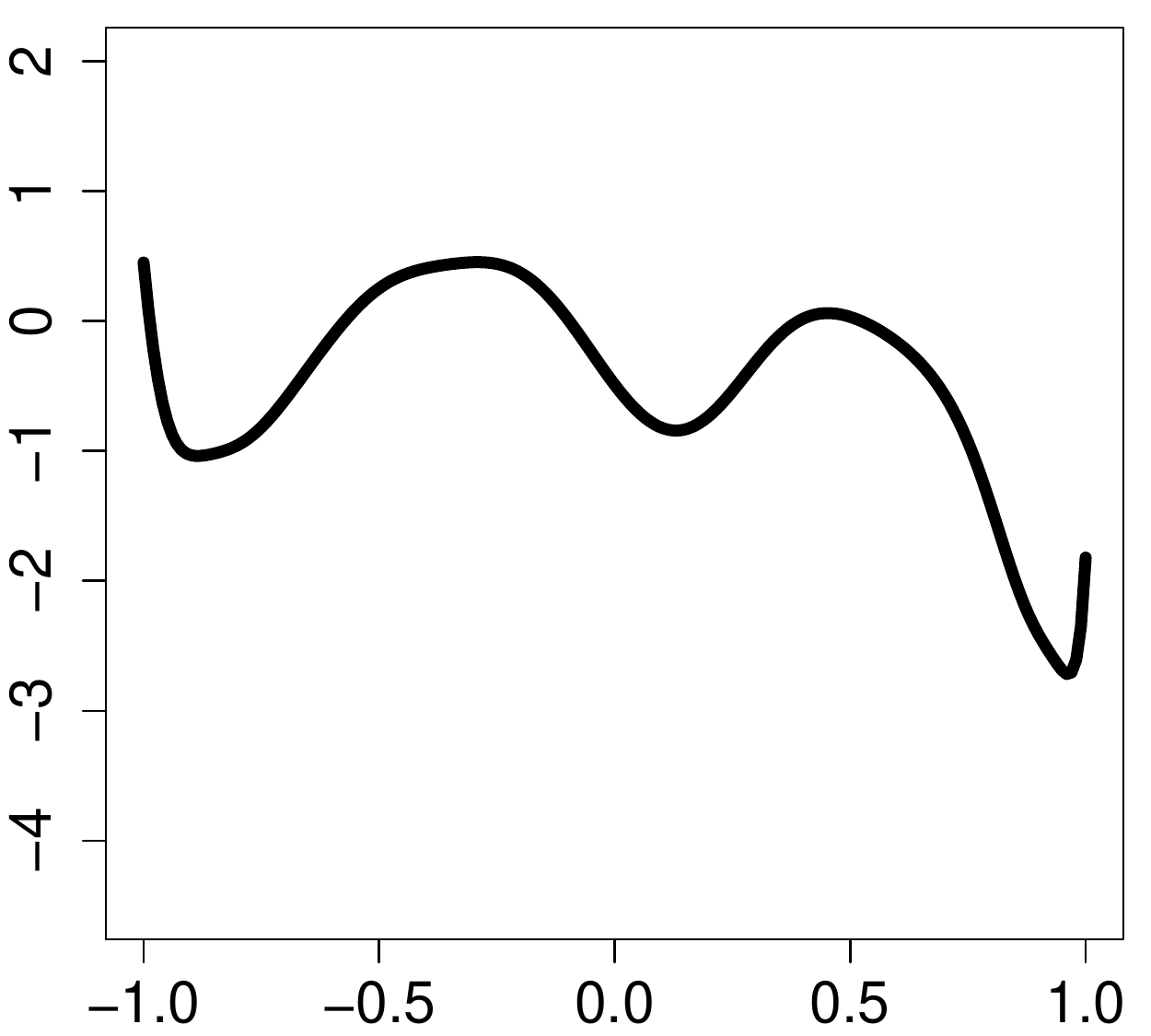}

\includegraphics[height = 3.75cm]{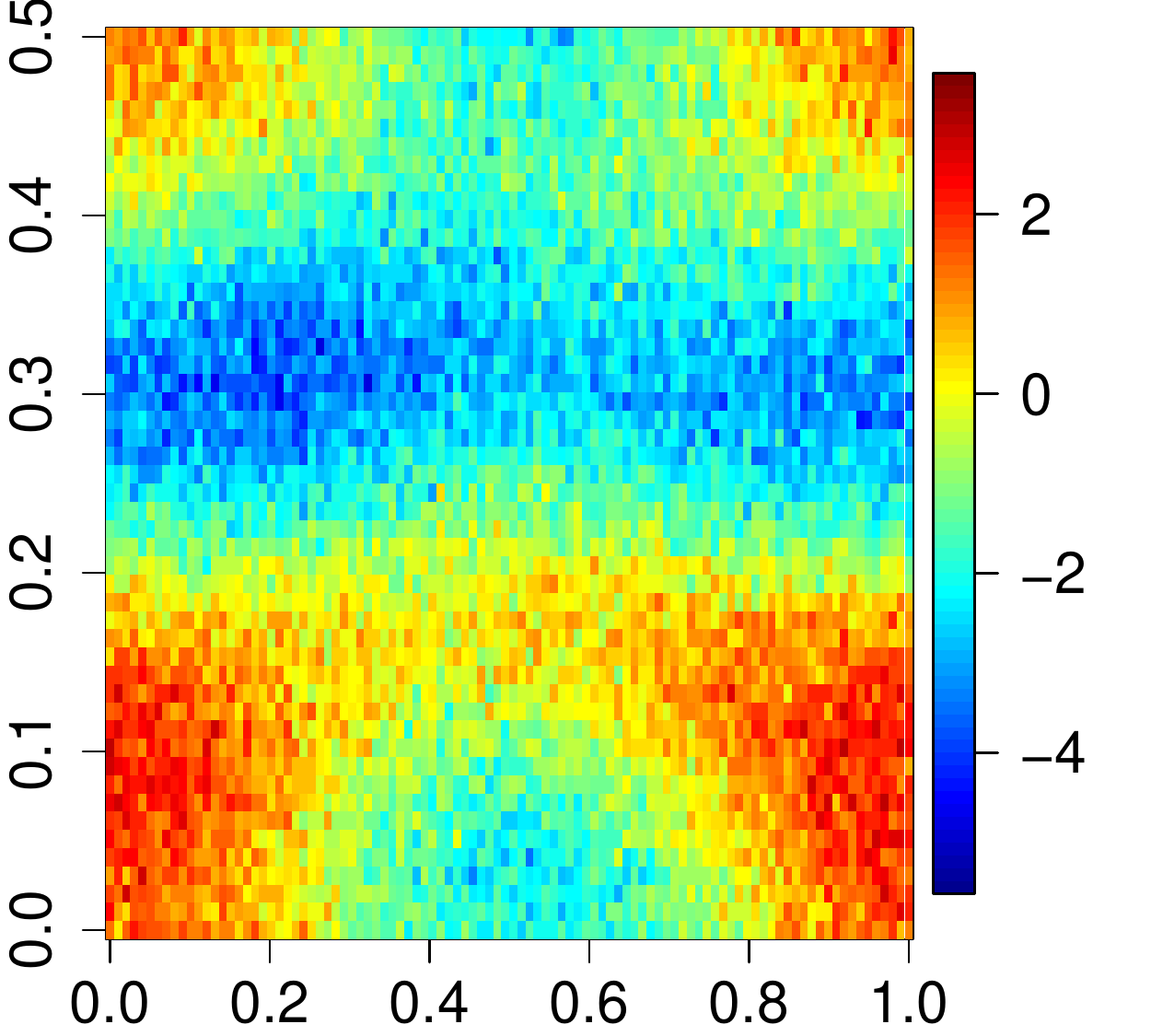}
\hspace{0.5cm}
\includegraphics[height = 3.75cm]{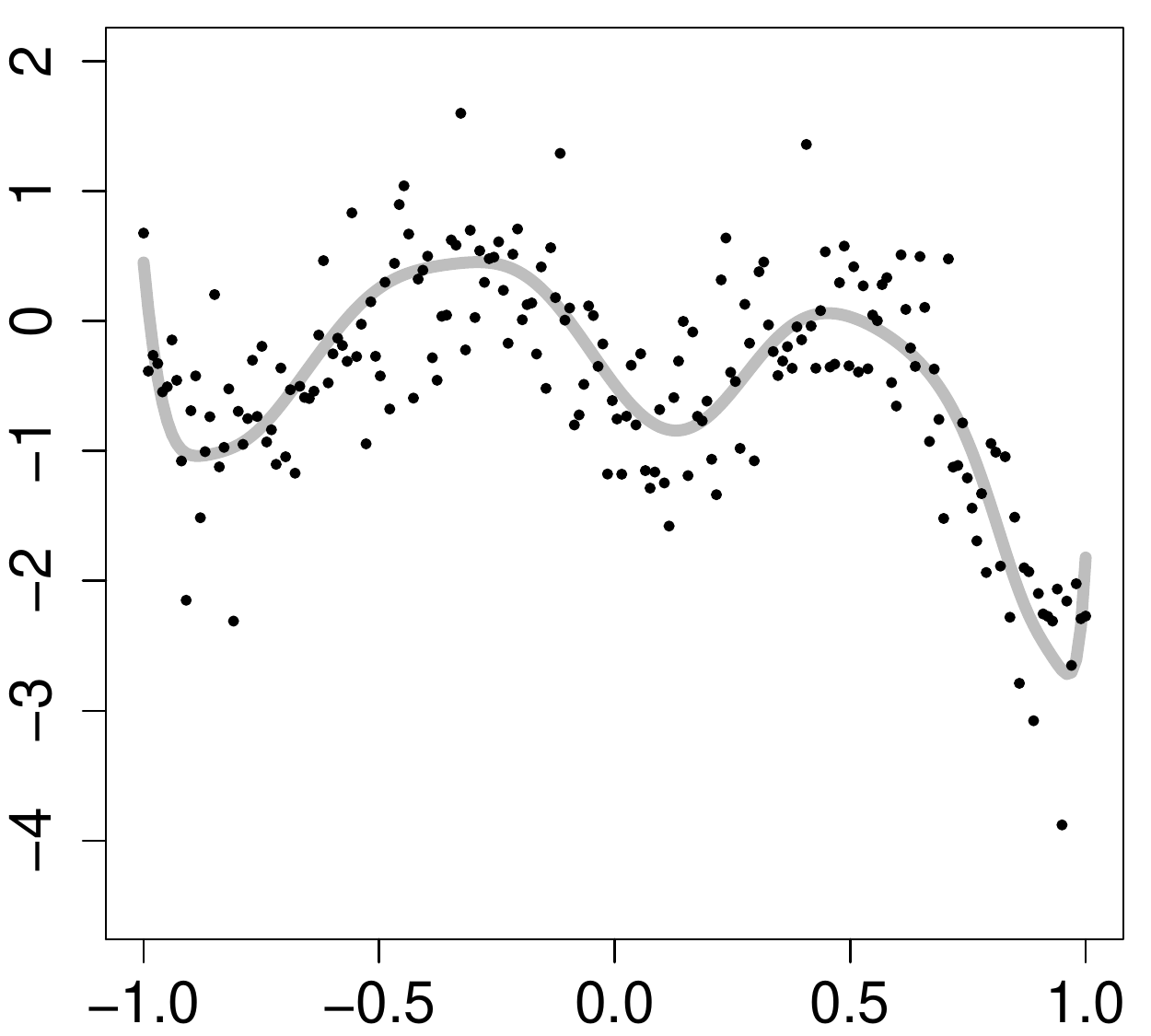}
\caption{Examples for simulated data in the third simulation setting (cf. Section~\ref{sec:simImage}) consisting of images ($x_i^{(1)}$, left) and functions ($x_i^{(2)}$, right) without (1st row) and with measurement error (2nd row).}
\label{fig:exImageData}
\end{figure}

\FloatBarrier

\newpage

\subsection*{Example Fits}

\begin{table}[!ht]
\centering
\caption{True and estimated eigenvalues for the first simulation setting (exponential eigenvalue decay,  eigenfunctions based on the first $M = 8$ Fourier basis functions) for one replication with $N =  250$ observations. 
The reconstruction errors are (in $\%$) $0.007 $ ($\text{MFPCA},~ \sigma^2 = 0$; simulation median: $0.008$), $0.734 $ ($\text{MFPCA},~ \sigma^2 = 0.25$; simulation median: $0.497$), $< 10^{-3} $ ($\text{MFPCA}_\text{RS},~ \sigma^2 = 0$; simulation median: $ < 10^{-3}$) and $0.710 $ ($\text{MFPCA}_\text{RS},~ \sigma^2 = 0.25$; simulation median: $0.480$). 
 The results for the corresponding eigenfunctions are given in Fig.~\ref{fig:fitEFun}.}
\begin{tabular}{lcccccccc}
\hline \hline
\multicolumn{1}{r}{$m = $}   & $1$ & $2$ & $3$ & $4$ & $5$ & $6$ & $7$ & $8$\\ 
\hline
 True Eigenvalues&1.000&0.607&0.368&0.223&0.135&0.082&0.050&0.030\\
$\text{MFPCA}~(\sigma^2 = 0)$&1.144&0.502&0.316&0.249&0.128&0.090&0.048&0.034\\
$\text{MFPCA}~(\sigma^2 = 0.25)$&1.140&0.501&0.316&0.249&0.128&0.087&0.046&0.031\\
$\text{MFPCA}_\text{RS}~(\sigma^2 = 0)$&1.140&0.500&0.315&0.248&0.127&0.090&0.048&0.034\\
$\text{MFPCA}_\text{RS}~(\sigma^2 = 0.25)$&1.139&0.504&0.317&0.252&0.130&0.091&0.048&0.035\\

\hline
\end{tabular}
\label{tab:fitEVal}
\end{table}

\begin{figure}[ht]
\centering
\includegraphics[height = 3.5cm]{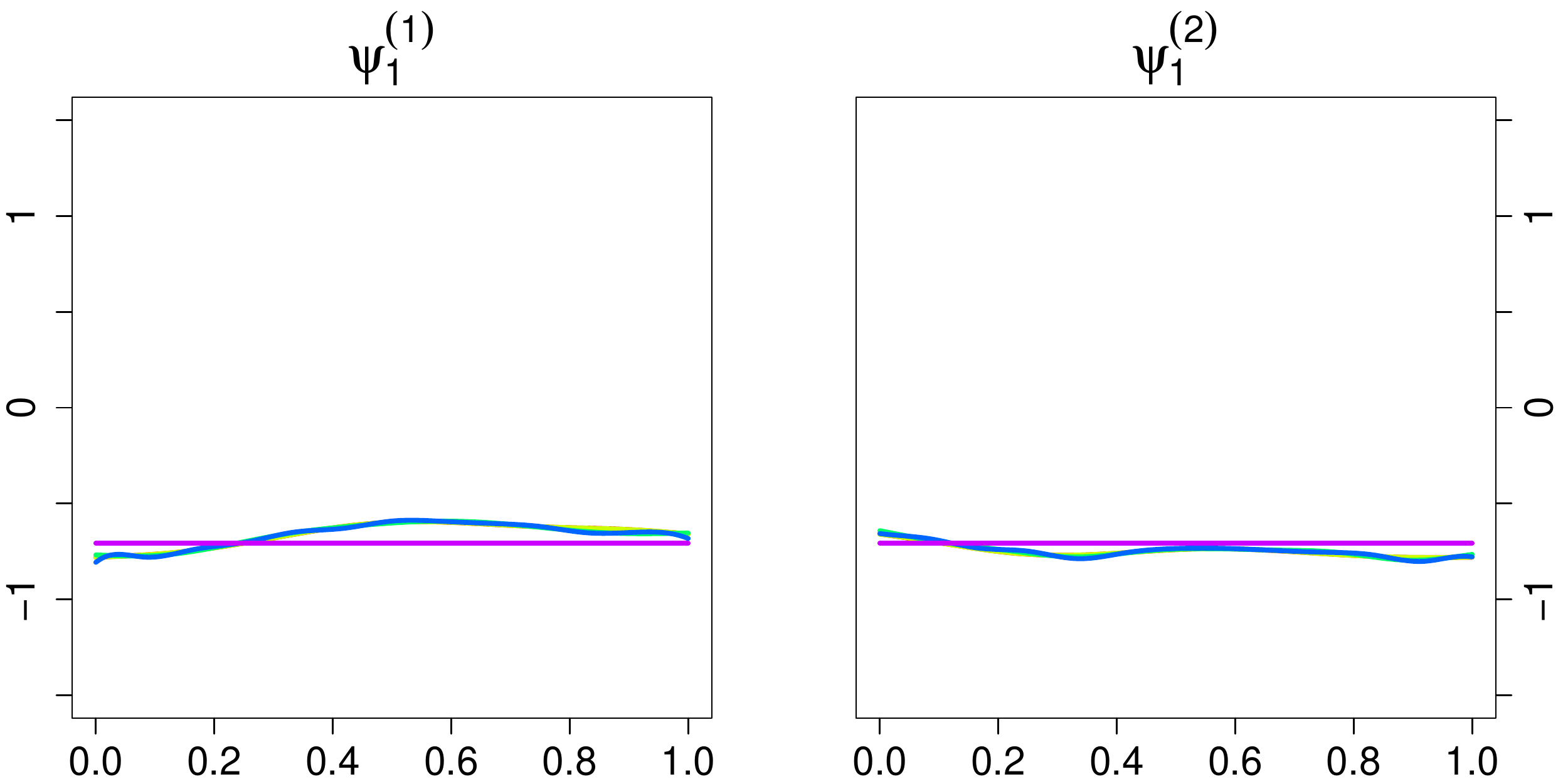}
\hfill
\includegraphics[height = 3.5cm]{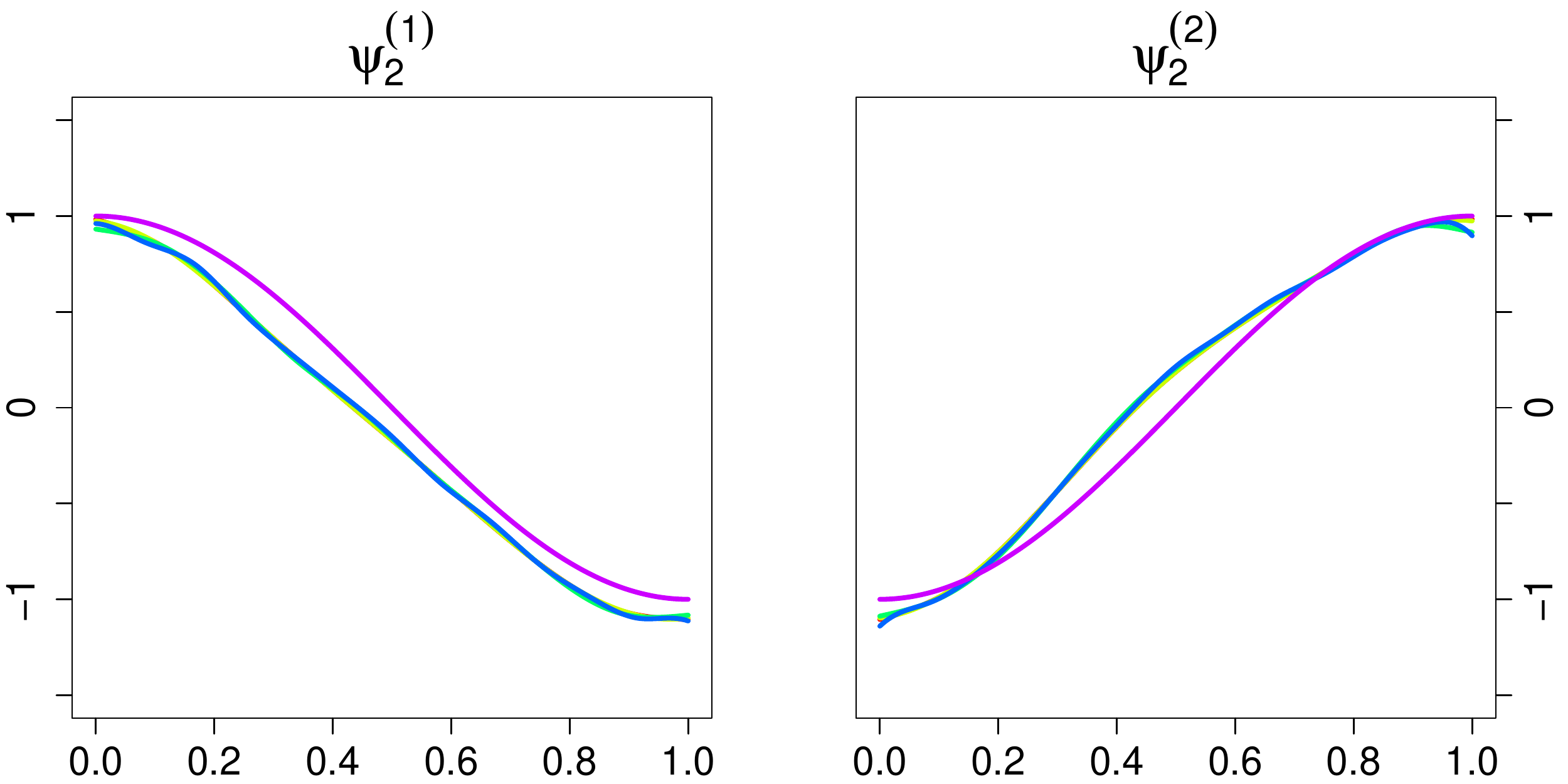}

\includegraphics[height = 3.5cm]{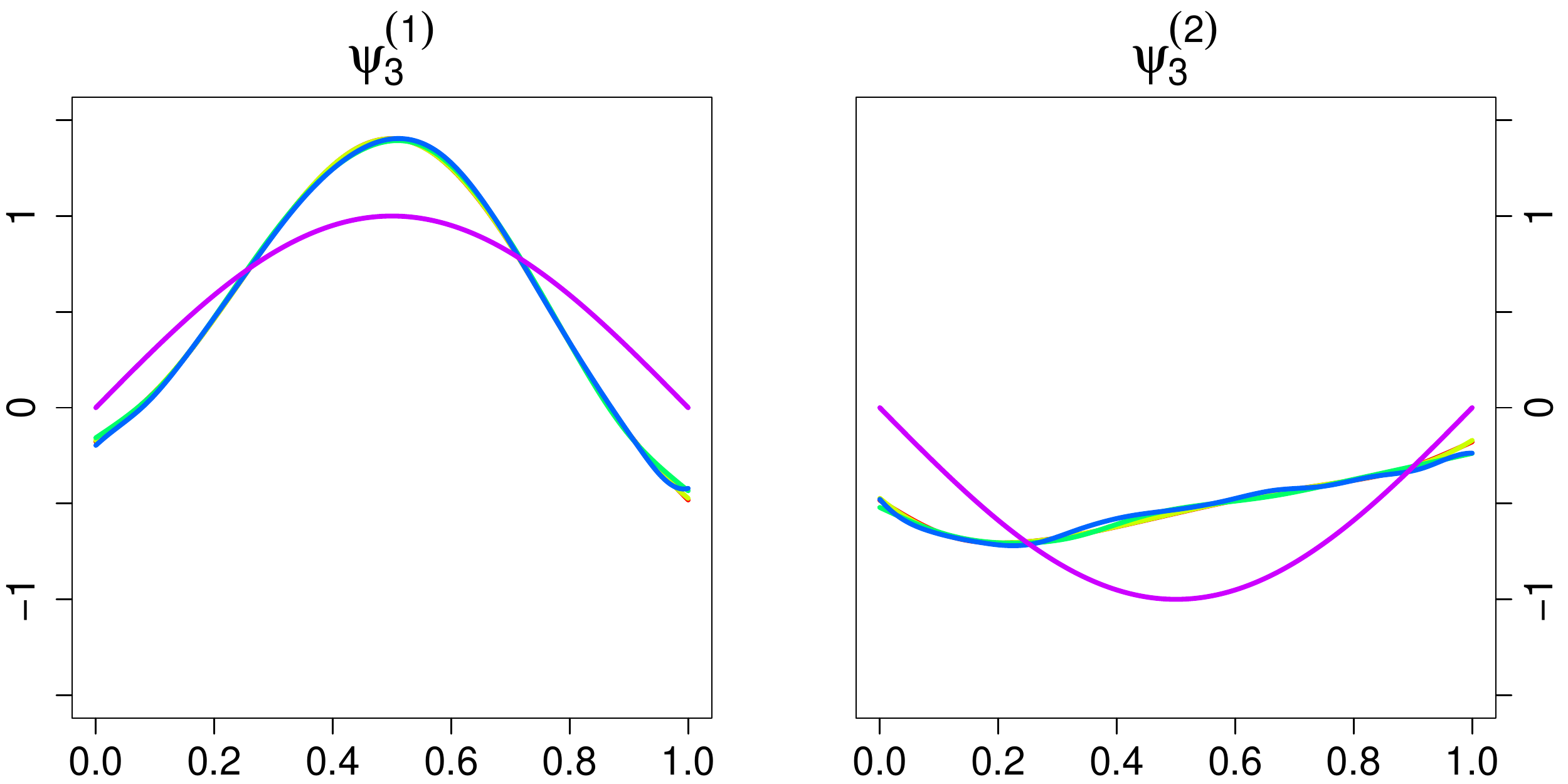}
\hfill
\includegraphics[height = 3.5cm]{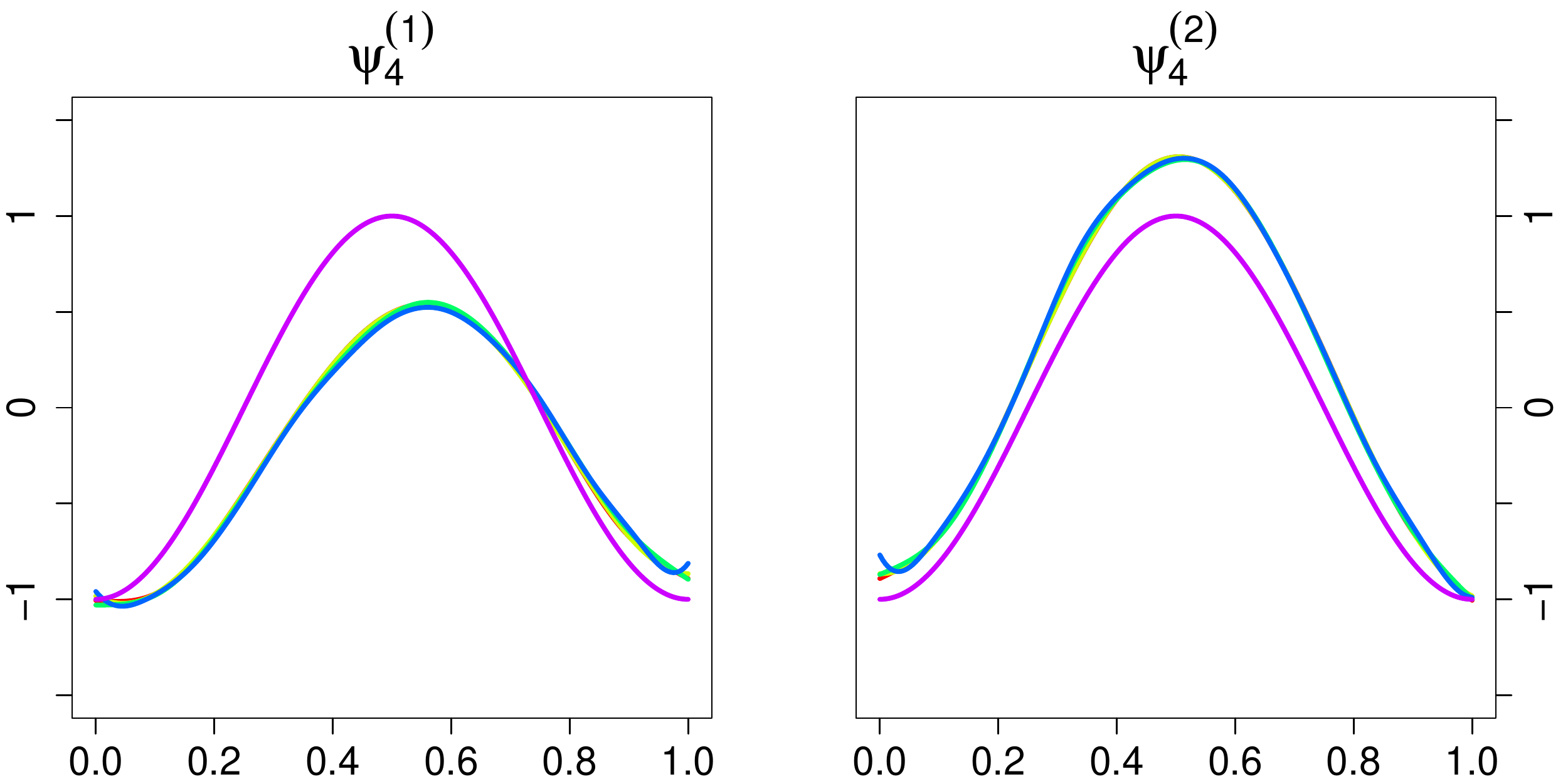}

\includegraphics[height = 3.5cm]{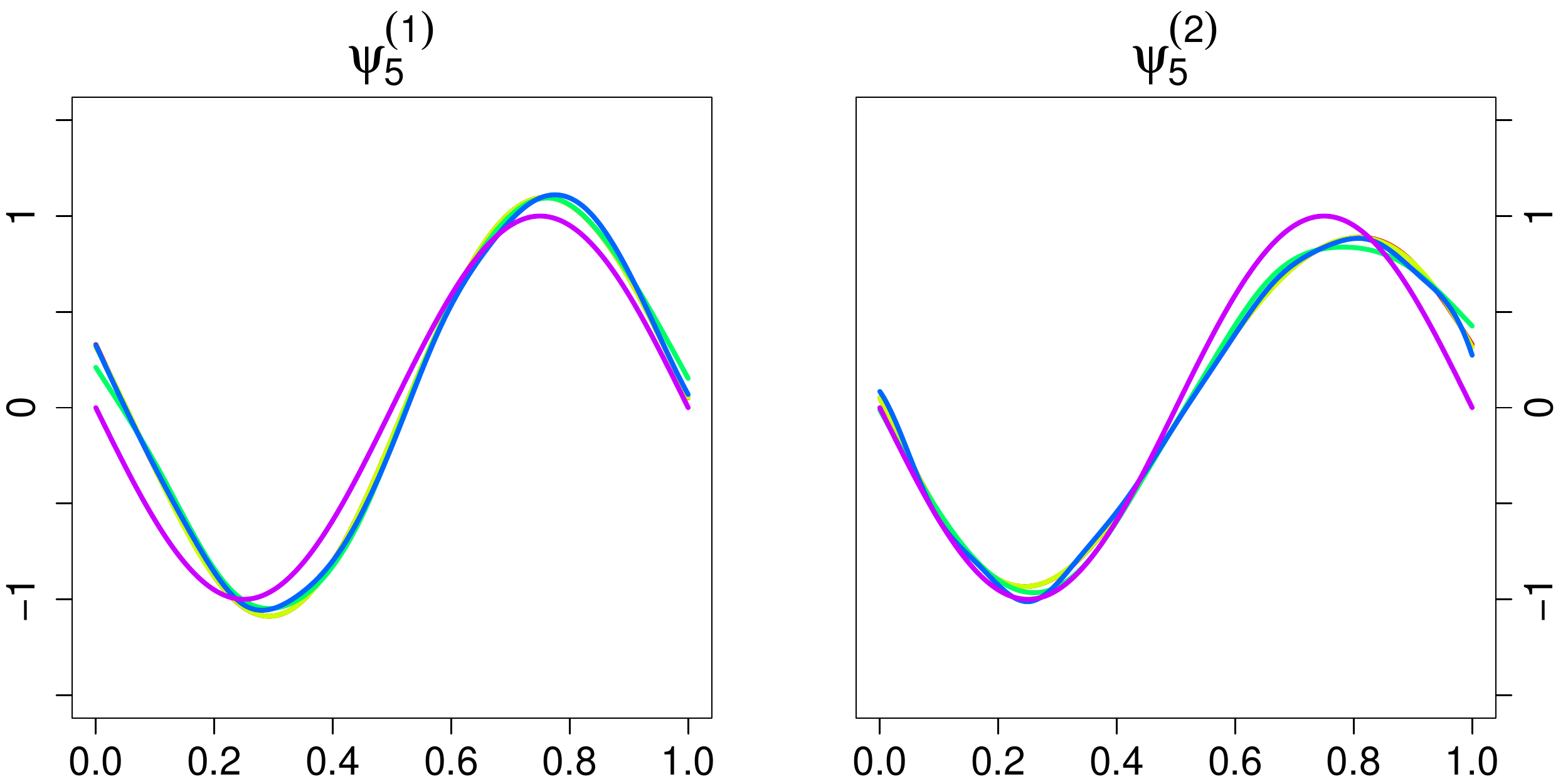}
\hfill
\includegraphics[height = 3.5cm]{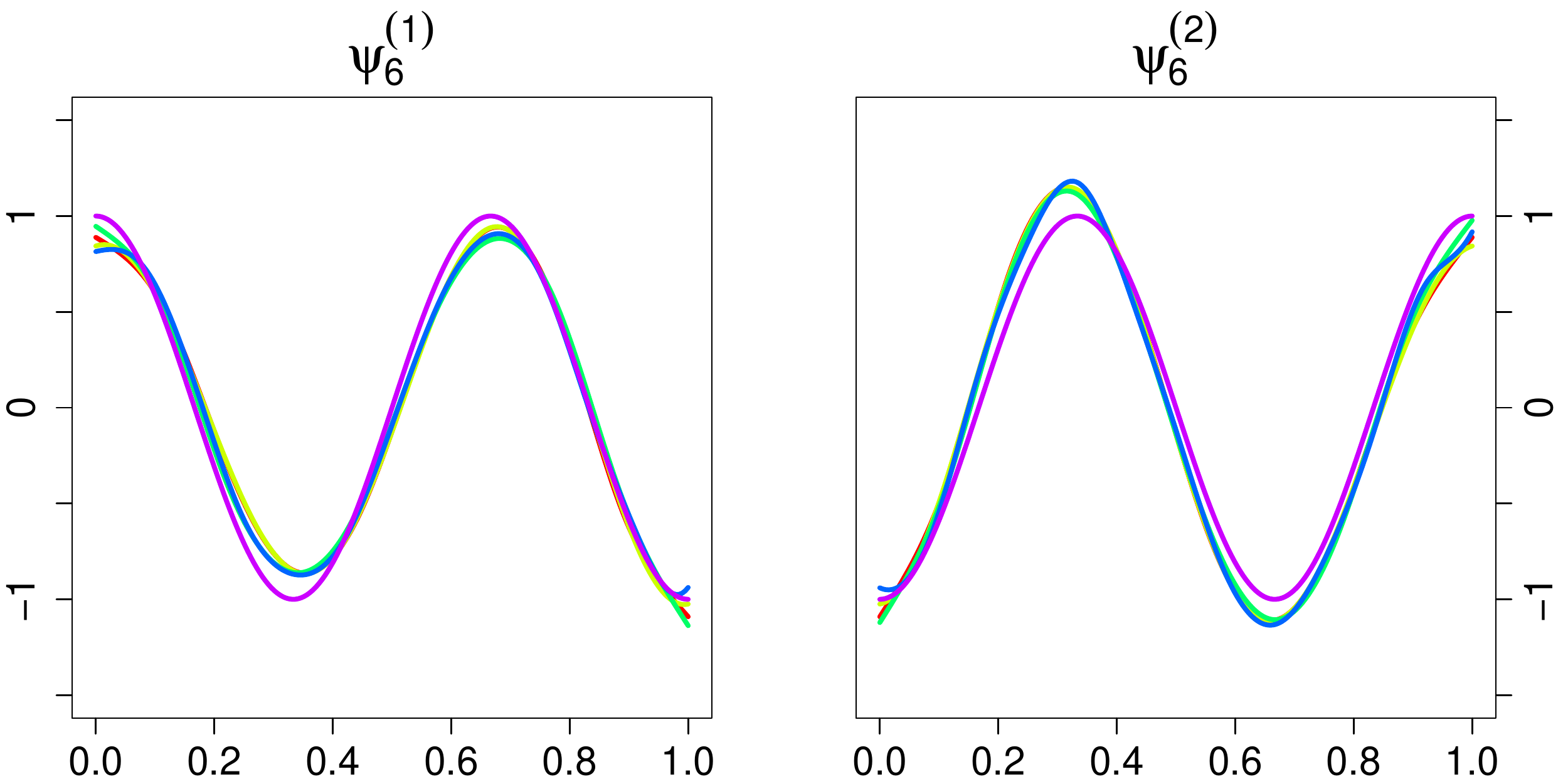}

\includegraphics[height = 3.5cm]{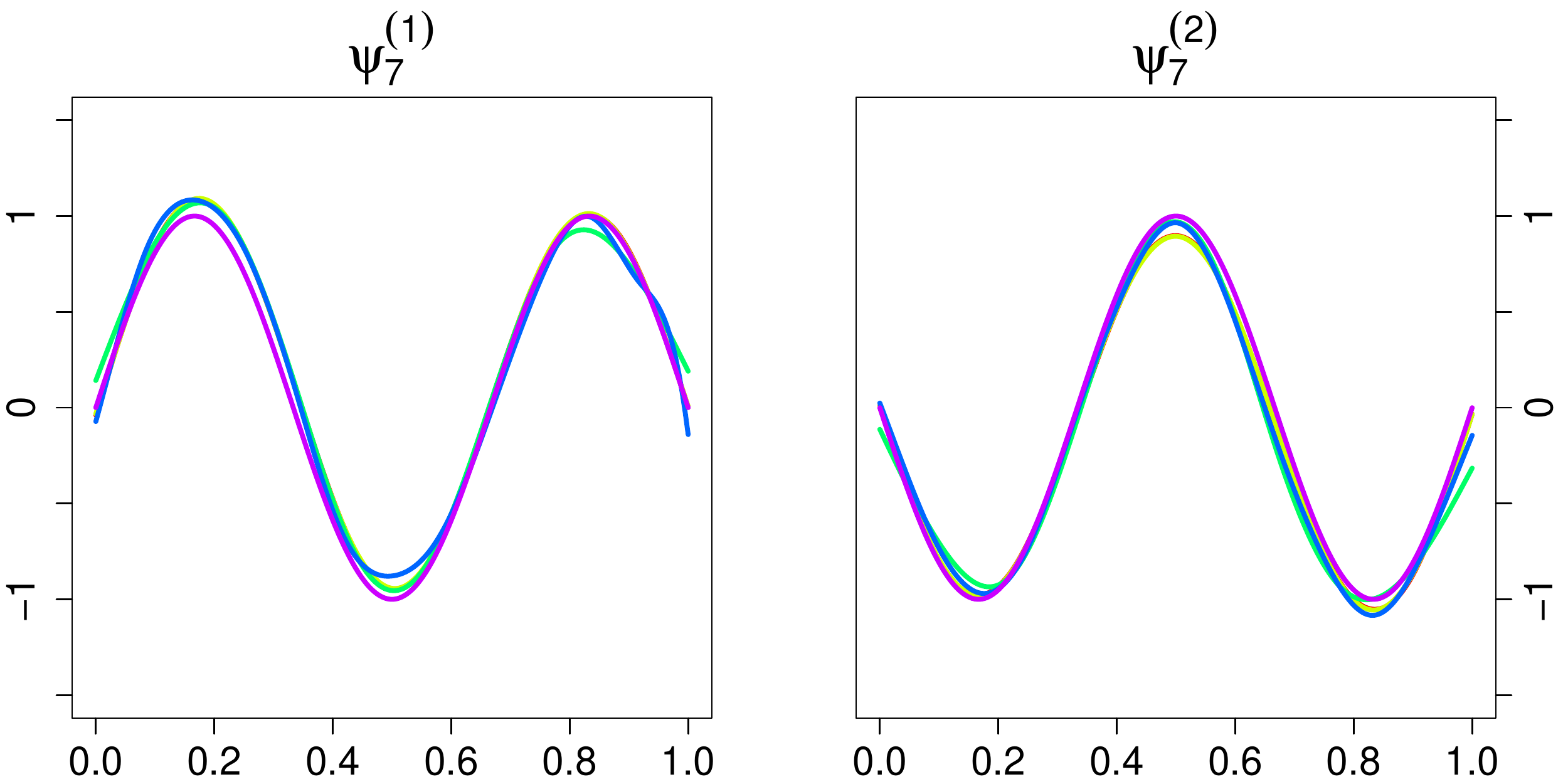}
\hfill
\includegraphics[height = 3.5cm]{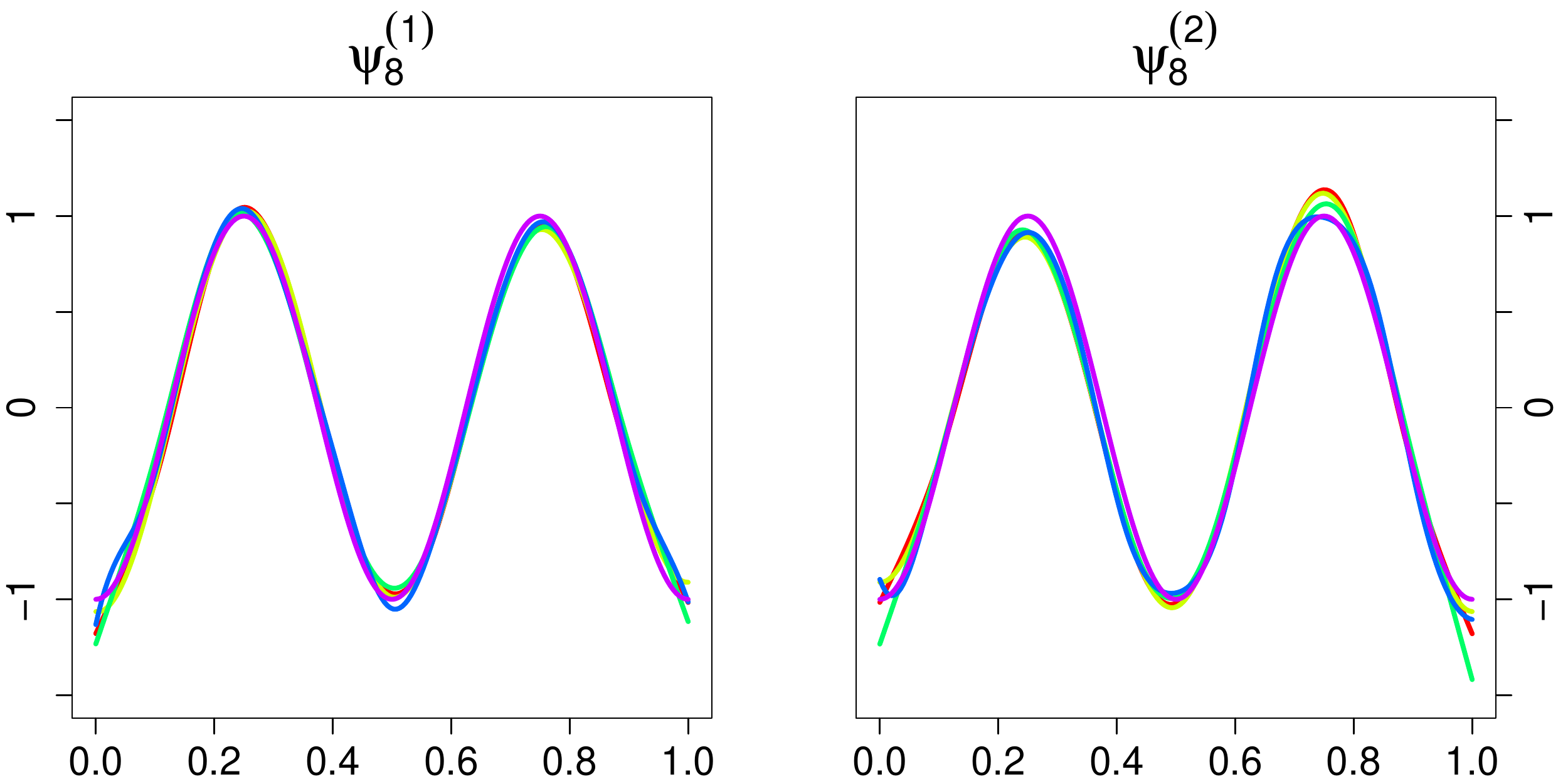}

\includegraphics[width = \textwidth]{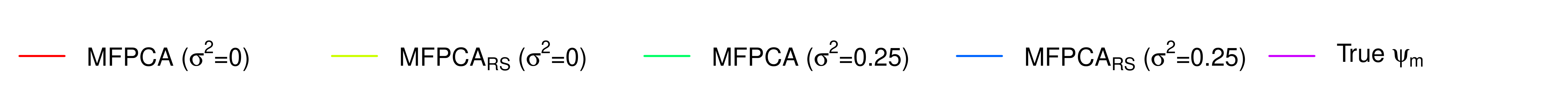}
\caption{True and estimated eigenfunctions for the first setting  based on one example replication with $N = 250$ observations. The results for the corresponding eigenfunctions are given in Table~\ref{tab:fitEVal}.}
\label{fig:fitEFun} 
\end{figure}

\newpage

\begin{sidewaysfigure}[ht]
\centering
\includegraphics[width = \textwidth]{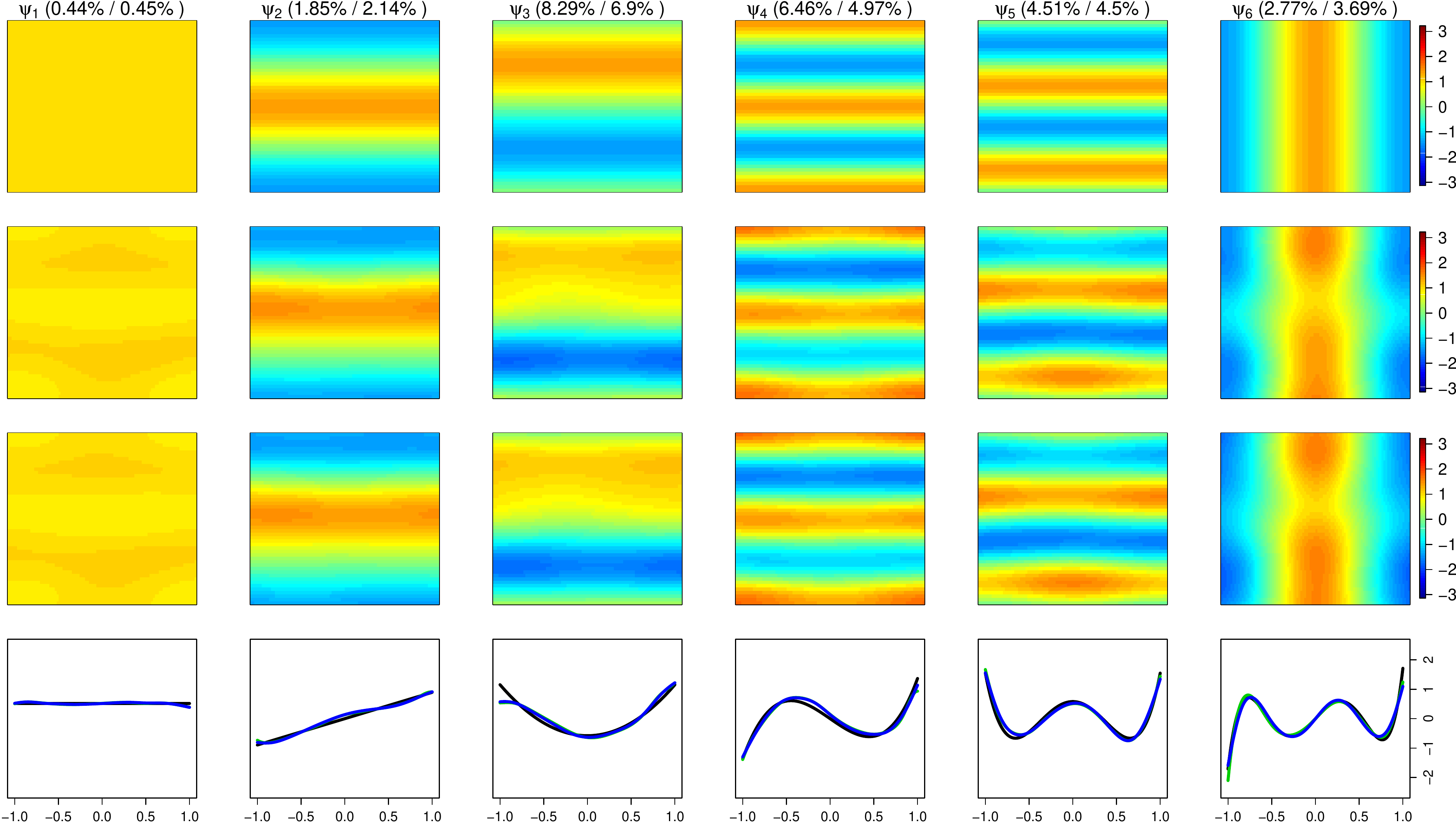}

\caption{Exemplary result for one replication in simulation setting 3 and MFPCA based on univariate spline expansions (results for the eigenfunctions $\psi_1, \ldots, \psi_6$). The first row shows the true first elements $\psi_m^{(1)}$,  the second/third row gives the results of $\hat \psi_m^{(1)}$ for data without/with measurement error. In the fourth row, the second elements $\psi_m^{(2)}$ of the true eigenfunctions are shown in black and the corresponding estimates for data without/with measurement error are shown in green/blue. Percentages in the titles give the relative errors $\operatorname{Err}(\hat \psi_m)$ for the estimates based on data without/with measurement error.  The reconstruction error $\operatorname{MRSE}$ is $0.40\% / 2.25\%$ for data without/with measurement error (simulation median: $0.40\% / 2.05\% $). Results for the eigenfunctions $\psi_7, \ldots, \psi_{12}$ are shown in Fig.~\ref{fig:multidimMFPCA_example2}.} 
\label{fig:multidimMFPCA_example1}
\end{sidewaysfigure}

\begin{sidewaysfigure}[ht]
\centering

\includegraphics[width = \textwidth]{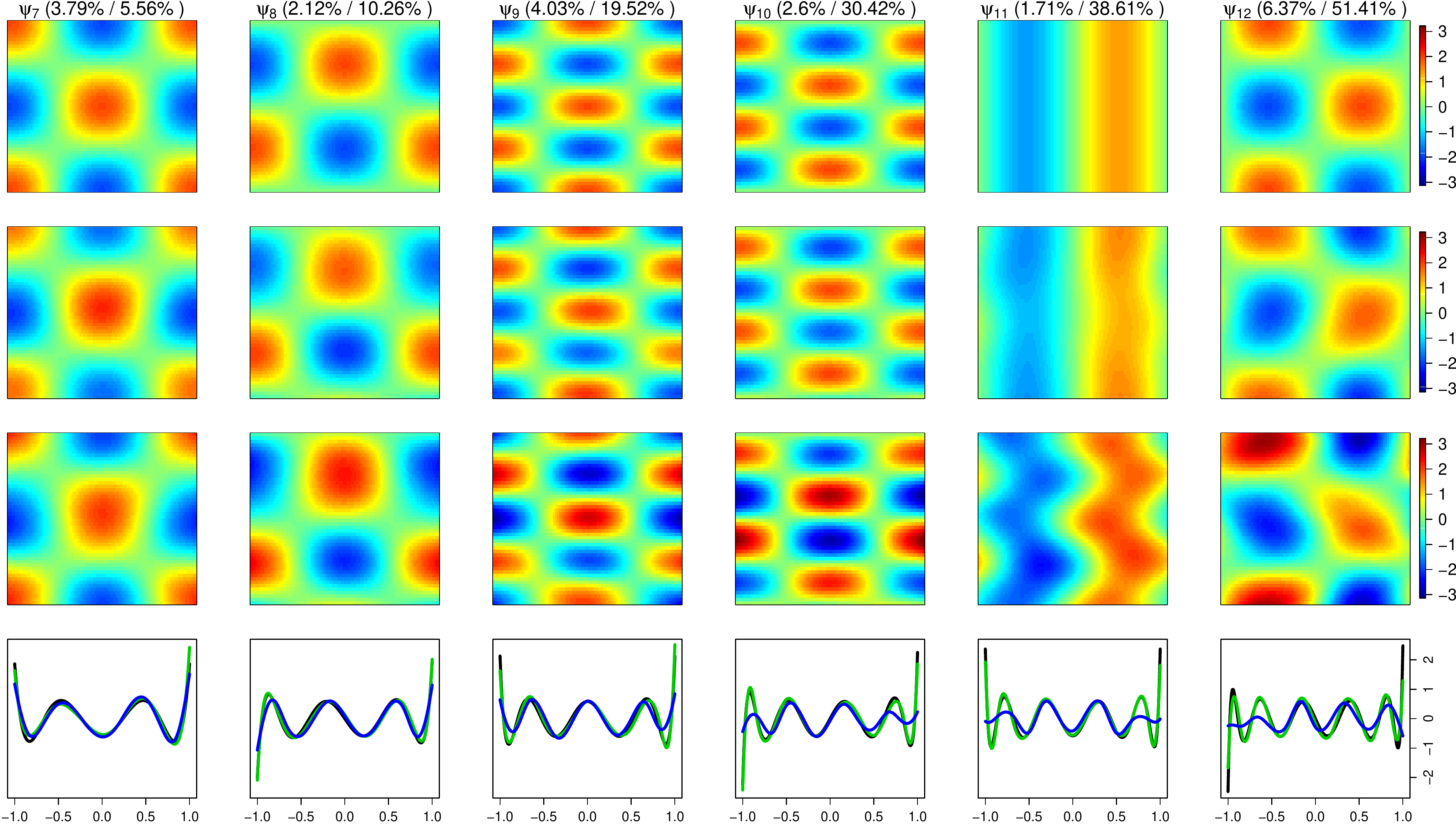}

\caption{Exemplary result for one replication in simulation setting 3 (results for the eigenfunctions $\psi_7, \ldots, \psi_{12}$). Please refer to Fig.~\ref{fig:multidimMFPCA_example1} for details.}
\label{fig:multidimMFPCA_example2}
\end{sidewaysfigure}

\FloatBarrier
\newpage

\subsection*{Sensitivity Analysis}
\label{sec:simSens}

As discussed in Section~\ref{sec:estMFPCA}, the number $M_j$ of univariate eigenfunctions used for MFPCA clearly has an impact on the results, as they control how much of the information in the univariate elements is used for calculating the multivariate FPCA. A standard approach in functional data analysis for quantifying the amount of information contributed by single eigenfunctions $\phi^{(j)}_m$ is the percentage of variance explained (pve), which is the ratio of the associated eigenvalue $\lambda^{(j)}_m$ and the sum of all eigenvalues. The following simulation systematically examines the sensitivity of the MFPCA result based on the pve of the univariate eigenfunctions.

\textbf{Simulation Setup: }
The simulation is based on $100$ replications with $N = 250$ observations of bivariate data on the unit interval (cf. setting 1 in Section~\ref{sec:simFun}), with $M = 8$ Fourier basis functions and exponentially decreasing eigenvalues for simulating the data. The number of univariate eigenfunctions $M_1, M_2$ for MFPCA is chosen based on $\text{pve} \in \{0.75, 0.90, 0.95, 0.99\}$ for both elements and $M_1 = M_2 = M = 8$ for comparison. The number of multivariate principal component functions is then set to $\min\{M_1 + M_2, M\}$.

\textbf{Results: }
The results of the sensitivity analysis are shown in Fig.~\ref{fig:simSenseValeFun} and  Table~\ref{tab:simSensrMSE}. The number of estimated multivariate eigenvalues/eigenfunctions is for all $100$ datasets  $\hat M = 4$ for $ \text{pve} = 0.75$, $\hat M = 6$ for $\text{pve} = 0.90$ and $\hat M = 8$ in all other cases. 
The results are as expected: Increasing the pve, and hence the information in the univariate FPCA, improves the estimation accuracy for both, multivariate eigenvalues and eigenfunctions. As a consequence, the reconstruction error reduces with increasing pve. Moreover, for a fixed $m$, the results show that there is a critical amount of information in univariate FPCA that is needed to describe the multivariate eigenvalues and eigenfunctions well. If this is reached (e.g. $\text{pve} = 0.95$ for $m = 5$, cf. Fig.~\ref{fig:simSenseValeFun}), the additional benefit of using more univariate eigenfunctions ($\text{pve} > 0.95$) becomes negligible. If, in contrast, the univariate FPCA does not contain enough information ($\text{pve} < 0.95$), the error rates for the MFPCA estimates are considerably increased.
For fixed pve, the error rates rise abruptly for the last pair of eigenfunctions ($m \in \{\hat M_+ -1, \hat M_+\}$). This is due to the fact that in this simulation, the multivariate functional principal components are derived from a Fourier basis. The last two eigenfunctions are hence sine and  cosine functions with highest frequency and cannot be represented well by the univariate functions used, as they contain only functions with lower frequency, in other words, they do not contain enough information.

\begin{figure}[hb!]
\centering
\includegraphics[width = \textwidth]{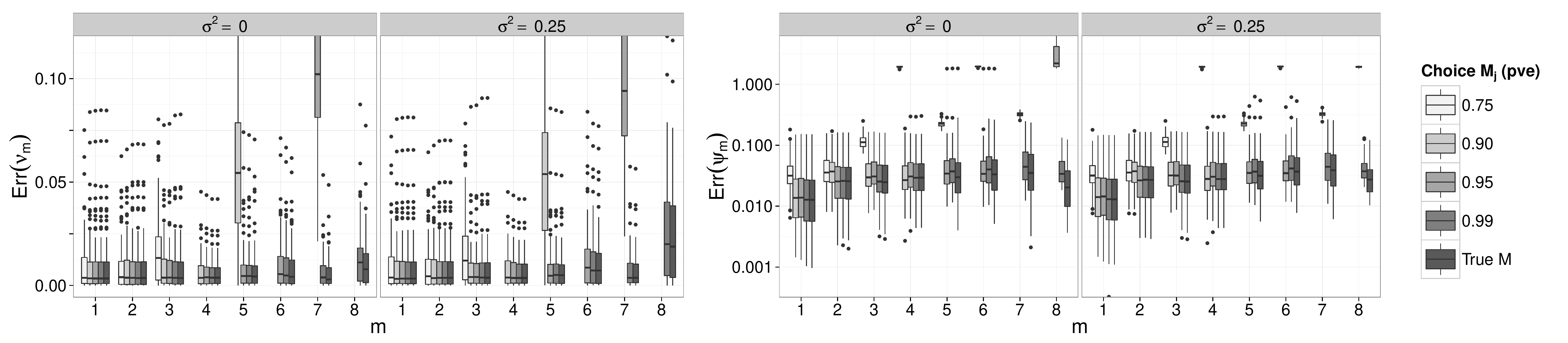}
\caption{Relative errors for estimated eigenvalues (left) and eigenfunctions (right, log-scale) for the sensitivity analysis. Extreme values cut off for better comparability.}
\label{fig:simSenseValeFun}
\end{figure}

\begin{table}[t]
\centering
\caption{
Average $\operatorname{MRSE}$ (in $\%$) in the sensitivity analysis.
}
\begin{tabular}{lccccc}
\hline \hline
& \multicolumn{4}{c}{Choice of $M_j$ (pve)} & True $M$ \\
& $ 0.75$ & $0.90$  & $0.95$ & $0.99$ &   \\ 
\hline 
$\sigma^2 = 0$  & 23.756 &   9.075 &   2.924 &   0.165 &  0.006 \\ 
$\sigma^2 = 0.25$ & 24.099 &   9.583 &   3.593 &   0.842 &  0.740 \\  
 \hline
\end{tabular}
\label{tab:simSensrMSE}
\end{table}

\FloatBarrier
\newpage

\subsection*{Coverage Analysis of Pointwise Bootstrap Confidence Bands}

In Section~\ref{sec:applicateADNI}, pointwise bootstrap confidence bands were calculated for the multivariate functional principal components estimated from the ADNI data to quantify the variability in the estimates. The following simulation study examines the coverage properties of such confidence bands.

\textbf{Simulation Setup:} The data generating process is the same as in the simulation in Section~\ref{sec:simImage}, mimicking the ADNI data that consists of functions on a one-dimensional domain and images.
In total, the simulation is based on 100 datasets, all having $N = 250$ observations. Each dataset is considered with and without measurement error. Both elements are represented in terms of B-spline basis functions with appropriate smoothness penalties in the presence of measurement error (cf. Section~\ref{sec:simImage}).
For each dataset and each estimated eigenfunction, a pointwise $95\%$ bootstrap confidence band is calculated based on $100$ bootstrap samples on the level of subjects (cf. Section~\ref{sec:applicateADNI}). The coefficients of the spline basis decompositions can efficiently be reused when bootstrapping, as the basis is fixed and does not depend on the bootstrap sample. In contrast, the univariate functional principal components for the ADAS-Cog trajectories in the ADNI application have to be re-estimated for each bootstrap sample. This computational aspect is taken into account in the bootstrap implementation in the \texttt{MFPCA} package{ \if1\blind \citep{MFPCA}\fi}. 
Finally, the confidence bands are calculated separately for each element as pointwise percentile bootstrap confidence intervals.
\\
For each eigenfunction and each observation point, the estimated coverage at one point $t_j \in \mathcal{T}_j$ is the percentage of datasets for which the true eigenfunction $\psi_m^{(j)}$  evaluated at $t_j$ is enclosed in the bootstrap confidence band (up to a sign change of the whole function). Fig.~\ref{fig:simBootstrap} shows the estimated coverages of the elements of the eigenfunctions for data with and without measurement error aggregated over the observation points. 

\textbf{Results:}
If the data is observed without measurement error, the pointwise confidence bands enclose the true functions fairly precisely in $95\%$ of all cases with very little variation between the observation points. 
For the leading eigenfunctions, the same holds true if the data is observed with measurement error. For higher order eigenfunctions, that explain hardly any variation in the data, the estimated coverage decreases, especially for the second element ($\psi_m^{(2)}$, one-dimensional domain) and shows a much higher variation between the observation points. On the one hand this may be caused by the fact that the true eigenfunctions $\psi_m^{(2)}$ have a stronger curvature for growing $m$ (cf. Fig.~\ref{fig:multidimMFPCA_example2}). Severe undercoverage for higher-order eigenfunctions occurs mainly in regions of high curvature and slope of the eigenfunctions, where the low signal-to-noise level leads to oversmoothing (cf. Fig.~\ref{fig:simBootstrapEx}).
 On the other hand, the results of Section~\ref{sec:simImage} show that the estimates for higher order eigenfunction elements become more inaccurate due to interchanging of eigenfunctions, hence the bootstrap confidence bands can be centered incorrectly.
 For the image elements $\psi_m^{(1)}$, the bootstrapped confidence bands give much better results, except for some outliers that form spatially smooth outlying regions (see e.g. Fig.~\ref{fig:simBootstrapEx}). This reflects that the pointwise coverages are not independent, as the true eigenfunctions as well as the confidence bands are smooth: If the function $\psi_m^{(j)}$ lies within the bootstrap confidence band at a point $t_j$, it is very likely that it will also be inside the confidence band at the neighbouring observation points (analogously for points outside the CI). 
This relation is highlighted in Fig.~\ref{fig:simBootstrapEx}, which  illustrates the coverage rates for $\psi_3$ (having a good coverage) and $\psi_9$ (having a rather poor coverage) in the case of measurement error. 
In summary, the results of the simulation show that the bootstrapped confidence bands give reliable results, in particular  for the leading eigenfunctions that explain most of the variation in the data. Moreover, smooth eigenfunctions will have a stabilizing effect for the coverage. However, when interpreting such pointwise confidence bands, one should keep in mind the dependence across neighbouring observation points due to the smoothness of the eigenfunctions.

\begin{figure}[ht]
\centering
\includegraphics[width = 0.85 \textwidth]{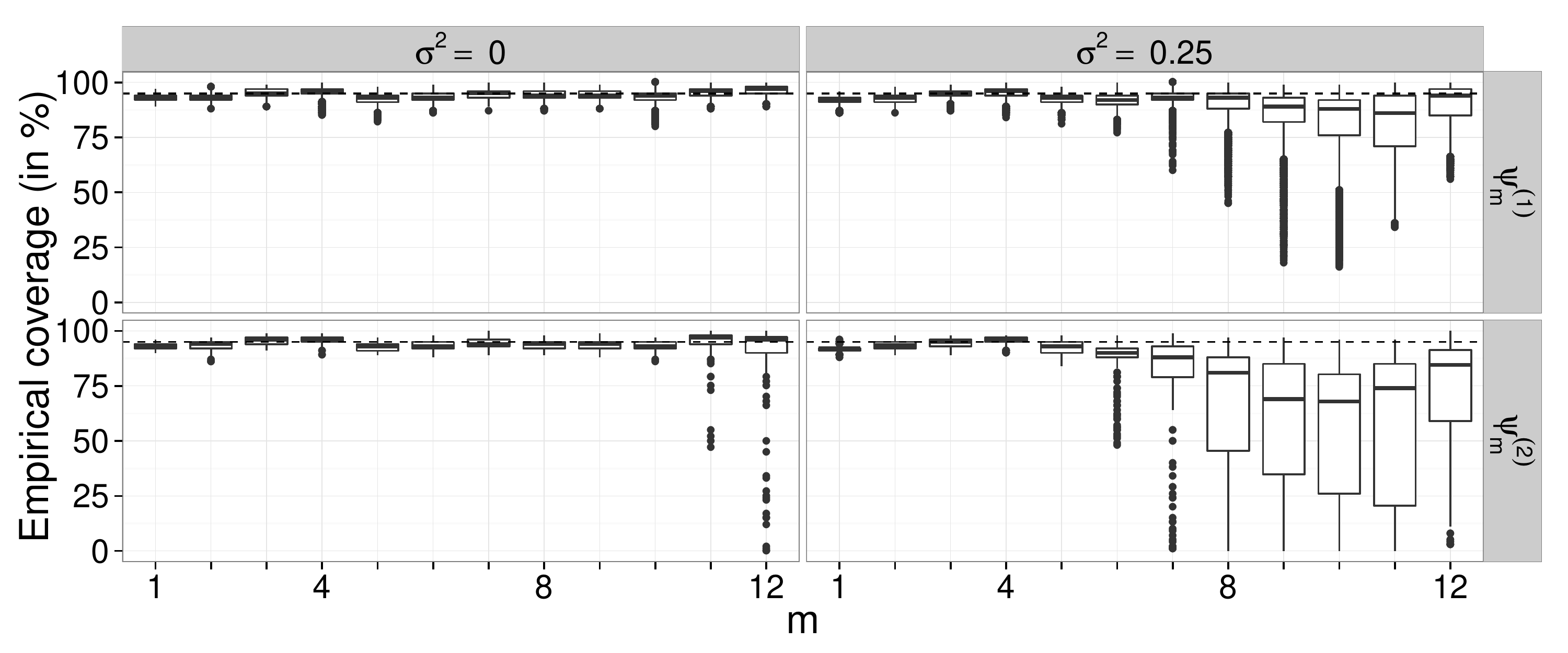}

\caption{Empirical coverages from the bootstrap simulation study for data without ($\sigma^2 = 0$) and with ($\sigma^2 = 0.25$) measurement error. The boxplots show the pointwise coverage of the bootstrap confidence bands aggregated over the corresponding domains for both elements of the true eigenfunctions $\psi_m,~ m = 1, \ldots, 12$ (1st row: Image element $\psi_m^{(1)}$, 2nd row: Element $\psi_m^{(2)}$ with one-dimensional domain). The dashed line marks a coverage of $95 \%$.}
\label{fig:simBootstrap} 
\end{figure}

\begin{figure}[ht]
\centering
\includegraphics[height = 3.5cm]{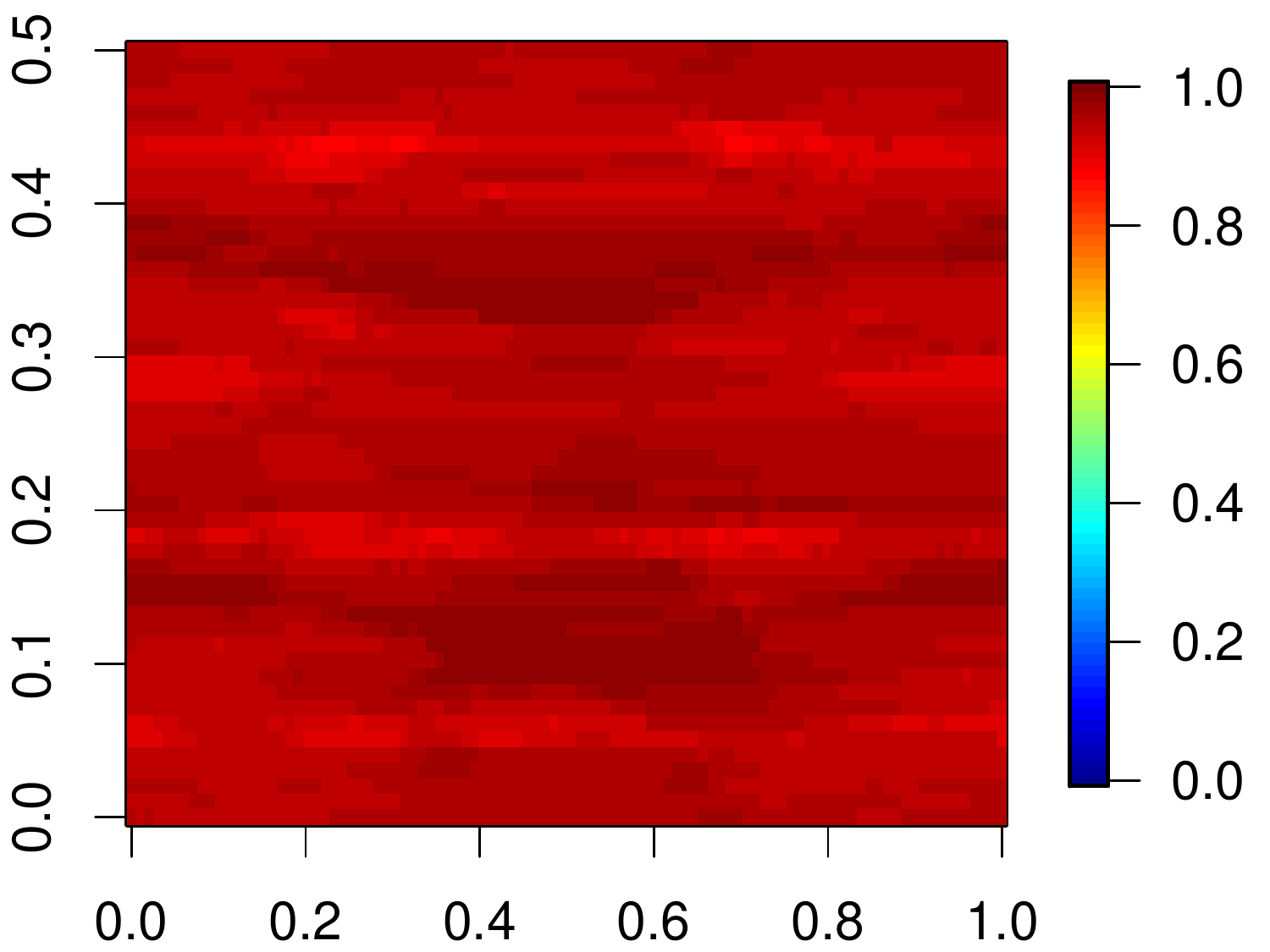}
%\hfill
\includegraphics[height = 3.5cm]{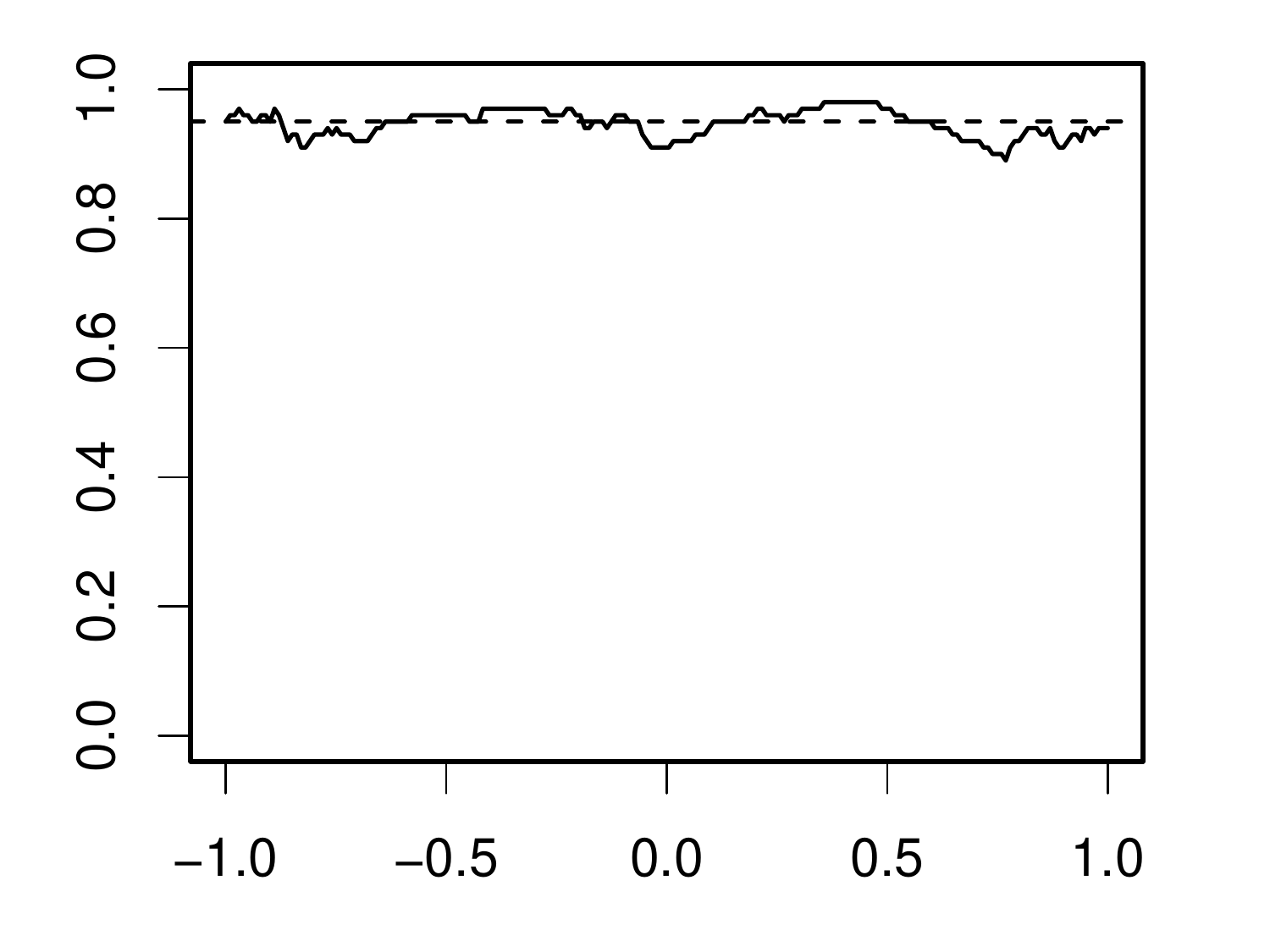}
%\hfill

\includegraphics[height = 3.5cm]{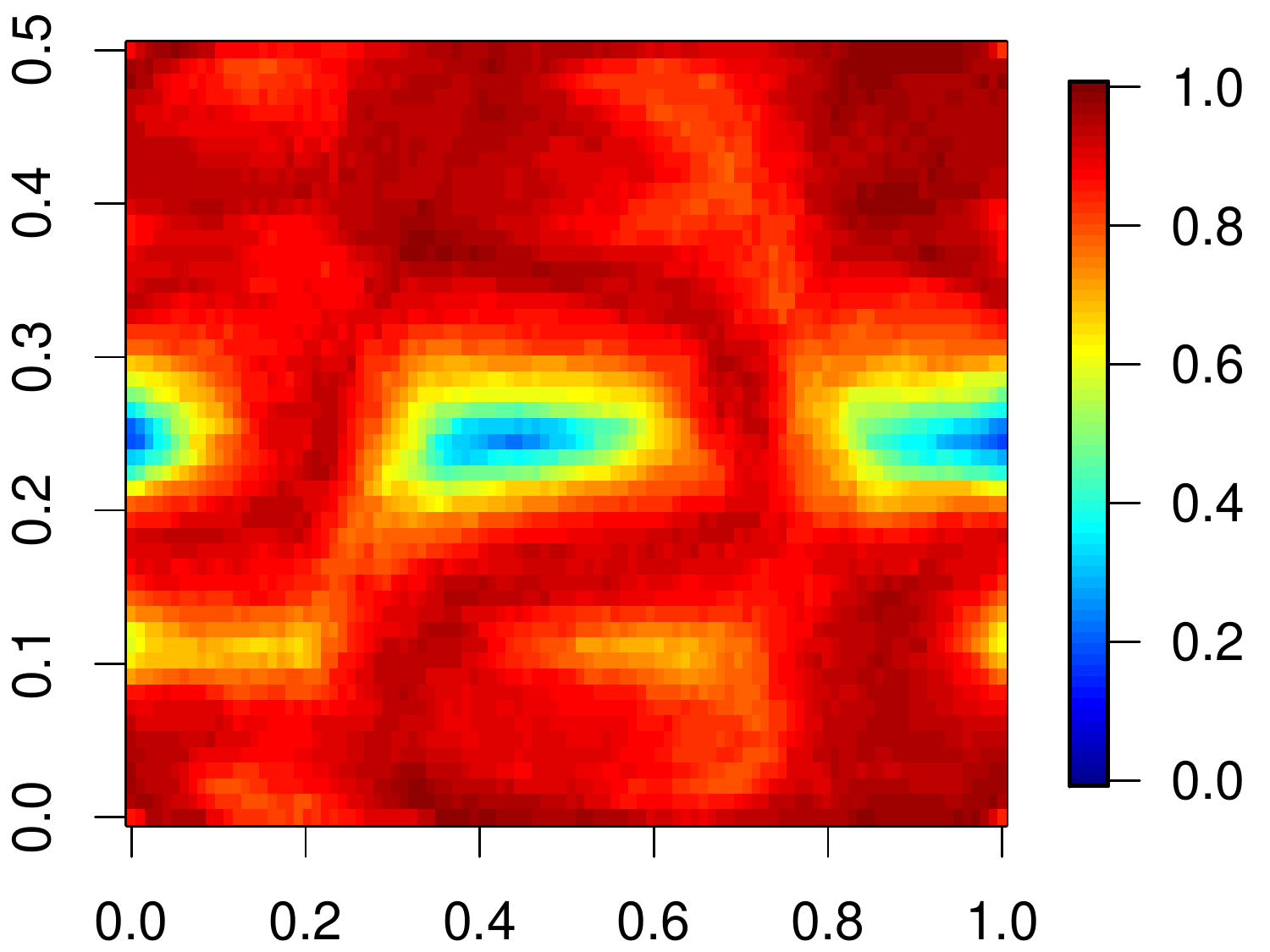}
% \hfill
\includegraphics[height = 3.5cm]{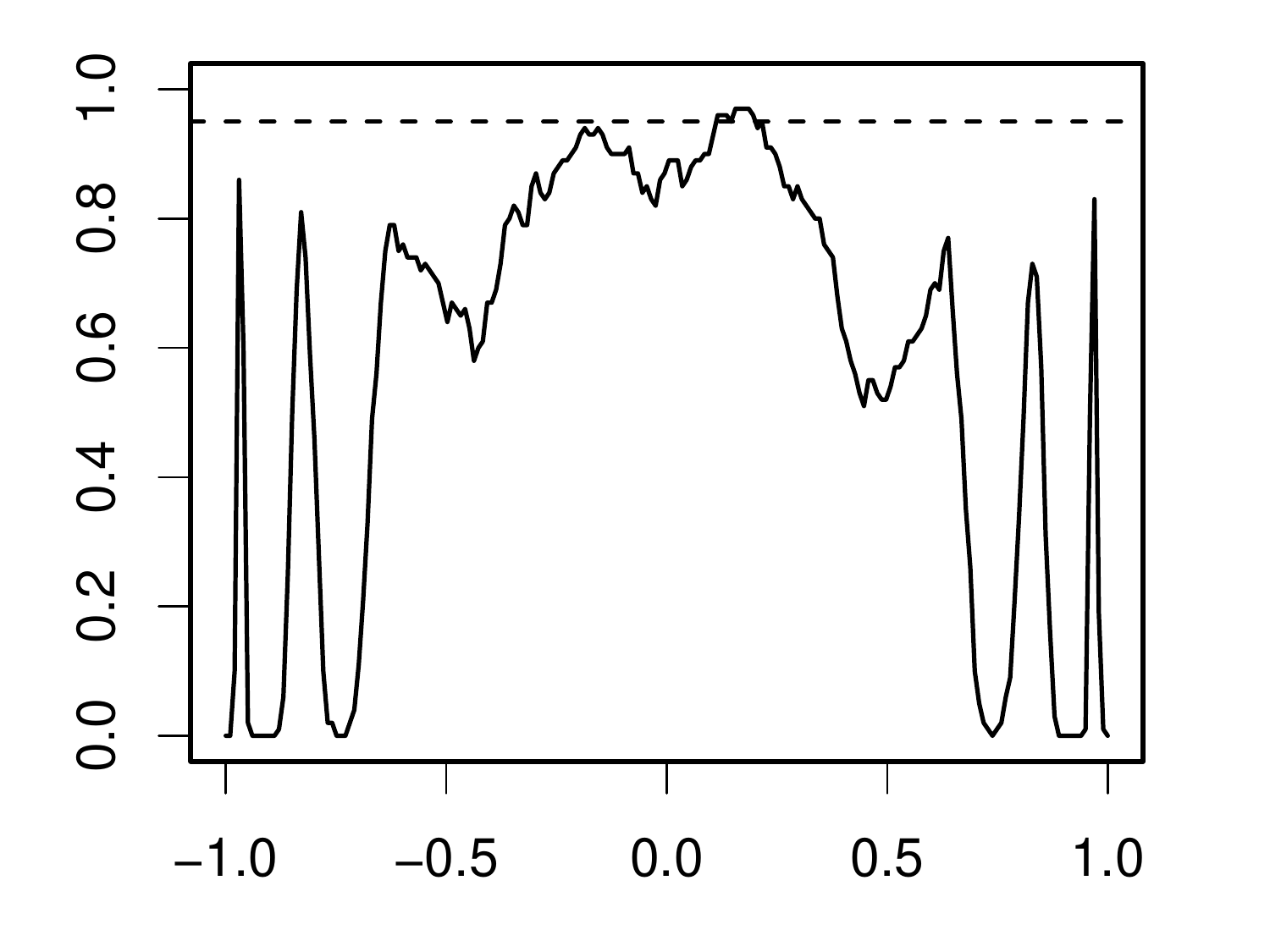}

\caption{Exemplary results from the bootstrap simulation study for data observed with measurement error. The first row shows the estimated coverages for the third eigenfunction,  the second row shows the estimated coverages for the eigenfunction of order $9$ (see also Fig.~\ref{fig:simBootstrap}). The first column corresponds to the estimated elements $\hat \psi_m^{(1)}$ and the second column corresponds to the estimated elements $\hat \psi_m^{(2)}$. For the latter, the dashed lines correspond to the nominal level of $95\%$.}
\label{fig:simBootstrapEx} 
\end{figure}

\FloatBarrier
\newpage

\section*{Applications -- Gait Cycle Data}
\label{sec:applicateGait}

For comparison to an existing method in the special case of densely sampled bivariate data on the same one-dimensional interval, the new MFPCA approach is applied to the gait cycle data (cf. Fig.~\ref{fig:motivationMFPCA} in the main document) and compared to the method of \cite{RamsaySilverman:2005} as implemented in the \texttt{R}-package \texttt{fda} \citep{fda}. The results are shown in Fig.~\ref{fig:gait}. For the new approach, the multivariate principal components are calculated based on univariate FPCA with $M_1 = M_2 = 5$ principal components. For $\text{MFPCA}_\text{RS}$, the observed functions are pre-smoothed using $K = 15$ cubic spline basis functions as in the simulation study (cf. Section~\ref{sec:simFun}). As for synthetic data, the two methods give nearly identical results.

\begin{figure}[ht]
\centering
\includegraphics[height = 3.5cm]{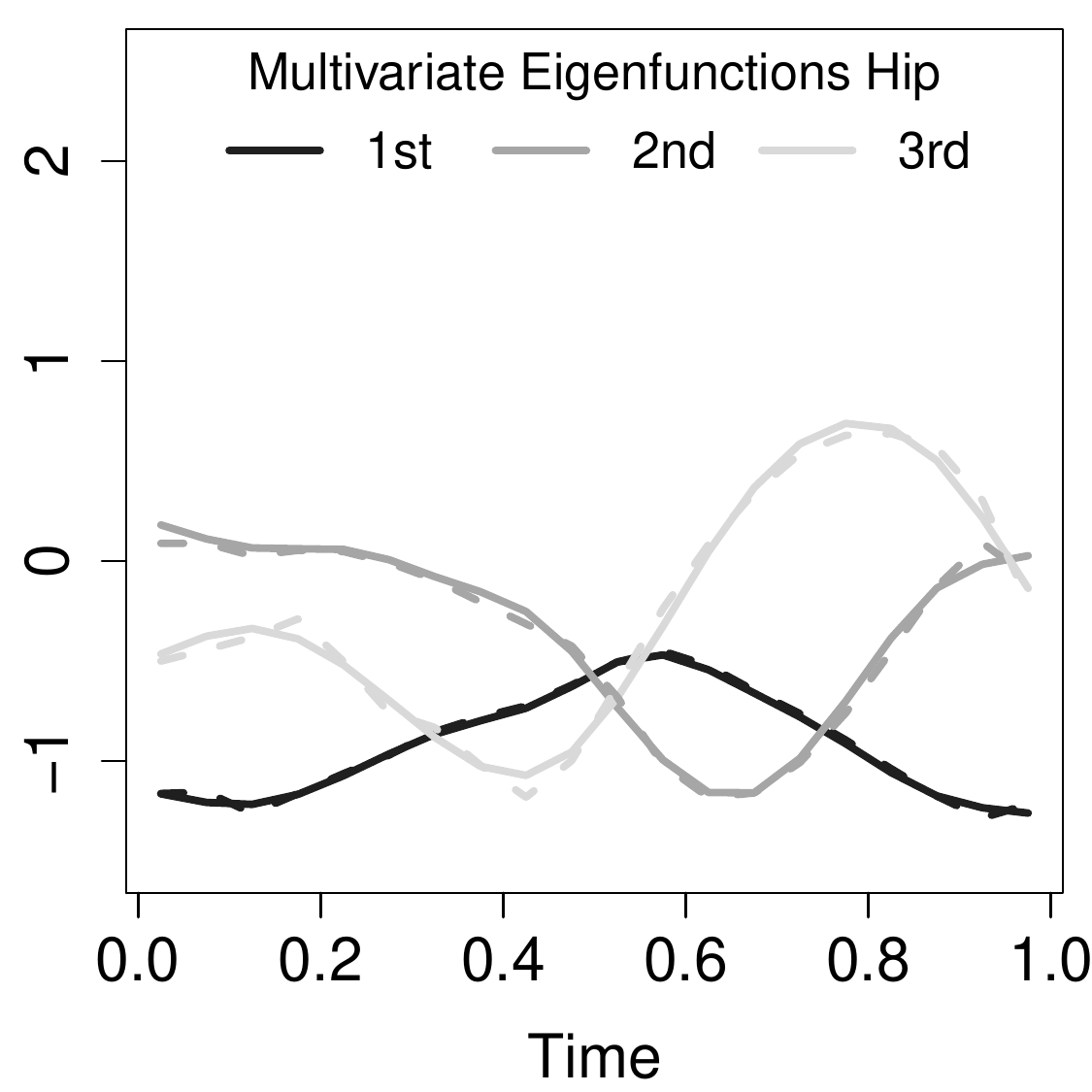}
\includegraphics[height = 3.5cm]{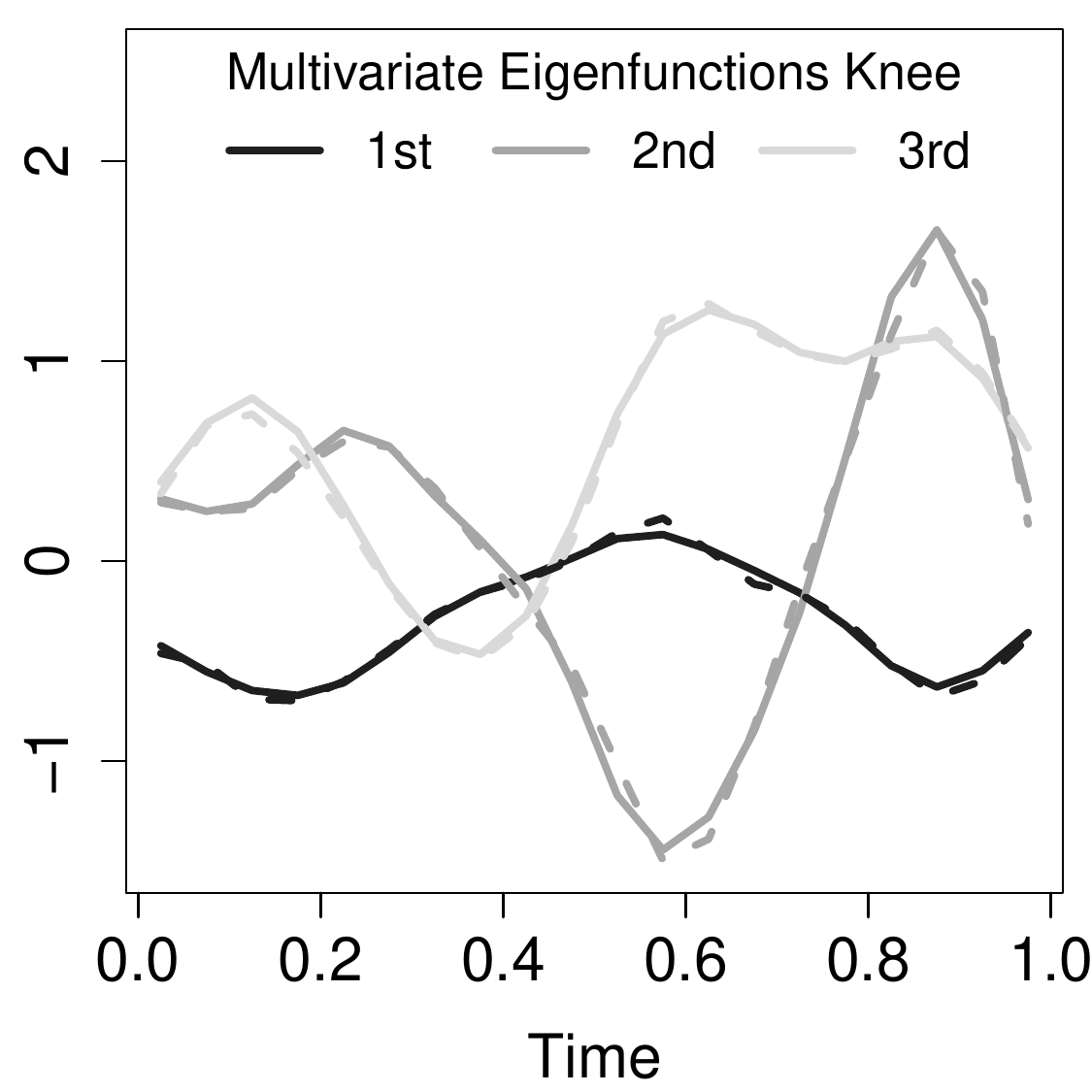}
\caption{The first three estimated bivariate eigenfunctions for the gait data set. Solid lines show the results of the new MFPCA approach, dashed lines correspond to the approach of \cite{RamsaySilverman:2005}. The functions have been reflected, if necessary, for comparison purposes.}
\label{fig:gait}
\end{figure}

\end{document}